\documentclass[onefignum,onetabnum]{siamonline250211}

\usepackage{lipsum}
\usepackage{amsfonts}
\usepackage{graphicx}
\usepackage{epstopdf}
\usepackage{algorithmic}
\Crefname{ALC@unique}{Line}{Lines} 
\ifpdf
  \DeclareGraphicsExtensions{.eps,.pdf,.png,.jpg}
\else
  \DeclareGraphicsExtensions{.eps}
\fi

\usepackage{amssymb,amsmath, enumerate,ulem,color,graphics}

\usepackage{enumitem}
\setlist[enumerate]{leftmargin=.5in}
\setlist[itemize]{leftmargin=.5in}

\newsiamremark{remark}{Remark}
\newsiamremark{hypothesis}{Hypothesis}
\newsiamremark{assumption}{Assumption}
\newsiamremark{example}{Example}
\crefname{hypothesis}{Hypothesis}{Hypotheses}
\newsiamthm{claim}{Claim}
\newsiamremark{fact}{Fact}
\crefname{fact}{Fact}{Facts}

\headers{Nonparametric Bayesian Calibration}{H. Shi, L. Yang, J. Chi, D. Bingham, T. Butler, D. Estep and H. Wang}

\title{Nonparametric Bayesian Calibration of Computer Models\thanks{
\funding{H.~Shi acknowledges the support of the Natural Sciences and Engineering Research Council of Canada. L.~Yang acknowledges the support of the U.S.~Department of Energy and the U.S.~National Science Foundation. J.~Chi acknowledges the support of the U.S.~National Science Foundation. D.~Bingham acknowledges the support of the Natural Sciences and Engineering Research Council of Canada. T.~Butler's work is supported by the U.S.~National Science Foundation under Grant No.~DMS-2208460. D.~Estep acknowledges the support of  the Canada Research Chairs Program,  the Natural Sciences and Engineering Research Council of Canada, Canada's Digital Technology Superclustre, the Dynamics Research Corporation, the Idaho National Laboratory, the U.S.~Department of Energy, the U.S.~National Institutes of Health, the U.S.~National Science Foundation, and Riverside Research. H.~Wang's work is supported by the U.S.~National Science Foundation. Any opinion, finding, and conclusions or recommendations expressed in this material are those of the authors and do not necessarily reflect the views of supporting institutions. }}}

\author{Haiyi Shi\thanks{Department of Statistics and Actuarial Science, Simon Fraser University 
  (\email{haiyi\_shi@sfu.ca}).}
\and Lei Yang\thanks{Department of Statistics, Colorado State University 
  (\email{yangleicq@gmail.com}.}
\and Jiarui Chi\thanks{Biostatistics and Programming, Sanofi
	(\email{jiarui.chi21@alumni.colostate.edu}).}
\and Derek Bingham\thanks{Department of Statistics and Actuarial Science, Simon Fraser University 
	(\email{dbingham@sfu.ca}).}
\and Troy Butler\thanks{Department of Mathematics and Statistics, University of Colorado at Denver 
	(\email{troy.butler@ucdenver.edu}).}
	\and Don Estep\thanks{Department of Statistics and Actuarial Science, Simon Fraser University 
		(\email{destep@sfu.ca}).}
\and Haonan Wang\thanks{Department of Statistics, Colorado State University 
	(\email{haonan.wang@gmail.com}).}
}
\usepackage{amsopn}

\begin{document}
	
	\maketitle
	
	\begin{abstract}
		Combining field data and computer models is a crucial step for making inferences, predictions, and decisions for complex science and engineering systems. We formulate and analyze a nonparametric Bayesian methodology for calibrating the distribution of parameters in a computer model  using field observations.  Our results include establishing; a unique nonparametric Bayesian posterior corresponding to a chosen prior with an explicit formula for the posterior density; a maximum entropy property of the posterior  corresponding to the uniform prior; the almost everywhere continuity of the posterior density; and a comprehensive statistical analysis of an estimator based on importance sampling. They also include establishing the well-posedness of the nonparametric Bayesian solution of the calibration problem. We illustrate the results using several examples.
	\end{abstract}
	
	


\section{Introduction}\label{Introd}

Reliance on \textbf{\textit{computer models}}, e.g., a numerical solver for a differential equation, to explore and design physical systems is central to many areas of science and engineering \cite{Sacks1989}. Physical processes and laws governing the behaviour of a  system are encoded into the computer model while the behaviour of any specific example of a system depends on \textbf{\textit{(physical) parameters}} in the computer model that describe physical characteristics. For instance, many materials respond to heat stress in similar ways, but the response of a particular object depends on characteristics such as composition, porosity, and thermal diffusivity. 

In many real-world settings, the parameters governing system behavior cannot be measured directly. Moreover, parameter values vary during repeated trials of an experiment, leading to variability in field observations. This gives rise to the \textbf{\textit{calibration problem}} of making inferences about the possible variation in parameters by combining a computer model with field observations. For example, the heat equation describes how an object responds to a heat source. Sampling the temperature of objects from a collection during an experiment in which the objects are heated corresponds to evaluating the solution of the heat model for some unknown thermal conductivities. The goal of calibration is to determine the possible  thermal conductivity values from temperature data.  Generally, calibration of computer models plays an important role in many applications. Correspondingly, the calibration literature is enormous. We discuss work most closely related to our approach in Sec.~\ref{relwork}.

\subsection{Stochastic models for parameters and Bayesian calibration}\label{StoModComp}
It is often natural to use a probability model for parameters in a computer model. One common situation is a parameter that is modeled as stochastic because it varies naturally. For example, when conducting stress tests on metal objects for quality control during production, the material properties of the objects  varies from trial to trial. Another common situation is the case where a parameter represents ``up-scaled'' behavior of system or quantity that varies in a complex fashion, at a rapid time scale, and/or small spatial scale. In the heat equation, thermal conductivity represents an upscaled representation of the behavior of the molecules in an object. Such upscaled behaviour can often be described by a probability model with high fidelity \cite{PT_book}. Another important application of a probability model is a situation in which parameters are uncertain or subject to uncertainty, e.g., due to model misspecification. See Example~\ref{droppingballs} where we attempt to estimate the standard acceleration gravity $g$ using falling times of objects.

In this setting, the observed data in a calibration problem results from performing multiple trials of the experiment modeled by the computer model. Stochastic variation in output data  results from random sampling from the distribution of the physical parameters that sets the conditions for each trial. Consequently, the natural inferential target for the calibration is a distribution on the parameters.

There is often considerable prior information about parameters, e.g., obtained from domain expertise, other experiments, sparse direct observations, and experimental design. This strongly motivates a Bayesian approach where the prior knowledge is described by a prior distribution and the inferential target is a posterior distribution of parameters that updates the prior by conditioning on the field observations. In particular, there is a large and rich literature on Bayesian calibration of computer models that uses parametric statistical models for the posterior.

However, there is often strong motivation to consider a nonparametric approach. For example, the range for the  parameters frequently spans values that lead to widely varying behavior in the system while the computer model is highly nonlinear. In this case, the distributions for the field observations and/or the computer model parameters generally have complex heterogeneous structure.  The use of parametric models to estimate the distribution on the  parameters is problematic in such situations \cite{Hjort_Holmes_2010}. 

\subsection{Overview of the results in this paper}\label{sec:Overview}

We  establish rigorous theoretical and estimation frameworks for nonparametric Bayesian  estimation of the distribution of calibration parameters of a computer model using field data. In lieu of Bayes' Theorem, the frameworks employ disintegration of measures combined with a detailed analysis of the structure of the sub-$\sigma$-algebra induced by the inverse of the computer model. The results include establishing a unique posterior distribution for a given prior with an explicit formula for the corresponding posterior density of calibration parameters, showing the posterior is well-behaved under practical regularity assumptions, showing that the uniform prior produces the posterior with maximum entropy, and constructing and analyzing robustly accurate and practical nonparametric estimators based on importance sampling and an accept-reject approach. 

The results are relevant to any methodology depending on conditioning a distribution on data, including parametric Bayesian calibration and classic descriptive Bayesian statistics. The results in this paper fill in theoretical gaps and validate unverified assumptions across a range of calibration methodologies.

\subsection{Related research}\label{relwork}

Calibration of computer models is an active area of research across many disciplines and there is an extensive research literature. There are various formulations employing a stochastic approach (with or without data), e.g., {Bayesian inverse problem}, {Bayesian melding}, {data-consistent inversion}, {distributional inference}, {inverse ensemble forecasting}, {geophysical inverse problems}, {model parameter estimation}, {populational inverse problem}, {populations of models}, {statistical calibration problem}, and {stochastic inverse problem}. We make some observations about research most closely related to the work in this paper.

\subsubsection*{Application-driven approaches}\label{BCrelwork}

There has been significant work on estimating a posterior density  for parameters in a computer model using experimental data in various fields. Some of the earliest work concentrated on geophysics applications \cite{Mosegard2002,Mosegaard2006}. The work \cite{rumbell2023novel} considers the solution of the SIP using rejection sampling, Markov chain Monte Carlo sampling, and generative adversarial networks.

A related problem is ensemble forecasting, which is a Monte-Carlo approach to predicting behavior of a dynamical system, e.g., as used in weather prediction \cite{Epstein1969,Lewis_2005,Gneiting}. Ensemble forecasting typically involves an inverse stage in which the distribution on initial conditions and/or parameters are estimated using observed data and a computer model by various means. Approaches akin to inverse ensemble forecasting are presented in \cite{marder2011multiple,britton2013experimentally,rumbell2023novel}.

\subsubsection*{Statistically-oriented approaches}\label{sec:StatSCP}

Statistical approaches to calibration have been used in a variety of applications. For example: quantifying variability in biological neurons \cite{marder2011multiple} and cardiac cellular electrophysiology \cite{britton2013experimentally}; estimating the distribution on reaction rates for an autocatalytic reaction \cite{BE13}; estimating the distributions of near shore bathymetry in \cite{BET+14} and Manning's $n$ parameter field quantifying bottom friction  in \cite{BGE+15,GBW+17}; estimating the distribution on the concentration of a contaminant at the source of groundwater pollution \cite{MBD+15}, characterizing the impurity concentration in graphite bricks used in nuclear reactors \cite{GrosskopfBingham20}; estimating the distribution of masses of stars in binary systems and initial orbital separation of the stars during  binary black hole mergers \cite{lin_bbh}; estimating the distribution of thermal diffusivity of protein and bone in cooking steaks \cite{SIPREV}; using muon tomography data for estimating the distribution of the possible location of a subsurface critical mineral deposit  \cite{YazBas2024} and the distribution on the possible geometry of an air gap in a block cave mine \cite{BironBell2025}; overcoming spurious inference in genome-wide association studies \cite{Janani2024}; and estimating the distribution of rate parameters in a computer model of a COVID-19 ``surge'' \cite{kimspaper}.

The parametric Bayesian  approach to model calibration has a rich literature in statistics, e.g.,  \cite{Kennedy_O_JRSSSB_2001,HigdonKennedy04, doi:10.1198,Wang01112009, tuo2015,  plumlee17, GrosskopfBingham20,lin_bbh,SungTuo24,YazBas2024,BironBell2025}. In a different direction, the idea of \textit{Bayesian melding} was proposed in \cite{PR2000}. 

Parametric  and nonparametric Bayesian calibration share a common goal of estimating a distribution for parameters so share a common theoretical foundation of regular conditional probability conditioned on a random variable \cite{Schervish,Rao2005,PT_book}. Additionally, parametric Bayesian calibration often employs a surrogate for the computer model called an emulator and models the  ``discrepancy'' between the calibrated  emulator and the mean of the physical system. Parametric Bayesian calibration employs Bayes’ rule \cite{Schervish}.  Not having access to Bayes' Theorem, nonparametric Bayesian calibration uses disintegration of measures \cite{chang1997conditioning,Schervish,Rao2005,PT_book}. The regular conditional probability resulting from conditioning on the sub-$\sigma-$algebra induced by the computer model $Q$ plays an explicit role in the construction of the posterior conditioned on the data.

Nonparametric Bayesian calibration shares a theoretical foundation with nonparametric Bayesian inference, which is a statistical approach to estimating a distribution for a dataset that replaces a parametric model that depends on a finite dimensional parameter by an infinite dimensional model in a space of suitably smooth probability distributions $\mathcal{F}$ \cite{Hjort_Holmes_2010, ghosalvaart}.

\subsubsection*{Mathematically-oriented approaches}\label{sec:MathSCP}

Mathematical approaches to calibration generally employ a stochastic optimization framework that involves the introduction of stochastic noise and regularization in some form. Some approaches are focused on estimating a {maximum a posteriori} estimate  \cite{Geman1984}.  Other approaches are based on minimizing a discrepancy term or loss functional \cite{Mosegard2002,Mosegaard2006}. Calibration is posed as a computer model constrained optimization problem with the goal of minimizing a distance between the push-forward measure of the computed solution and the distribution of the observed data subject to a constraint or regularization term \cite{li2025inverse,Li2024inverse,Vadeboncoeura2025}.  The problem of adding sufficient information to a system  to enable estimation of the (unknown) trial generating distribution is considered in \cite{addinfosip}. 

The problem of estimating the fixed value of an unknown physical parameter from noisy values of the computer model, known as the Bayesian inverse problem, is well-studied in engineering and mathematics, e.g., \cite{starktenorio, 0266-5611-30-11-110301, Stuart_Bayesian, tarantola}. This approach involves the introduction of noise,  regularization of the model, and a different approach to choosing prior than used in Bayesian statistics. The relationship between this work and the research described in this paper is explored in \cite{SIPREV}.  


\subsubsection*{Prior research of the authors}

Earlier work by the authors uses a different form of disintegration and analyzes a different estimator under several unverified regularity assumptions \cite{BBE11,BES12,BET+14,ButlerUnpub}.  The current formulation of the disintegration, the analysis of generalized contours, the expression for the posterior density, and the proof of a.e.~continuity of the density under the assumption of an Ansatz prior are presented in \cite{LYang}. A form of the maximum entropy property of the posterior corresponding to a uniform prior is given in \cite{JChi}. The thesis \cite{LYang} also gives a basic  convergence result for the simple function estimator developed in \cite{ButlerUnpub,BJW18a}. This paper contains significant extensions and improvements of the analysis in \cite{LYang,JChi}, including treatment of a general prior, strengthening the a.e.~continuity and maximum entropy results, and conducting an extensive analysis of the statistical properties of the estimator.

Important properties of the posterior density as well as  convergence and accuracy properties of estimators based on random sampling are presented in \cite{BJW18a, BJW18b}. Analysis and application have continued \cite{ZM2023, rumbell2023novel, MSB+22, tran2021solving, BGW2020, FBB25, li2025inverse}. The problem of learning an optimal observation map from  noisy datasets is considered in \cite{BH20, MSB+22,RHB25}.  Computational aspects of calibration are investigated in \cite{ButlerUnpub,ButWilYen,BGM+17,BJW18a,condSIP,BBW24,li2025inverse,ZhuEstep2024,ZhuEstep2025}. 

\subsubsection*{The condition of computer model calibration}

The condition of the calibration problem is related to ill-conditioning of mathematical inverse problems. It was introduced in \cite{BGE+15} to analyze the impact of choosing different observation maps in a coastal hydrodynamics application.
In \cite{GBW+17}, it was used for the design of optimal buoy configurations in Bay St. Louis utilizing simulations of Hurricane Gustav (2008).  The definition and theory of condition was significantly expanded in \cite{condSIP}, while \cite{BJP+25} utilizes similar concepts to develop and analyze optimal experimental design criteria.

\subsection{Outline}\label{sec:Outline}

In Sec.~\ref{SIP_intro}, we formulate the statistical calibration problem for a computer model and establish a Bayesian posterior solution. We derive important properties of the posterior in Sec.~\ref{sec:fSCIProp}. In Sec.~\ref{sec:numsol}, we describe and analyze several estimators of the posterior based on reweighting of random samples.   In Sec.~\ref{sec:examples}, we present several examples that illustrate the theoretical development and methodology. We present a conclusion in Sec.~\ref{sec:conclusion}.  Proofs and additional materials are provided in the appendices.

\section{Formulation of the calibration problem}\label{SIP_intro}

We formulate the calibration problem in statistical terms and establish the solution using a Bayesian approach.

\subsection{The statistical calibration problem}\label{sec:Bayessol}

A \textbf{\textit{computer model}} is a map $q=Q(\lambda)$ with input parameters $\lambda$  in the \textbf{\textit{sample space}} $\Lambda$. Typically, $Q(\lambda) = q(Y(\lambda))$ is obtained by applying an \textbf{\textit{observation map}}, or \textbf{\textit{quantity of  interest}}, $q$ to the numerical solution $Y$ of a \textbf{\textit{process model}} $\mathcal{M}(Y,\lambda)=0$ that encapsulates physical properties, e.g.,  a differential equation expressing conservation of mass and energy\footnote{The error introduced from numerical solution of the model may be significant and can be estimated \cite{EstLarWil,EstepMalqvist1,EstepMalqvist2,chauburest,ButlerUnpub}.  Similarly, observational error or noise can be treated using an extension of the approach in this paper \cite{ButWilYen,constrSIP}. For simplicity of presentation, we forgo considering both of these issues.}. See Sec.~\ref{sec:examples} for examples. Generally, $Q$ is not $1-1$, i.e., multiple physical instances  $\lambda$  correspond to the same field observation. This reflects the general case in which experiments measure limited aspects of system behavior.   

Following the  model of a random experiment \cite{PT_book}, we assume that the field observations are measured under stochastic experimental conditions determined by a probability distribution $P^{\mathrm{t}}_\Lambda$ for the calibration parameters. Each sample from  $P^{\mathrm{t}}_\Lambda$ in $\Lambda$ determines the physical conditions $\lambda$ for the trial that results in a field observation $q=Q(\lambda)$. We call $P^{\mathrm{t}}_\Lambda$  the \textbf{\textit{trial generating distribution}}. See  Sec.~\ref{sec:examples}  for examples.

To be precise, we assume
\begin{assumption}\label{Assump:1}
	$(\Lambda,\mathcal{B}_\Lambda)$ is a measurable space, where $\Lambda \subset \mathbb{R}^n$, $n\geq 1$, is a bounded Borel measurable set and $\mathcal{B}_\Lambda$ is the Borel $\sigma-$algebra restricted to $\Lambda$. $P^{\mathrm{t}}_\Lambda$ is a probability measure on $(\Lambda,\mathcal{B}_\Lambda)$. 
	 $Q:(\Lambda,\mathcal{B}_\Lambda) \to (\mathcal{D},\mathcal{B}_\mathcal{D})$ is a random vector, where  $\mathcal{D} =Q(\Lambda)\subset \mathbb{R}^m$ for $1\leq m \leq n$ and $\mathcal{B}_\mathcal{D}$ is the restriction of the Borel $\sigma-$algebra to $\mathcal{D}$. 	
\end{assumption}

Consequently, a collection of field observations $\{q_i\}_{i=1}^K$ corresponds to a sample of ``hidden'' parameter values $\{\lambda_i\}_{i=1}^K\sim P^{\mathrm{t}}_{\Lambda}$, where $q_i = Q(\lambda_i)$ for $1 \leq i \leq K$.  The field observations $\{q_i\}_{i=1}^K$ are therefore distributed according to the  \textbf{\textit{push-forward}}  probability measure,
\[
P_\mathcal{D}(A)  = P^{\mathrm{t}}_\Lambda (Q^{-1}(A)), \quad  A \in \mathcal{B}_\mathcal{D}, 
\]
on $(\mathcal{D},\mathcal{B}_\mathcal{D})$, where the inverse is defined $
Q^{-1}(A) = \{\lambda \in \Lambda: Q(\lambda)\in A\}$ for $ A \in \mathcal{B}_{\mathcal{D}}$. We sometimes use the notation $P_\mathcal{D} = Q P^{\mathrm{t}}_{\Lambda}(A)$. The calibration problem arises in the typical situation where  $P^{\mathrm{t}}_\Lambda$ is unknown.
\begin{definition}[\textbf{Statistical Calibration Problem (SCP)}]\label{SCP}
Determine a probability distribution ${P}_\Lambda$ on $(\Lambda,\mathcal{B}_\Lambda)$ such that  any collection of observations $\{q_i\}_{k=1}^K$ from trials of an experiment modeled by $Q$ is a random sample from the push-forward probability distribution  $P_\mathcal{D} = Q{P}_\Lambda$.
\end{definition}
\noindent We emphasize that $P^{\mathrm{t}}_{\Lambda}$ determines $P_\mathcal{D}$ uniquely, but   the SCP generally has multiple solutions because $Q$ is generally not $1$--$1$ in applications. Below, we describe  the precise relationship between various solutions and how to select a unique solution from among the possibilities.

We develop the foundational theory for the solution of the SCP using an abstract version in which we ``observe''  $P_{\mathcal{D}}$. This circumvents the need to express results in terms of limits when stating  theoretical results. We return to consider data and the SCP when developing the estimation methodology in Sec.~\ref{sec:numsol}. The abstract version is,
\begin{definition}[\textbf{Abstract Statistical Calibration Problem (aSCP)}]\label{aSCP}
	Given the  distribution $P_{\mathcal{D}}$ on $(\mathcal{D},\mathcal{B}_\mathcal{D})$ for results of trials of the experiment modeled by $Q$, compute a  probability distribution ${P}_\Lambda$ on $(\Lambda,\mathcal{B}_\Lambda)$ such that $P_\mathcal{D} = Q{P}_\Lambda$.
\end{definition}
\noindent This abstract version is also called the \textbf{\textit{Stochastic Inverse Problem (SIP)}}.

\subsection{Disintegration of measures and conditional probability}\label{sec:disint}

The solution of the SCP is a posterior  distribution  conditioned on the field observations. Technically, we work in the framework of regular conditional probabilities conditioned on  the values of random variables. Instead of Bayes's Theorem, we employ \textbf{\textit{disintegration of measures}} \cite{Schervish,chang1997conditioning,Rao2005,PT_book}, which provides a systematic way to handle conditioning probability measures on data.

The following result is an immediate consequence of the standard theory of disintegration. While the result is rather impenetrable at first reading, it has profound practical implications for the use of conditional probability.
\begin{theorem}	\label{thm:dis2}
Assume that Assumption~\ref{Assump:1} holds and that $\Psi_\Lambda$ is a bounded measure on $(\Lambda,\mathcal{B}_\Lambda)$. Let  $\Psi_{\mathcal{D}}= Q\Psi_\Lambda $ denote the push-forward measure on $(\mathcal{D}, \mathcal{B}_\mathcal{D})$. There is a family of probability measures $\{\Psi_N(\,\cdot\,|q)\}_{q\in \mathcal{D}}$ on $(\Lambda, \mathcal{B}_\Lambda)$ uniquely defined for $\Psi_{\mathcal{D}}$-almost every $q\in \mathcal{D}$ such that
	\begin{equation}\label{disconc}
	\Psi_N(Q^{-1}(q)|q)=1 \quad \mathrm{ and } \quad \Psi_N(\Lambda \backslash Q^{-1}(q)|q)=0,
	\end{equation}
yielding the disintegration,
	\begin{equation}\label{TheDistin}
	\Psi_\Lambda(A)=\int_{Q(A)}\Psi_N(A|q)\, d\Psi_{\mathcal{D}}(q)= \int_{Q(A)}\int_{Q^{-1}(q)\cap A}\, d\Psi_N(\lambda|q)\, d\Psi_{\mathcal{D}}(q)
	\end{equation}
	for all $A\in\mathcal{B}_\Lambda$. 
\end{theorem}

In this case, $\Psi_N(\,\cdot\,|q)$ is a regular conditional probability measure on $\Lambda$ that  is conditional on the value $q$ in the sense that it is nonzero only on $Q^{-1}(q)$.  We say that the $\Psi_N(\,\cdot\,|q)$ is \textbf{\textit{concentrated}} on $Q^{-1}(q)$ and  $\Psi_\Lambda$ \textbf{\textit{disintegrates}} to $\{\Psi_N(\,\cdot\,|q)\}_{q\in \mathcal{D}}$ and  $\Psi_{\mathcal{D}}$. 
We call $ Q^{-1}(q)$  the \textbf{\textit{generalized contour}} corresponding to $q\in\mathcal{D}$. It is the generalization of a contour curve for a scalar map on two dimensions, see Figures \ref{disint2} and \ref{expdecay_forward}. We note that $Q^{-1}(q)$ is a set of Lebesgue measure $0$, $Q^{-1}(q_1) \cap Q^{-1}(q_2) = \emptyset$ when $q_1 \neq q_2$, and $\Lambda = \bigcup_{q \in \mathcal{D}} Q^{-1}(q)$.

The subscript ``$N$'' is a  reference to the null space of $Q$ when $Q$ is a linear function. For example, a common - but often unremarked - example of disintegration is the Product Measure Theorem  in which $Q$ is the orthogonal projection onto some coordinates \cite{PT_book}. Equation \eqref{TheDistin} can be interpreted as a  ``nonlinear product measure theorem'', see Fig.~\ref{disint2}. The inner integral $\int_{Q^{-1}(q)\cap A}\,$ $ d\Psi_N(\lambda|q)$ gives the probability of the ``slice'' of $A$ that intersects $Q^{-1}(q)$. We integrate these probabilities against the measure of the contours. Example~\ref{ex:condcircle} presents a simple example of disintegration.  
\begin{figure}[tbh]\centering
	\includegraphics[width=.8\textwidth]{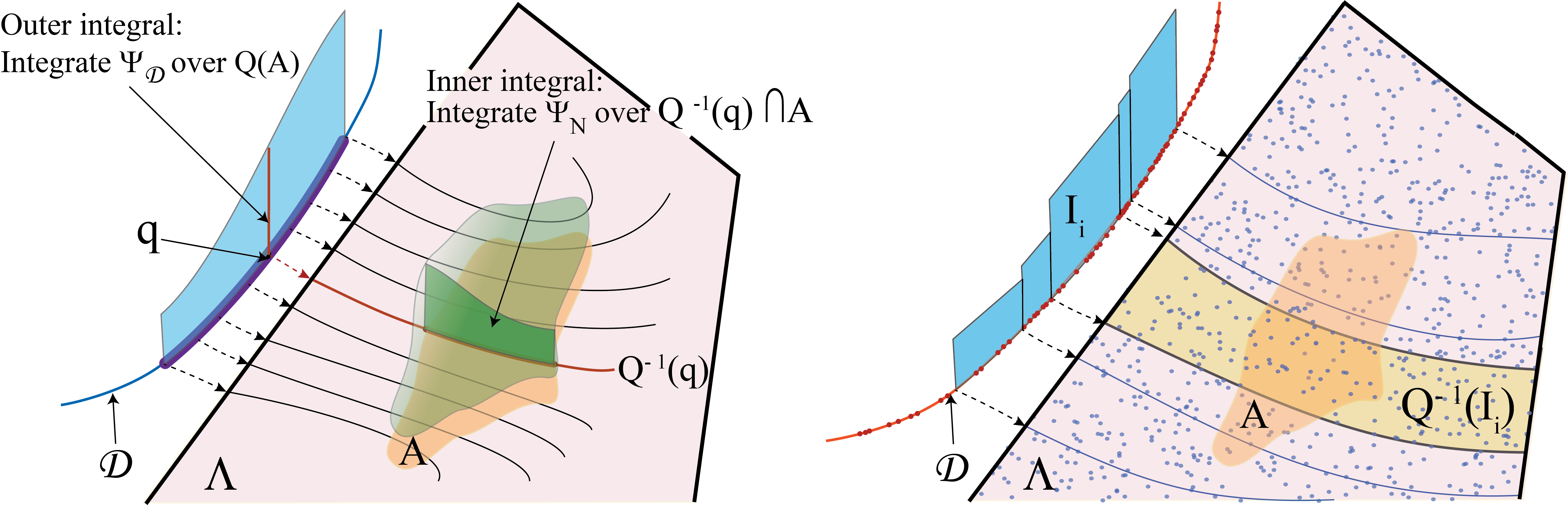}
	\caption{Left: we illustrate disintegration using a figure adapted from \cite{kimspaper}. The probability on $\Lambda$ is represented by a density shaded in green and the probability on $\mathcal{D}$ is represented by a density shaded in blue. The disintegrated density on $Q^{-1}(q)$ is indicated with a darker shade of green. Right: we illustrate the estimation by random sampling. The observed data is indicated by dark red points on $\mathcal{D}$ and the associated empirical estimator is shown in blue relative to the partition $\{I_i\}$. We show the samples $\{\lambda_i\}$ in $\Lambda$ as points and  $Q^{-1}(I_i)$ is shaded in gold.
	}\label{disint2} 
\end{figure}

Disintegration provides a concrete way to express the fact that $Q$ only provides information on a sub-$\sigma-$algebra of $\mathcal{B}_\Lambda$ when the experiment only  reveals partial information about the system and it supports practical application of abstract conditional probability. It turns out that Bayes' Theorem related is closely related to disintegration \cite{Schervish,PT_book}. Thus, parametric and nonparametric approaches to Bayesian computer model calibration share the same underlying theoretical foundation.

\subsection{Bayesian solution of the aSCP}\label{sec:Bayesol}

Theorem~\ref{thm:dis2} implies that for any solution ${P}_\Lambda$ of the aSCP there is a family of probability measures $\{P_N(\,\cdot\,|q)\}_{q\in \mathcal{D}}$ on $(\Lambda, \mathcal{B}_\Lambda)$ uniquely defined  $P_{\mathcal{D}}$ a.e. such that 
\begin{equation}\label{disintmeas}
P_N(Q^{-1}(q)|q)=1,\quad  P_N(\Lambda \backslash Q^{-1}(q)|q)=0, 
\end{equation}
and 
\begin{equation}\label{SIPAbsDis}
	{P}_\Lambda(A)=\int_{Q(A)}P_N(A|q)\, dP_{\mathcal{D}}(q)= \int_{Q(A)}\int_{Q^{-1}(q)\cap A}\, dP_N(\lambda|q)\, dP_{\mathcal{D}}(q)
\end{equation}
for all $A\in\mathcal{B}_\Lambda$.   Of course, the trial generating distribution $P^{\mathrm{t}}_\Lambda$ determines a unique a.e. family  $\{P^{\mathrm{t}}_N(\,\cdot\,|q)\}_{q\in \mathcal{D}}$, but this is unknown.

We are in a situation that is  analogous to the  introduction of Bayes' Theorem in a parametric approach.   We adopt a Bayesian approach by describing prior information about the system in terms of prior distributions.  Under Assumption~\ref{Assump:1}, any family of  $\{P_N(\,\cdot\,|q)\}_{q\in \mathcal{D}}$ of probability measures on $(\Lambda, \mathcal{B}_\Lambda)$  satisfying \eqref{disintmeas} yields a unique solution $P_\Lambda$ of the aSCP satisfying \eqref{SIPAbsDis}. We call such a family $\{P_N(\,\cdot\,|q)\}_{q\in \mathcal{D}}$ an \textbf{\textit{Ansatz prior}}. The notation reflects the fact that we may choose any prior measure on each contour $Q^{-1}(q)$ for every value $q$ of $Q$. 

 While very general, choosing an Ansatz prior $\{P_N(\,\cdot\,|q)\}_{q\in \mathcal{D}}$ is generally infeasible since we have no knowledge of the contours $Q^{-1}(q)$ even if we desire to choose uncountably many probability measures on complex manifolds.
Instead, we observe that any probability measure $P_p$ on $(\Lambda, \mathcal{B}_\Lambda)$ yields an Ansatz prior  family $\{P_N(\,\cdot\,|q)\}_{q\in \mathcal{D}}$ implicitly via Theorem~\ref{thm:dis2}. 
\begin{theorem}[\textbf{nonparametric Bayesian solution of the aSCP}] \label{SIPBayessol}  
	Assume that  Assumption~\ref{Assump:1} holds. Let $P_p$ be a probability measure on $(\Lambda, \mathcal{B}_\Lambda)$. There is a unique solution $P_\Lambda$ of the aSCP satisfying the disintegration \eqref{SIPAbsDis}, where $\{P_N(\,\cdot\,|q)$ $\}_{q\in \mathcal{D}}$ is obtained by disintegration of $P_p$.
\end{theorem}
In this context, $P_p$ is the  \textbf{\textit{prior (distribution)}} chosen based on prior belief and knowledge about the physical  parameters. Inclusion of prior physical knowledge is facilitated because $P_p$ is placed directly on physical quantities. In an epistemological sense, the solution ${P}_\Lambda$ represents the knowledge encapsulated in the prior updated with the information provided by the field observations. We emphasize that ${P}_\Lambda(\, \cdot \,)= P_\Lambda(\,\cdot\, | Q )$ is conditioned on $Q$ and  determines the same  distribution $P_{\mathcal{D}}$  on observed data as the trial generating distribution $P^{\mathrm{t}}_{\Lambda}$. Therefore, we call ${P}_\Lambda$ the \textbf{\textit{posterior (distribution)}}.  The estimator $\hat{P}_{\Lambda,K,N_\jmath}(\,\cdot\, | q_1, \cdots, q_K)$ of ${P}_\Lambda(\,\cdot\,)$ developed below is conditioned directly on the field observations $\{q_1, \cdots, q_K\}$. 

Example~\ref{discretedisint} describes a simple  solution of the aSCP on a discrete probability space that illustrates key points of the theory. 

\section{Important properties of the posterior}\label{sec:fSCIProp}
We derive important properties of the posterior solving the aSCP. These properties provide the conditions that allow for a nonparametric estimator.

\subsection{Properties of $Q^{-1}$}\label{sec:MakeSIPreal}

While disintegration is powerful, it can yield extremely ``rough'' posteriors. We show that reasonable assumptions lead to posteriors that are  regular.   
\begin{assumption}\label{Assump:2}
	$\quad$
	
\begin{enumerate}
	\item\label{assum1} $\Lambda\subset \mathbb{R}^n$ is compact with measurable boundary $\partial \Lambda$ satisfying $\overline{\mathrm{int}\, (\Lambda)} = \Lambda$ and $\mu_{\mathcal{L}}(\partial \Lambda)=0$, where $\mathrm{int}\, (B)$ is the interior of a set $B$, $\overline{B}$ is the closure of a set $B$ and $\mu_{\mathcal{L}}$ is the Lebesgue measure;
	\item\label{assum2}  There is an open set $U_\Lambda$ with $U_\Lambda \supset \Lambda$ such that $Q$ is continuously differentiable  a.e. on $U_\Lambda$;
	\item\label{assum3} $Q$ is \textbf{\textit{geometrically distinct (GD)}} on $U_\Lambda$, which means that the $m\times n$ Jacobian matrix $J_Q$ of $Q$ is full rank $m$ except on a finite union of manifolds of dimension less than or equal to $n-1$.
\end{enumerate}
\end{assumption}

Compactness of $\Lambda$ is reasonable because parameters are generally  bounded for  process models. Determining $\Lambda$ is an important step in choosing a prior, see \cite{kimspaper}. Assumptions  $\overline{\mathrm{int}\, (\Lambda)} = \Lambda$ and  $\mu_{\mathcal{L}}(\partial \Lambda)=0$ are used to establish regularity of the measure of the contours. The smoothness of $Q$  can be established for broad classes of models. The assumption of GD means that $\Lambda$ can be covered by a finite number of open sets up to a set of measure $0$ such that $J_Q$ has full rank on each set. In these sets, no component of $J_Q$ can be expressed as a linear combination of its other components. If this is violated, $Q$ should be simplified to obtain a new map with smaller dimension.  

An important consequence of these assumptions is a concrete interpretation of the inner integral in \eqref{SIPAbsDis}. Roughly speaking, each generalized contour is a smooth manifold that can be locally parameterized as a graph over $\mathbb{R}^{n-m}$  and the inner integral in the disintegration is a surface integral over the manifold.
\begin{theorem}\label{thm:GCisMan}
	Assume Assumptions \ref{Assump:1} and \ref{Assump:2} hold. 
	For almost all $q \in \mathcal{D}$, the generalized contour $
	Q^{-1}(q)=\{\lambda \in U_\Lambda| \, Q(\lambda)=q,\;  \mathrm{rank}(J_Q(\lambda))=m \}$
	is an $n-m$-dimensional a.e. smooth manifold immersed in $\mathbb{R}^n$ and the inner integral in \eqref{SIPAbsDis} is a surface integral defined for a piecewise smooth parametrization on $(\mathbb{R}^{n-m},\mathcal{B}_{\mathbb{R}^{n-m}},\mu_{\mathcal{L}})$. 
\end{theorem}

We discuss the proof  in Sec.~\ref{app:GCISMAN}.  We reflect  Theorem~\ref{thm:GCisMan}  by writing an integral of a conditional density $f$ over a generalized contour $Q^{-1}(q)$ as $\int_{Q^{-1}(q)} f(s|q)\,  ds$, where $s$ indicates the integral is computed for some suitable parametrization of $Q^{-1}(q)$ over $(\mathbb{R}^{n-m},\mathcal{B}_{\mathbb{R}^{n-m}},\mu_{\mathcal{L}})$. Equivalently, the integrals can be expressed in terms of Hausdorff measures on the generalized contours. While Theorem~\ref{thm:GCisMan} is important theoretically,  the parametrization is immaterial for the estimator presented in Sec.~\ref{sec:basicest}.

\subsection{A posterior density}\label{sec:SIPdensol}
We develop a formula for the  posterior  density function. We use $\mu_{\Lambda}$ respectively $\mu_{\mathcal{D}}$ to denote the restriction of the Lebesgue measure to $\Lambda$ respectively $\mathcal{D}$. We assume,
\begin{assumption}\label{Assump:shouldbe3}
$P_\mathcal{D}$ is absolutely continuous with respect to $\mu_\mathcal{D}$, i.e. $dP_\mathcal{D}=\rho_\mathcal{D}\,d\mu_{\mathcal{D}}$ a.e. for an integrable probability density function $\rho_\mathcal{D}$.
\end{assumption}
 By Theorem~\ref{thm:dis2}, the Lebesgue measure disintegrates as, 
\begin{equation}\label{mulamdist}
	\mu_{\Lambda}(A)=\int_{Q(A)}\mu_N(A|q)\, d\tilde{\mu}_{\mathcal{D}}(q)
	=\int_{Q(A)}\int_{Q^{-1}(q)\cap A} \, d\mu_N(\lambda|q)\, d\tilde{\mu}_{\mathcal{D}}(q),
\end{equation}
where $\{\mu_N(\,\cdot\,|q)\}_{q \in \mathcal{D}}$ is a family of regular conditional probability measures on $(\Lambda,\mathcal{B}_\Lambda)$ such that $\mu_N(\,\cdot\,|q)$ is concentrated on $Q^{-1}(q)$ and  $\{\mu_N(\cdot|q)\}_{q \in \mathcal{D}}$ is unique $\tilde{\mu}_{\mathcal{D}}$ a.e. with $\tilde{\mu}_{\mathcal{D}} = Q \mu_\Lambda$.

We note that $P_{\mathrm{p}}$ on $(\Lambda,\mathcal{B}_\Lambda)$ induces the measure $\tilde{P}_{\mathrm{p},\mathcal{D}}= Q P_{\mathrm{p}}$. We assume,
\begin{assumption}\label{Assump:shouldbe4}
The  prior $P_{\mathrm{p}}$ is absolutely continuous with respect to $\mu_\Lambda$, so $dP_{\mathrm{p}} = \rho_{\mathrm{p}} \, d\mu_\Lambda$ for an integrable probability density $\rho_{\mathrm{p}}$, and $\mu_N(\rho_{\mathrm{p}}(\lambda)=0|q) \neq 1$ for almost all $q \in \mathcal{D}$.
\end{assumption}

The main result replaces Bayes' Theorem in a sense. It is proved in Sec.~\ref{sec:solproof}.
\begin{theorem}\label{thm:sol}
	Assume Assumptions \ref{Assump:1}, \ref{Assump:2}, \ref{Assump:shouldbe3}, and \ref{Assump:shouldbe4} hold.   There is a posterior ${P}_\Lambda$   that is absolutely continuous with respect to $\mu_\Lambda$ and  $dP_{\Lambda}={\rho}_\Lambda\, d\mu_\Lambda$ with 
	\begin{equation}\label{UASIPdengen}
		{\rho}_\Lambda(\lambda)=\rho_{\mathcal{D}}(Q(\lambda))\cdot \frac{\rho_{\mathrm{p}}(\lambda)}{\tilde{\rho}_{\mathrm{p},\mathcal{D}}(Q(\lambda))}  \quad \mu_\Lambda~a.e.,
	\end{equation}
	where
	\begin{equation}\label{UASIPdengenscale}
		\tilde{\rho}_{\mathrm{p},\mathcal{D}}(q)=\int_{Q^{-1}(q)}\rho_{\mathrm{p}} \, \frac{1}{\sqrt{\det (J_{Q}J_{Q}^\top) }}\, ds.
	\end{equation}
\end{theorem}
\noindent This is related to the coarea formula in geometric measure theory \cite{evansgar}. The normalization factor \eqref{UASIPdengenscale} reflects the fact that the nonlinear behavior of $Q$ may vary over $\Lambda$.

\subsection{The uniform prior and maximum entropy}\label{sec:unians}

The posterior corresponding to the \textbf{\textit{uniform prior}} $P_p = \mu_\Lambda/\mu_\Lambda(\Lambda)$ has density,
\begin{equation}\label{UASIPdenun}
	{\rho}_\Lambda(\lambda)=\frac{\rho_{\mathcal{D}}(Q(\lambda))}{\tilde{\rho}_{\mathcal{D}}(Q(\lambda))} \;\; \mu_\Lambda \;\; a.e., \quad \tilde{\rho}_{\mathcal{D}}(q)=\int_{Q^{-1}(q)}\frac{1}{\sqrt{\det (J_{Q}J_{Q}^\top) }}\, ds .
\end{equation}
We call the posterior corresponding to the uniform prior the \textbf{\textit{posterior conditioned on the uniform prior}} or the  \textbf{\textit{uniform prior posterior}}. 

The uniform prior has a special property that makes it a natural choice for solving the aSCP in the absence of prior information or belief.  Namely, it is  an expression of the Principle of Insufficient Reason. The translation invariance of the uniform prior respects the fact that since distinct points on the generalized contour $Q^{-1}(q)$ cannot be distinguished using values of $Q$, the probability of points on $Q^{-1}(q)$ should be equal, see Example~\ref{EX:expdecay}. The next theorem shows that the posterior ${P}_\Lambda$ corresponding to the uniform prior is the least biased solution of the aSCP in a certain sense.  We give the proof in Sec.~\ref{sec:proofsunians}.
\begin{theorem}\label{entropy}
	Assume Assumptions \ref{Assump:1}, \ref{Assump:2}, and \ref{Assump:shouldbe3} hold. The posterior  corresponding to the uniform prior has maximum entropy among all posteriors that are absolutely continuous with respect to $\mu_\Lambda$ and whose prior is absolutely continuous with respect to the Lebesgue measure.
\end{theorem}

Note that the estimators developed below handle general priors. A prior that reflects prior knowledge is generally preferable over the uniform prior, see Example~\ref{droppingballs}.

\subsection{Continuity of the posterior density}\label{sec:content}

We develop sufficient regularity conditions that guarantee the density of the posterior is continuous a.e. A key ingredient is the a.e.~continuity of the surface  (Hausdorff) measure of the generalized contours, which depends on the interaction of the generalized contours with the boundary of $\Lambda$.

\begin{figure}[htb]
	\begin{center}
		\includegraphics[height=1.1in]{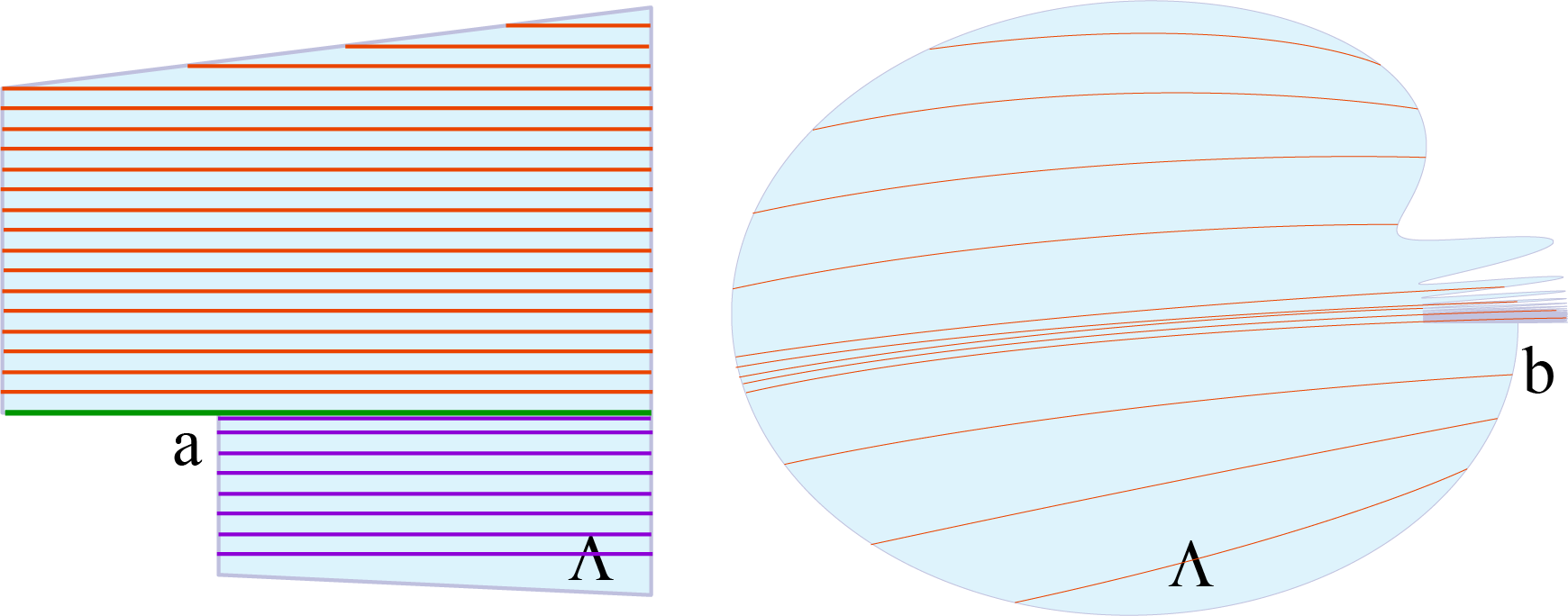}
		\caption{Left: generalized contour lines colored in red, green, and purple. The green contour line intersecting the boundary of $\Lambda$ at the point $a$ is parallel to the boundary  and thus violates the rank condition on 	$\big( J_{Q} \;	J_{B} \big)^\top$. The measure of the contours obviously change discontinuously at the point $a$. Right: a two dimensional domain defined by a map $B$ that fails the assumptions. The boundary has Hausdorff dimension larger than $1$ and measures of the contours that intersect the boundary at point $b$ change rapidly as the intersection point changes.}
		\label{illboundary}
	\end{center}
\end{figure}

To analyze the interaction of the generalized contours with the boundary of $\Lambda$, we assume the boundary is a piecewise smooth manifold determined as the level set of a locally smooth function. Specifically, we assume
\begin{assumption}\label{Assump:3}
 $\partial \Lambda = \{\lambda : B(\lambda) = 0\}$ is determined by a continuously differentiable function $B: U_\Lambda \to \mathbb{R}$ and, for $q\in \mathcal{D}$ with $\partial \Lambda \cap Q^{-1}(q)\neq \emptyset$, the boundary $Q^{-1}(q)\cap \partial\Lambda$ of $Q^{-1}(q)$ in $\Lambda$  is determined by the augmented system, 
\begin{equation*}
	\begin{cases}
		Q(\lambda)-q =0,\\
		B(\lambda) =0.
	\end{cases}
\end{equation*}
For all $q\in \mathcal{D}$ with $\partial \Lambda \cap Q^{-1}(q)\neq \emptyset$, 
$\big(J_{Q} \;J_{B} \big)^\top$ is full rank on $\partial \Lambda \cap Q^{-1}(q)$.
\end{assumption}
\noindent The condition $\big(J_{Q} \;J_{B} \big)^\top$ is full rank on $\partial \Lambda \cap Q^{-1}(q)$ implies that the generalized contour is not tangent to the boundary at $\partial \Lambda \cap Q^{-1}(q)$. We illustrate problems that can arise if these assumptions are not satisfied in Fig.~\ref{illboundary}.  We prove the following result in Sec.~\ref{sec:proofconti}. 
\begin{theorem}\label{thm:conti}
		Assume Assumptions \ref{Assump:1}, \ref{Assump:2}, \ref{Assump:shouldbe3}, \ref{Assump:shouldbe4}, and \ref{Assump:3} hold, and  the density of the prior $\rho_{\mathrm{p}}$ is continuous a.e. Then,
	$\tilde{\rho}_{\mathrm{p},\mathcal{D}}(q)$ is continuous in a neighborhood of each point $q \in \mathcal{D}$. 	In addition, if the density function $\rho_\mathcal{D}$ is continuous a.e., then the density ${\rho}_\Lambda$ of the posterior $P_\Lambda$ corresponding to the prior $P_{\mathrm{p}}$ is continuous in a neighborhood of each point $q\in \mathcal{D}$.
\end{theorem}

We say that a statistical calibration problem that satisfies Assumptions \ref{Assump:1}, \ref{Assump:2}, \ref{Assump:shouldbe3}, \ref{Assump:shouldbe4}, and \ref{Assump:3} is \textbf{\textit{regular}}. We  summarize the results in classic terms.
\begin{theorem}\label{wellposed}
Assume the SCP is regular and  the densities  $\rho_{\mathrm{p}}$  and $\rho_\mathcal{D}$ are continuous a.e. The Bayesian solution of the aSCP is well posed, i.e., the solution exists, it is unique, and it is continuous a.e.
\end{theorem}

In \cite{Prasadan2026}, we strengthen the continuity result by investigating the continuity properties, or ``stability'', of the solution operator of the nonparametric Bayesian calibration procedure.


\section{Estimating the posterior}\label{sec:numsol}

We consider nonparametric estimation of the  posteriors solving the aSCIP and SCP. We focus on an importance sampling approach, but we describe an accept-reject algorithm in Sec.~\ref{sec:reject}. 

\subsection{Estimating the posterior for the aSCP}\label{sec:absest}

The posterior for the aSCP is estimated using simple function approximations computed respectively using the field data, assumed to randomly distributed according to $P_{\mathcal{D}}$, and random samples in $\Lambda$ distributed according to the prior. We want to allow for generality about how these samples are chosen, e.g., to accommodate sampling approaches suitable for high dimensions. The key properties required for the approximations is captured in the following assumption.

\begin{assumption}\label{Assump:4}
	Let $\{M_\ell\}_{\ell=1}^\infty$ and $\{N_j\}_{j=1}^\infty$ be strictly monotone sequences of positive integers. Collections of  points $ \big\{\{d^{(\ell)}_i\}_{i=1}^{M_\ell}\big\}_{\ell = 1}^\infty$ and associated Voronoi tessellations $\{\mathcal{I}_{M_\ell}\}_{\ell = 1}^\infty$ in $\mathcal{D}$ and points $ \big\{\{\lambda^{(\jmath)}_i\}_{i=1}^{N_\jmath}\big\}_{\jmath = 1}^\infty$ and associated Voronoi tessellations  $\{\mathcal{T}_{N_\jmath}\}_{\jmath = 1}^\infty$ in $\Lambda$ are $\mathcal{B}_{\mathcal{D}}$-consistent resp. $\mathcal{B}_\Lambda$-consistent. Further, for sets $C\in \mathcal{B}_{\mathcal{D}}$ with $\mu_{\mathcal{D}}(\partial C) = 0$ and $A\in \mathcal{B}_{\Lambda}$ with $\mu_{\Lambda}(\partial A) = 0$,
	\begin{equation}\label{setapproxassump}
		\lim_{\ell \to \infty} \mu_{{\mathcal{D}}}(C_{M_\ell} \triangle C) = 0\; \mathrm{a.s.} \; \mathrm{and} \;  \lim_{\jmath \to \infty} \mu_{{\Lambda}}(A_{N_\jmath} \triangle A) = 0\; 
		\mathrm{ a.s.} 
	\end{equation}
	where $\triangle$ indicates the \textbf{\textit{symmetric set difference}}.
\end{assumption}

For background material on  $\mathcal{B}$-consistent random samples, see Sec.~\ref{sec:randprop}. 
Generating collections of points by a Poisson point process that has an a.e.~positive probability density function with respect to the Lebesgue measure is $\mathcal{B}_{\Omega}$-consistent, see Theorem~\ref{lem:SLLNS}\footnote{The approximation result \eqref{setappconv} also holds for sequences of partitions of $\Omega$ comprising open or closed generalized rectangles with fixed aspect ratios or open balls \cite{folland2013real,PT_book}.}.

The approximation result is proved in Sec.~\ref{sec:disapproof}.

\begin{theorem} \label{thm:disap}
	Assume the SCP is regular, Assumption \ref{Assump:4} holds, and the densities  $\rho_{\mathrm{p}}$  and $\rho_\mathcal{D}$ are continuous a.e. There exists a sequence of approximations $P_{\Lambda,M_\ell,N_\jmath}$ constructed from simple functions and requiring only calculations of measures of cells in the partitions  $\{\mathcal{I}_{M_\ell}\}_{\ell = 1}^\infty$ and $\{\mathcal{T}_{N_\jmath}\}_{\jmath = 1}^\infty$ such that 
	\begin{equation}\label{disapconv}
		{P}_\Lambda(A)= \lim_{\ell\to \infty}\lim_{\jmath\to \infty}{P}_{\Lambda,M_\ell,N_\jmath}(A_{N_\jmath}),
	\end{equation}
	for all sets $A \in \mathcal{B}_\Lambda$ with $\mu_\Lambda (\partial A)=0$.
\end{theorem}

The order of the limits in \eqref{disapconv} is important. As the cells in the partitions of $\mathcal{D}$ become finer, their inverse images in $\Lambda$ become ``narrower'', and the cells partitioning $\Lambda$ must also become finer.

\subsection{Estimating the posterior for the SCP}\label{sec:basicest}

The methodology for estimating the posterior distribution for the SCP is stated in Algorithm~\ref{alg:estSCP}.
\begin{algorithm}[htb]
	\caption{Estimating the SCP posterior}\label{alg:estSCP}
	\begin{algorithmic}[1]
		
		\STATE  Compute an empirical estimate of the distribution or the density for the field data 
		
		\STATE Evaluate the computer model at a set of random samples from the prior
		
		\STATE Combine 1.~and 2.~using  importance sampling to reweight the samples from the prior to obtain a sample from the corresponding posterior solving the SCP (using \eqref{accessapp} below)

	\end{algorithmic}
	
\end{algorithm}

Important implementation choices are the methodologies used to estimate the distribution from the field data and to draw samples from the prior. In this paper, we analyze an estimator computed using a  histogram estimate of $\rho_{\mathcal{D}}$ and plain random sampling from the prior. This is not unreasonable, given that a large proportion of computer model calibration problems involve scalar or low dimensional observations and many applications involve  parameter spaces that are not very high dimensional. In general, we view the treatment of a histogram estimate for the density of the field data as providing a ``baseline'' approximation result. We extend the main convergence result to other estimators for the distribution of the field data in Sec.~\ref{sec:altrhoest}.

In cases where the parameters and possibly the computer model output are high dimensional, the methodologies for the first two steps of Algorithm~\ref{alg:estSCP} must be chosen appropriately, e.g., employing adaptive sampling, Markov Chain Monte Carlo sampling, machine learning, and so on. Computational issues  are discussed in \cite{ZM2023, rumbell2023novel, MSB+22, tran2021solving, BGW2020, FBB25, li2025inverse,BH20, MSB+22,RHB25,ButWilYen,BGM+17,BJW18a,condSIP,BBW24,li2025inverse,ZhuEstep2024,ZhuEstep2025}. 

Nominally, there are three discretizations involved with the SCP; the number of observed datum $K$, the number of cells $M$ in a partition of $\mathcal{D}$, and the number of cells $N$  in a partition of $\Lambda$.  However, limits on the accuracy of empirical density approximations means that $M$ must be related to $K$ in practice.  We denote the field observations by $\{q_i \}_{i=1}^\infty$, which we assume is a random sample generated by a Poisson point process with distribution $P_{\mathcal{D}}$.  For $K\geq 1$, we construct an approximation using  $\{q_i \}_{i=1}^K$ that is associated with a partition $\mathcal{I}_{M_K}$ of $\mathcal{D}$, where the number of cells $M_K$ in $\mathcal{I}_{M_K}$ is a function of $K$ that  tends monotonically to  $\infty$ as $K\to \infty$. We assume the family $\{\mathcal{I}_{M_K}\}_{K = 1}^\infty$ satisfies Assumption~\ref{Assump:4}. The approximation is also associated with partitions $\{\mathcal{T}_{N_\jmath}\}_{\jmath = 1}^\infty$ of $\Lambda$, which we assume satisfies Assumption~\ref{Assump:4}.

To construct the approximation, we begin with
\begin{equation}\label{hatpi}
	\hat{p}_{K,i} = \frac{\# q_k \in I_i^{(K)}}{K}, \quad 1 \leq i \leq M_K.
\end{equation}
We generalize to other estimators in Sec.~\ref{sec:altrhoest}. The approximation for ${P}_\Lambda$ is
\begin{equation}\label{UAcountapp}
\hat{P}_{\Lambda,K,N_\jmath}(A) =
\sum_{i=1}^{M_K} \, \frac{\# \lambda^{(\jmath)}_j \in A \cap Q^{-1}(I_i^{(K)})}{\# \lambda^{(\jmath)}_j  \in Q^{-1}(I_i^{(K)})}\, \hat{p}_{K,i} , \quad A \in \mathcal{B}_\Lambda.
\end{equation}
We  write the estimator in a form more conducive to analysis,
\begin{equation}\label{UAuseapp}
	\hat{P}_{\Lambda,K,N_\jmath}(A) = \sum_{i=1}^{M_K} \, \frac{\sum_{j=1}^{N_\jmath} \chi_{\lambda_j^{(\jmath)}}(Q^{-1}(I_{i}^{(K)})\cap A ) }{ \sum_{j=1}^{N_\jmath} \chi_{\lambda_j^{(\jmath)}}(Q^{-1}(I_{i}^{(K)})) }\cdot \frac{1}{K} \sum_{k=1}^{K} \chi_{q_k}(I_{i}^{(K)}), \quad A \in \mathcal{B}_\Lambda ,
\end{equation}
where  $\chi_{A}$ denotes the \textbf{\textit{characteristic}} or \textbf{\textit{indicator function}} of $A$. Observing that computing $Q^{-1}(I_{i}^{(K)})$ is impractical, we can write
\begin{equation}\label{accessapp}
\hat{P}_{\Lambda,K,N_\jmath}(A) =
	\sum_{i=1}^{M_K} \, \frac{\# \lambda_j^{(\jmath)} \in A: \, Q(\lambda_j^{(\jmath)}) \in I_i^{(K)}}{\# \lambda_j^{(\jmath)} :\,  \, Q(\lambda_j^{(\jmath)}) \in I_i^{(K)}}\; \frac{\# q_k \in I_i^{(K)}}{K}, \quad A \in \mathcal{B}_\Lambda,
\end{equation}
which shows that the estimate involves counting points in cells in the partition of $\mathcal{D}$.

We illustrate the estimator in Fig.~\ref{disint2}.  We partition the range $\mathcal{D}$ into a collection of cells and use the field data to approximate the probability $P_{\mathcal{D}}$ of each cell.  The inverse image under $Q$ of each cell in the partition of $\mathcal{D}$ is a ``fat'' contour. We compute a sample of points in $\Lambda$ distributed according to the prior. We estimate the probability of the intersection of a set $A$ with a fat contour by taking the ratio of the number of samples in the  intersection of $A$ with the fat contour and the total number of samples in the fat contour. Finally, we sum up all the estimates over fat contours that intersect $A$, reweighting the estimate in each fat contour with the probability of its corresponding cell in $\mathcal{D}$. Note that the a.e. continuity of the surface measure of the contours implies that when a fat contour is sufficiently thin, the surface measures of almost all the contours in it are approximately equal.


\subsubsection{Statistical properties of the estimator}\label{sec:eststatprop}

We derive properties of the estimator that quantify convergence in the limit of large numbers of samples. The results imply that if sufficient random samples from the prior are used to build the importance sampling estimator for the posterior, the error of the estimator is determined primarily by how well the distribution of the field observations can be estimated. That is a statistical sample size problem and it is the best possible result we can expect. We illustrate in Example~\ref{droppingballs}.

The ratio $\breve{p}_i^{(K)} \big/ p_i^{(K)}$, where
\[ 
p_i^{(K)} = P_{\mathrm{p}}\big(Q^{-1}(I_{i}^{(K)})\big) \quad \mathrm{and} \quad  \breve{p}_i^{(K)} = P_{\mathrm{p}}\big(A\cap Q^{-1}(I_{i}^{(K)})\big), 
\]
for all $i$ and $K$, is the probability that a point selected at random on $Q^{-1}(q)$ intersects $A$. It features prominently. We prove the next theorem in Sec.~\ref{sec:sampsolproof}.
\begin{theorem}\label{thm:sampsolprop}
Assume the SCP is regular, Assumption \ref{Assump:4} holds, and the densities  $\rho_{\mathrm{p}}$  and $\rho_\mathcal{D}$ are continuous a.e. 
\begin{enumerate}
\item\label{thm:sampsolprop1} The estimator is asymptotically unbiased. For $A \in \mathcal{B}_\Lambda$,
\begin{equation}\label{Esampsol}
{E}\big(\hat{P}_{\Lambda,K,N_\jmath}(A) \big) = 	\sum_{i=1}^{M_K} \frac{\breve{p}_i^{(K)}}{p_i^{(K)}}\big(1 - (1-p_i^{(K)})^{N_\jmath})\, P_{\mathcal{D}}\big(I_{i}^{(K)}\big),
\end{equation}
and $
	\lim_{K \to \infty }\lim_{\jmath \to \infty } {{E}\big(\hat{P}_{\Lambda,K,N_\jmath}(A) \big))} = {P}_\Lambda(A)$.
\item\label{thm:sampsolprop2} For $A \in \mathcal{B}_\Lambda$,
\begin{multline}\label{Varsampsol}
\lim_{\jmath \to \infty}\mathrm{Var}\,\big(\hat{P}_{\Lambda,K,N_\jmath}(A) \big)
\leq 	\frac{1}{K}\Big(\sum_{i=1}^{M_K} P_{\mathcal{D}}(I_{i}^{(K)}) \frac{(\breve{p}_i^{(K)})^2}{(p_i^{(K)})^2} -  \Big(\sum_{i=1}^{M_K} P_{\mathcal{D}}(I_{i}^{(K)}) \frac{\breve{p}_i^{(K)}}{p_i^{(K)}}\Big)^2\Big)
\end{multline}
and  $\lim_{K \to \infty}\lim_{\jmath \to \infty} \mathrm{Var}\,\big(\hat{P}_{\Lambda,K,N_\jmath}(A) \big) = 0.$
\item\label{thm:sampsolprop3} The estimator is strongly consistent. For $A \in \mathcal{B}_\Lambda$,
\begin{equation}\label{samappconsist}
	\lim_{K \to \infty} \lim_{\jmath \to \infty} \hat{P}_{\Lambda,K,N_\jmath}(A) = {P}_\Lambda (A).
\end{equation}
\end{enumerate}	
\end{theorem}

\subsubsection{Asymptotic  properties of the error of the estimator}\label{sec:asymperror}
 We write,
\begin{align}
	&\hat{P}_{\Lambda,K,N_\jmath}(A) = \sum_{i=1}^{M_K} \frac{\breve{p}_i^{(K)}}{p_i^{(K)}} P_{\mathcal{D}}(I_{i}^{(K)})  + \sum_{i=1}^{M_K} \frac{\breve{p}_i^{(K)}}{p_i^{(K)}} \Big( \frac{1}{K} \sum_{k=1}^{K} \chi_{q_k}(I_{i}^{(K)})-	P_{\mathcal{D}}(I_{i}^{(K)}) \Big) \nonumber\\
		&\qquad +  \sum_{i=1}^{M_K}  \Bigg( \frac{\sum_{j=1}^{N_\jmath}\chi_{\lambda_j^{(K)}}(Q^{-1}(I_{i}^{(K)})\cap A ) }{ \sum_{j=1}^{N_\jmath} \chi_{\lambda_j^{(K)}}(Q^{-1}(I_{i}^{(K)}))} - \frac{\breve{p}_i^{(K)}}{p_i^{(K)}}\Bigg) \cdot \frac{1}{K} \sum_{k=1}^{K} \chi_{q_k}(I_{i}^{(K)})= T_1+ T_2 + T_3.\nonumber
\end{align}
Since $\lim_{K\to\infty} T_1 = {P}_\Lambda(A)$ a.e., we analyze the stochastic errors $T_2$ and $T_3$. $T_2$ depends on the error of the empirical approximations to $P_{\mathcal{D}}(I_{i}^{(K)})$ while $T_3$ depends on the empirical approximations to ${\breve{p}_i^{(K)}}/{p_i^{(K)}}$.

For $T_2$, we prove the following in Sec.~\ref{sec:asymperrorproof}.
\begin{theorem}\label{CLTerror2}
Assume the SCP is regular, Assumption \ref{Assump:4} holds, and the densities  $\rho_{\mathrm{p}}$  and $\rho_\mathcal{D}$ are continuous a.e.  Then, $
		\lim_{K\to \infty}	\sqrt{K}\, T_2 \stackrel{d}{=} \mathcal{N}(0,\sigma_{p}^2)$, where
		\begin{equation}\label{thm:CLTerror2}
			\sigma_{p}^2  = \int P_{\mathrm{p},N}(A|q)^2\, dP_{\mathcal{D}} (q) - \left(\int P_{\mathrm{p},N}(A|q)\, dP_{\mathcal{D}}(q)\right)^2.
		\end{equation}
\end{theorem}
\noindent The second moment $\sigma^2_p$ quantifies the variation in the disintegrated conditional probabilities conditioned on the data as the observations vary.  We can control $T_2$ by increasing the amount of observed data.

Obtaining a simple asymptotic result for $T_3$ is apparently more difficult. We summarize the analysis presented in Sec.~\ref{sec:asymperrorproof}.  The size of $T_3$ is determined by the expressions,
\[
T_{3,i,\jmath} = \frac{\sum_{j=1}^{N_\jmath}\chi_{\lambda_j^{(K)}}(Q^{-1}(I_{i}^{(K)})\cap A ) }{ \sum_{j=1}^{N_\jmath} \chi_{\lambda_j^{(K)}}(Q^{-1}(I_{i}^{(K)}))} - \frac{\breve{p}_i^{(K)}}{p_i^{(K)}}, \quad 1 \leq i \leq M_K,\; 1 \leq \jmath.
\]
We prove,
\begin{theorem}\label{CLTerror3}
	Assume the SCP is regular, Assumption \ref{Assump:4} holds, and the densities  $\rho_{\mathrm{p}}$  and $\rho_\mathcal{D}$ are continuous a.e.  Then, $\lim_{\jmath \to \infty} {E}(T_{3,i,\jmath}) = 0 $ and $\lim_{\jmath \to \infty} \mathrm{Var}\, (T_{3,i,\jmath}) = 0 $.
\end{theorem}
\noindent Roughly speaking, the proofs in Sec.~\ref{sec:asymperrorproof} imply that to make $T_3$ small, $N_\jmath p_i^{(K)}$ has to be large for all $i$. As $K \to \infty$, $Q^{-1}(I_{i}^{(K)})$ become ``narrower'' and $p_i^{(K)} \downarrow 0$.   $N_\jmath$ has to increase at least linearly with $1/p_i^{(K)}$ to ensure $T_3 $ becomes small.

\subsubsection{Using other estimators of $\rho_{\mathcal{D}}$}\label{sec:altrhoest}
We consider a general estimator $\hat{\rho}_{k,i} = \hat{P}_{\mathcal{D}}\big(I_i^{(k)}\big)$ in \eqref{UAcountapp}, where $\hat{P}_{\mathcal{D}}$ has density $\hat{\rho}_{\mathcal{D}}$ with respect to $\mu_{\mathcal{D}}$.  We prove the following theorem in Sec.~\ref{sec:gensampsolproof}.
\begin{theorem}\label{thm:gensampsolprop}
	Assume the SCP is regular, Assumption \ref{Assump:4} holds, and the densities  $\rho_{\mathrm{p}}$  and $\rho_\mathcal{D}$ are continuous a.e. 
	If $\hat{\rho}_{\mathcal{D}} \stackrel{L^1}{\rightarrow} \rho_{\mathcal{D}}$ $P_\mathcal{D}$ a.e., then $\hat{P}_{\Lambda,K,N_\jmath}$ is strongly consistent and \eqref{samappconsist} holds. 
\end{theorem}

We apply this theorem to a kernel density estimate (KDE). Choosing a \textbf{\textit{kernel function}} $\Xi$ (measurable, nonnegative, and $\int \Xi \, d\mu_{\mathcal{D}} = 1$), we define
\[
\hat{\rho}_{\mathcal{D}}(q) = \frac{1}{K h_K^m} \, \sum_{i=1}^K \Xi\left(
\frac{q - q_i}{h_K}\right),
\]
where $h_K$ is a \textbf{\textit{scaling length}}. Theorem~6.1 of \cite{DevroyeGyorfi} implies that if
\begin{equation}\label{KDEscalelen}
\lim_{K\to \infty}	h_K + \big(K h_K^m\big)^{-1} = 0 ,
\end{equation}
then $\int_{\mathcal{D}} |\hat{\rho}_{\mathcal{D}} - {\rho}_{\mathcal{D}} | \, d\mu_{\mathcal{D}} \to 0$ and Theorem~\ref{thm:gensampsolprop} holds.

On another hand, if ${\rho}_{\mathcal{D}}$ belongs to a parameterized family of distributions $P_{\mathcal{D}}(q:\theta)$ with parameters $\theta \in \Theta$, we  define the estimate  $\hat{\rho}_{\mathcal{D}}(q)={\rho}_{\mathcal{D}}(q:\hat{\theta})$ for parameter estimate $\hat{\theta}$. To use Theorem~\ref{thm:gensampsolprop}, we require family of distributions and an estimate $\hat{\theta}_K$ computed from data $\{q_i\}_{i=1}^K$ such that $\hat{\rho}_{\mathcal{D}}(\cdot:\hat{\theta}_K) \stackrel{L^1}{\rightarrow} \rho_{\mathcal{D}}(\,\cdot\,:\theta)$ $P_\mathcal{D}$ a.e. as $\lim_{K \to \infty} \hat{\theta}_K = \theta$. 

\subsection{An estimator based on accept-reject sampling}\label{sec:reject}

We can use accept-reject sampling instead of importance sampling in the inversion step (Step 3) of Algorithm~\ref{alg:estSCP}. We adapt a methodology from \cite{BJW18a}, which was applied on a problem with a 100-dimensional parameter space,  to construct  an estimator that produces a collection of independent random samples approximately distributed according to the posterior. We use the same setup as for the random sampling method presented in Sec.~\ref{sec:numsol}, but simplify notation.  The algorithm, shown in Algorithm~\ref{alg:accept-reject} in Sec.~\ref{sec:rejectalg}, uses an acceptance criteria based on the ratio $\rho_{\mathcal{D}}(Q(\lambda)) \big/ \tilde{\rho}_{\mathrm{p},\mathcal{D}}(Q(\lambda))$ 
applied to a sample $\lambda$ from the prior density $\rho_{\mathrm{p}}(\lambda)$, where
\begin{equation*}
	\hat{\rho}_{\mathcal{D}}(Q(\lambda)) =  \sum_{i=1}^M\left( \frac{1}{K\,\mu_{\mathcal{D}}(I_i)} \sum_{k=1}^{K} \chi_{q_k}(I_{i})\right) \, \chi_{I_i}(Q(\lambda)), 
\end{equation*}
and
\begin{equation*}
	\widehat{\tilde{\rho}}_{\mathrm{p},\mathcal{D}}(Q(\lambda)) =  \sum_{i=1}^M\left( \frac{1}{N\,\mu_{\mathcal{D}}(I_i)} \sum_{j=1}^{N} \chi_{Q(\lambda_j)}(I_{i})\right) \, \chi_{I_i}(Q(\lambda)) ,
\end{equation*}
which can be computed after Steps 1 and 2 of of Algorithm~\ref{alg:estSCP}.  We assume $\{I_i\}_{i=1}^M$ is a partition of $\mathcal{D}$ from a family satisfying Assumption~\ref{Assump:4}, and $\{\lambda_j\}_{j=1}^N \in \Lambda$ are independent samples from the prior $P_{\mathrm{p}}$. In words, $\hat{\rho}_{\mathcal{D}}(Q(\lambda))$ is the estimated density computed from the observed data and  $\widehat{\tilde{\rho}}_{\mathrm{p},\mathcal{D}}(Q(\lambda))$  is the estimate of the density of the push-forward measure of the prior.

The accepted samples  $\{\breve{\lambda}_\ell\}_{\ell=1}^{\breve{N}}$ are  approximately distributed according to the posterior $P_\Lambda$ with density,
\begin{equation*}
	{\rho}_\Lambda(\lambda)=\frac{\rho_{\mathcal{D}}(Q(\lambda))}{\tilde{\rho}_{\mathrm{p},\mathcal{D}}(Q(\lambda))}\cdot \rho_{\mathrm{p}}(\lambda).
\end{equation*}
In our experience, the estimator based on reweighting random samples is more computationally efficient when there is a high reject rate in the accept-reject estimator.

\section{Examples}\label{sec:examples}

We illustrate the theoretical results using examples chosen for simplicity of presentation. We explore convergence using simulations with varying sample sizes in Example~\ref{EX:expdecay}. We repeat Example~\ref{EX:expdecay} using the accept-reject estimator in Example~\ref{EX:expdecay5}. We formulate a calibration problem for a real-world experiment in Example~\ref{droppingballs}. We present an example with a higher dimensional parameter space in Example~\ref{S:CM_Example}. The code, data, and computational details are provided in the supplemental material. 

\begin{example}\label{EX:expdecay}
	We consider a simple model for exponential decay,
	\begin{equation}\label{expdecay}
		\begin{cases} \frac{dy}{dt} = - \lambda_2 y, & 0 < t \leq T,\\
			y(0) = \lambda_1, & 
		\end{cases}
	\end{equation}
	where we treat the initial condition and the rate as unknown. The quantity of interest is  $Q(\lambda_1,\lambda_2) = y(T) = \lambda_1 \exp(-\lambda_2 T)$. We set $\Lambda = [0,1]\times[0,1]$ and $\mathcal{D} = [0,1]$. 
	
	In Fig.~\ref{expdecay_cont}, we plot  $32$ generalized contours $Q^{-1}(q)=\{(\lambda_1,\lambda_2): \lambda_1 = q \exp(T \, \lambda_2), 0 \leq \lambda_2 \leq 1\}$.  
	\begin{figure}[tbh]\centering
		\includegraphics[width=.65\textwidth]{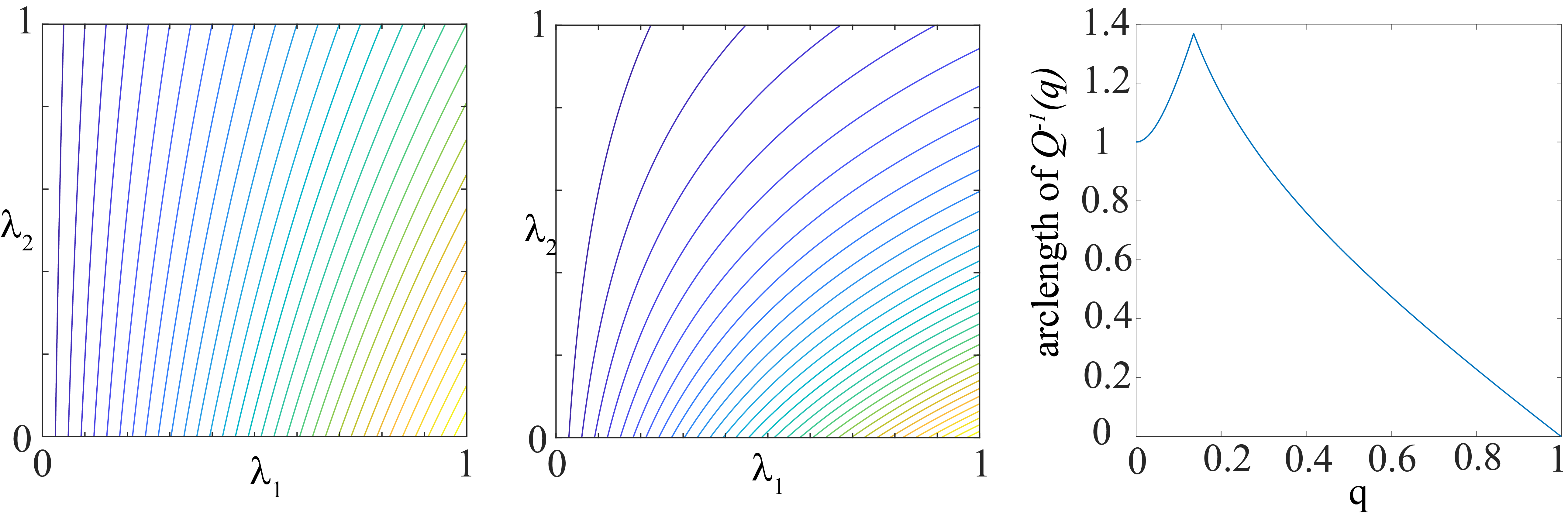}
		\caption{We show generalized contours for the map $Q$ for the exponential decay model in Example~\ref{EX:expdecay} corresponding to $T=.5$ (left) and $T=2$ (center) for $32$ values of $q$. On the right, we plot the arclength of the generalized contours $Q^{-1}(q)$  against $q$ for $T=2$.}\label{expdecay_cont} 
	\end{figure}
	To generate \textbf{\textit{synthetic data}}, we assume a trial generating distribution $P^{\mathrm{t}}_\Lambda$ equal to a product of independent $\mathrm{Beta}\,(12,12)$ distributions for $\lambda_1$ and $\lambda_2$. We draw samples from $P^{\mathrm{t}}_\Lambda$ and evaluate the solution to generate observed data.  We plot empirical densities for $P^{\mathrm{t}}_\Lambda$ and the corresponding output distributions $P_{\mathcal{D}}$ at $T=.5$ and $T=2$ using $200,000$ samples in Fig.~\ref{expdecay_forward}.
	
	\begin{figure}[tbh]\centering
		\includegraphics[width=.8\textwidth]{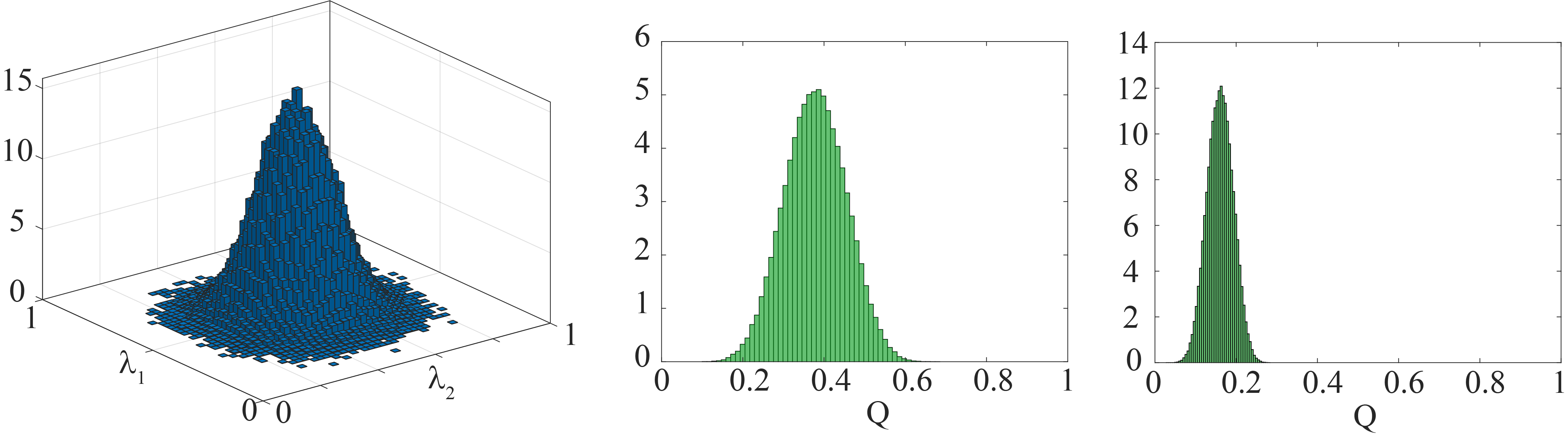}
		\caption{Left: empirical density for $P^{\mathrm{t}}_\Lambda$ for Example~\ref{EX:expdecay} computed using $200,000$ points. Center: empirical density for $P_{\mathcal{D}}$ at $T=.5$. Right: empirical density for $P_{\mathcal{D}}$ at $T=2$.		}\label{expdecay_forward} 
	\end{figure}
	
	We estimate the posterior corresponding to the uniform prior using the nonparametric estimator  presented in Sec.~\ref{sec:numsol} with a large number of samples in order to virtually eliminate effects of finite sampling.  We generate synthetic data by sampling $K=16\times 10^6$ points in $\Lambda$ randomly from $P^{\mathrm{t}}_\Lambda$ and building an empirical output distribution on a partition of $\mathcal{D}$ with $M=100$ cells.  We compute an empirical estimator of the posterior using  $N=16\times 10^6$ uniformly distributed points in $\Lambda$.	
	
	In Fig.~\ref{expdecaySIPS}, we plot heatmaps of the estimators for $T=.5$ and $T=2$ produced by computing $\{P_\Lambda(B_i)\}_{i=1}^{6400}$ for a partition $\{B_i\}_{i=1}^{6400}$ of $\Lambda$ into small cells $B_i$ and generating a heatmap from the probabilities. Recalling the contours shown in Fig.~\ref{expdecay_cont}, the posteriors have the characteristic ``ridge'' shape oriented along contours that reflects the fact that parameter values on the same contour have equal probability due to the choice of the uniform prior. This reflects the \textit{\textbf{indeterminacy}} induced by the loss of information in the experimental observation because we only observe the value of the solution at a particular time.
	\begin{figure}[tbh]\centering
		\includegraphics[width=\textwidth]{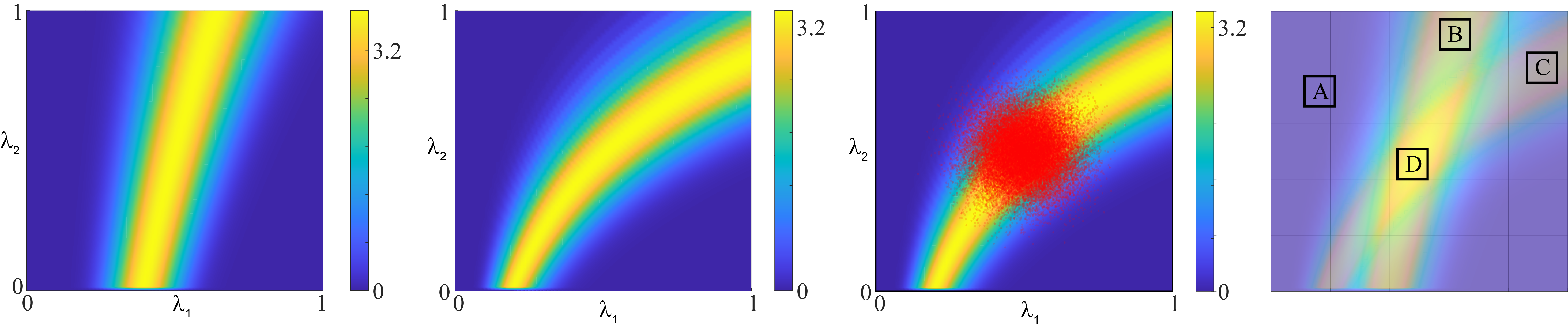}
		\caption{Left two plots: Heatmaps of estimators of the posterior for the exponential decay model using the uniform prior for $T=.5$ (left) and $T=2$ (center). Second plot from the right: The posterior using the uniform prior for $T=2$ together with the sample points (red) from the trial generating distribution used to create the simulated output data.  Rightmost plot: Overlays of the heatmaps for $T=.5$ and $T=2$.
		}\label{expdecaySIPS} 
	\end{figure}
	
To give intuition about the connection between the trial generating distribution and the posterior, we plot the samples from the trial generating distribution that produced the simulated output data together with the posterior for $T=2$ in the third panel of Fig.~\ref{expdecaySIPS}. Choosing the uniform prior means that the  posterior ``extends'' the profile of the trial generating distribution along the generalized contours. In Example~\ref{droppingballs}, we show the use of an informed prior that results in a posterior with more localized support.

We illustrate the use of the posterior to estimate probabilities of events. In the right hand plot in Fig.~\ref{expdecaySIPS}, we show four events of equal size located in different parts of $\Lambda$. We see that event $A$ has relatively low probability with respect to the posteriors for $T=.5$ and $T=2$, while event $D$ has relatively high probability for both posteriors. Events $B$ and $C$ are relatively high probability for one of the posteriors but not the other. 

We present visual evidence regarding convergence. We hold the discretization and computational parameters fixed except as noted.  In the first experiment, we vary the number of points $N_\jmath \in \{50^2, 100^2, 400^2, 1600^2\}$ in the samples $\{\lambda_i\}_{i=1}^{N_\jmath}$ from the prior. We plot the results in Fig.~\ref{expdecaySIPS2}. There are visible resolution issues in the heatmap when $N_\jmath = 50^2$.  The increasing resolution is apparent up to $N_\jmath = 1000^2$ but there is little improvement after that.
	\begin{figure}[tbh]\centering
		\includegraphics[width=.85\textwidth]{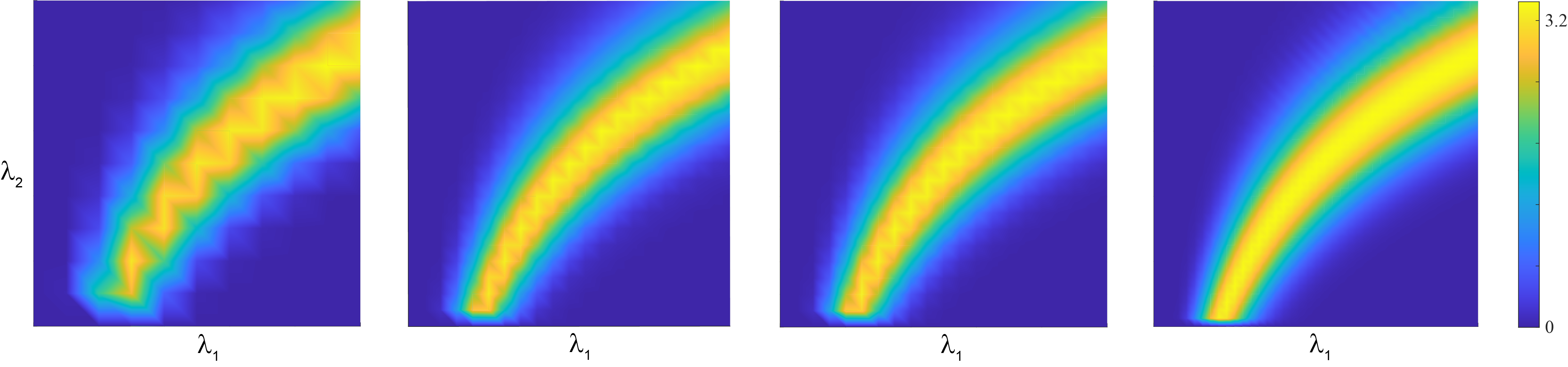}
		\caption{Heatmaps of the uniform prior posteriors for the exponential decay model for $T=2$ computed using tessellation points $\{\lambda_i\}_{i=1}^{N_\jmath}\in \Lambda$ for $N_\jmath \in \{50^2, 100^2, 400^2, 1600^2\}$ in order from left to right.}\label{expdecaySIPS2} 
	\end{figure}
	In the second experiment, we vary the number of cells $M_K \in \{12, 25, 50, 100, 200\}$ in the partition $\{I_i\}_{i=1}^{M_K}$ of $\mathcal{D}$. We plot the results in Fig.~\ref{expdecaySIPS3}.  The lower values of $M_k$ result in poor resolution of $P_{\mathcal{D}}$ because the ``fat'' contours are too fat. The increasing resolution is apparent up to $M_K =100$ but there is little improvement after that.
	\begin{figure}[tbh]\centering
		\includegraphics[width=\textwidth]{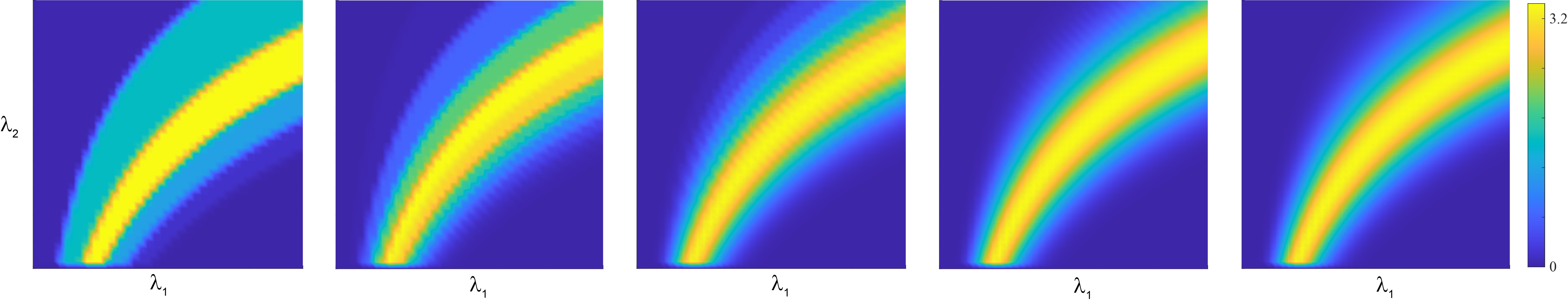}
		\caption{Heatmaps of the uniform prior posteriors for the exponential decay model for $T=2$ computed using partitions $\{I_i\}_{i=1}^{M_K}$ of $\mathcal{D}$ for $M_K \in \{12, 25, 50, 100, 200\}$ in order from left to right.}\label{expdecaySIPS3} 
	\end{figure}

\end{example}

\begin{example}\label{EX:expdecay5}
	We apply the accept-reject algorithm to Example~\ref{EX:expdecay} with the same discretization and computational parameters except as noted. We choose $M=100$ for a uniform partition of $\mathcal{D}=[0,1]$. We use $K=10^6$ samples from the trial generating distribution $P_\Lambda^{\mathrm{t}}$ to compute the synthetic data used to construct $\hat{\rho}_{\mathcal{D}}$.  We use  $N=40,000$ samples from uniform prior on $\Lambda$ to construct $	\widehat{\tilde{\rho}}_{\mathrm{p},\mathcal{D}}$. The algorithm rejected $29,317$ of these samples. In Fig.~\ref{AR-D-dist}, we show the estimates $\hat{\rho}_{\mathcal{D}}$ and $	\widehat{\tilde{\rho}}_{\mathrm{p},\mathcal{D}}$ and samples that are accepted by the algorithm.
	\begin{figure}[tbh]\centering
		\includegraphics[width=.7\textwidth]{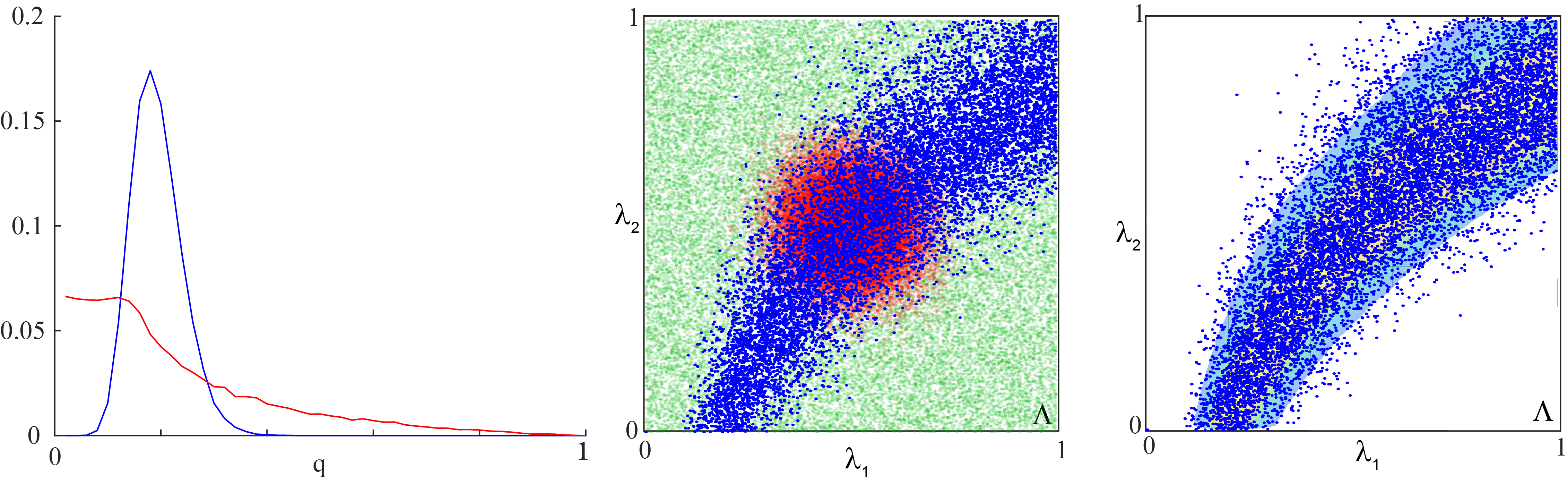}
		\caption{Left:  Plots of $\hat{\rho}_{\mathcal{D}}$ (blue) and $	\widehat{\tilde{\rho}}_{\mathrm{p},\mathcal{D}}$ (red). Center: Scatter plots of samples from the trial generating distribution (red), uniform prior (green), and the posterior (blue). Right: Scatter plot of uniform prior samples accepted by the accept-reject algorithm that are distributed according to the posterior  on top of the heatmap of the empirical posterior computed from re-weighting samples.}\label{AR-D-dist} 
	\end{figure}
\end{example}

\begin{example}\label{droppingballs}

We calibrate a computer model for falling objects. While the example is simple, it presents challenges involved with formulating and solving an SCP that generally arise in calibration of computer models. See \cite{Mosegard2002} for lengthy discussion about this example.

\paragraph*{{The experiment}}
The \textbf{\textit{standard acceleration of gravity}}  $g$ approximates the acceleration of an object falling near the surface of a planet. On Earth, the value of $g\approx 9.81$ but varies with location by as much as $.5\%$. In June 2013, author Derek Bingham and our colleague Dave Higdon performed an experiment aimed at estimating the value of $g$ in Vancouver British Columbia.

\paragraph*{{The experimental set up and field observations}}
They dropped a collection of balls from a spot on the Alex Fraser Bridge approximately $35 $m above the ground at sea level. The balls were launched horizontally to have a clear flight path, see Fig.~\ref{fig:ballexp}.  
\begin{figure}[htb]
	\begin{center}
		\includegraphics[width=4in]{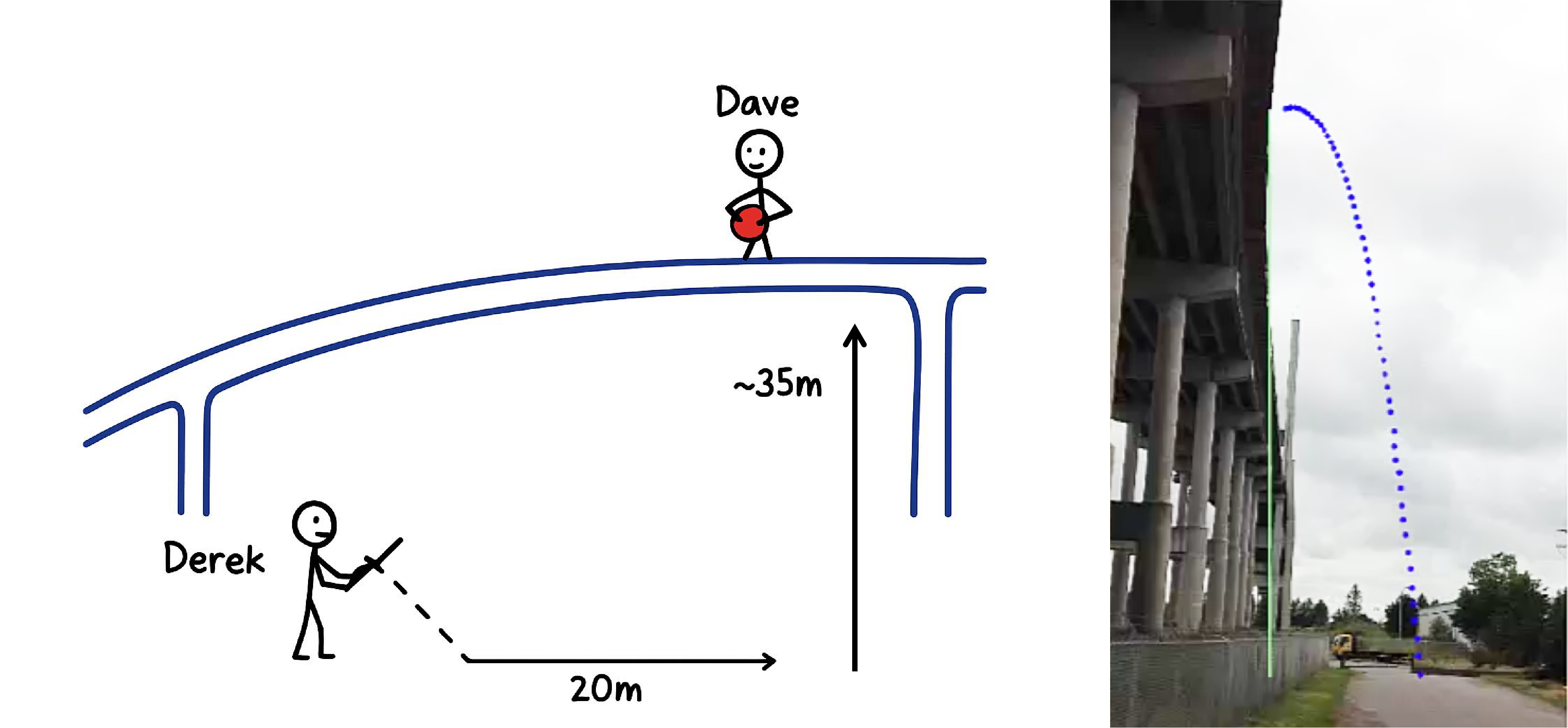}
		\caption{Left: Illustration of the ball dropping experiment. Right: A frame from a video capture of a trial.}
		\label{fig:ballexp}
	\end{center}
\end{figure}
Physical properties of the balls are presented in Table~\ref{tab:ballchar} in Sec.~\ref{droppingballsextra}.

Each ball was dropped 2-5 times. The motion was captured by video and the time of flight was computed using \textit{Logger Pro} software. The $17$ measured times in seconds are given in Table~\ref{tab:balldata} in Sec.~\ref{droppingballsextra}. The low number of data points is due to the intervention of local authorities during the experiment.

\paragraph*{{The model}}
If the flight time of a falling object is sufficiently long, air resistance has a significant effect. A high fidelity model is
\begin{equation}\label{fallingfull}
\frac{dv}{dt} = - g +\frac{1}{2} \rho C_d \frac{A}{m} v^2, 
\end{equation}
where $v$ is the \textit{vertical velocity} of the following object, $g$ is the \textit{gravitational constant}, $\rho$ is the \textit{air density}, $C_d$ is the \textit{coefficient of drag}, $m$ is the \textit{mass}, and $A$ is the \textit{cross-sectional area}. The second term on the right, which is the \textbf{\textit{drag term}}, models the effect of air resistance. The model \eqref{fallingfull} is accompanied by an \textit{initial height} $H_0$ and \textit{initial vertical velocity} $V_0$ which are uncertain since  they vary with each trial. Solving \eqref{fallingfull} yields the time of flight $T$ when the height is $0$.

Using the full model \eqref{fallingfull} has drawbacks. If we treat all the parameters as unknown, then we must estimate a posterior in a seven dimensional parameter space from a one dimensional observation. Since the generalized contours are six dimensional, the posterior will have significant indeterminacy. Alternatively, since  $\rho$, $C_d$, $m$, and $A$ can be measured independently, we could treat them as random effects. However, then we have to deal with the complexity of the SCP that includes random effects \cite{constrSIP}.

The drag term is small when the velocity is small and  wind tunnel data  indicates that the velocities of the balls remain within the linear growth regime\cite{hyperphysics}. It is reasonable to consider a simpler model that neglects drag,
\begin{equation}\label{fallingnodrag}
	\frac{dv}{dt} = - g,
\end{equation}
with solution $H(t) = -\frac{1}{2}  g t^2 + V_0 t + H_0, \; t \geq 0$, where $H(t)$ is the height at time $t$.  Solving for the impact time gives $
T = \frac{1}{g} \big( V_0 + \sqrt{V_0^2 + 2 g H_0}\big)$. There is a  single SCP for all of the balls collectively so we have $17$ data points. 

We pay for the simpler model in the choice of prior. Without knowledge of the uncertainty introduced by neglecting drag, we allow for a large degree of uncertainty in the parameters.  With $(H_0 \; V_0 \; g )^\top$  as the set of parameters for the SCP for the simple model, we choose $\Lambda = [27, 43]\times [-1,1] \times [8.8,10.8]$ centered on the nominal values $(35 \; 0 \; 9.8)^\top$.

\paragraph*{{Dealing with a small amount of data}}
The number of observations ($N\jmath=17$)  is too small to accurately estimate $\hat{p}_{K,M_K,i}$ directly.  We deal with this using simulation-based inference on the distribution of the field observations. We augment the measured values by adding i.i.d. noise to produce a sufficiently large dataset for the inversion. This approach amounts to jittering the observations many times to created a smoother empirical distribution function.  We obtain very similar results using a Beta distribution fit to the field observations.

Allowing for additional uncertainty beyond the minimum and maximum values of the data in Table~\ref{tab:balldata}, we define the interval $\mathcal{D} = [2.55, 3.19]$. We create a set of ``noisy'' output data by simulating $50,000$ iid samples from the  $\mathrm{Beta}(8,8)$ distribution shifted and scaled to the interval $[q_i-.35,q_i+.35]$ for each $q_i$ value in Table~\ref{tab:balldata} ($17\times 50,000 $ is on the order of $10^6$). 

We use $N=27\times 10^6$ samples in $\Lambda$ and $M=80$ partitions of $\mathcal{D}$ to compute the posterior for the uniform prior. We use a large numbers of samples for the estimator to reduce inaccuracies arising from finite Monte Carlo sampling. We obtain qualitatively similar results using samples on the order of hundreds.

\paragraph*{Calibrating the computer model using all data}
We show the posterior corresponding to creating ``noisy'' output data in Fig.~\ref{dropmarg2}.
\begin{figure}[htb]
	\begin{center}
		\includegraphics[width=5.25in]{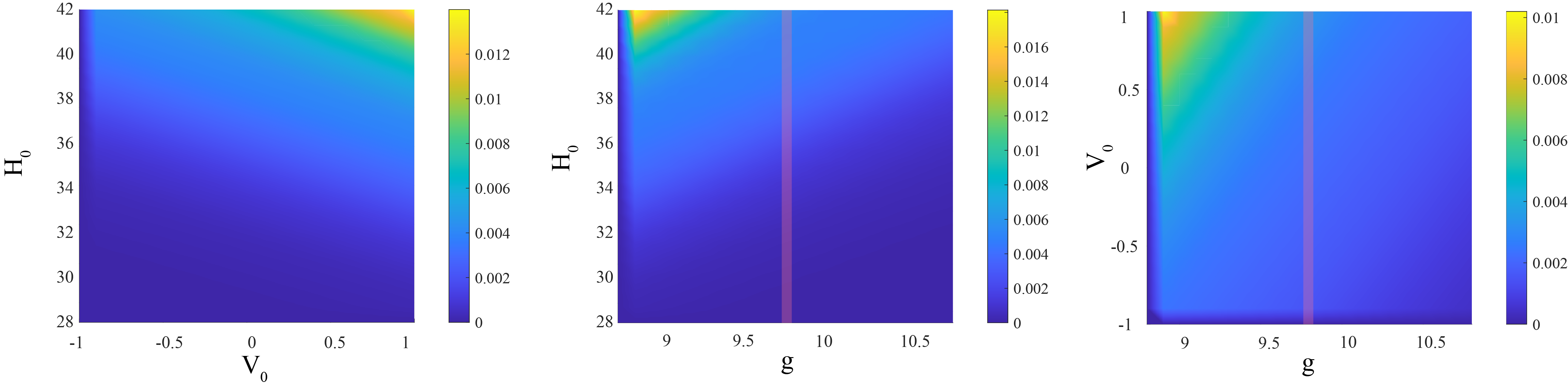}
		\caption{Heatmaps of marginal distributions for the posterior computed using noisy data. Left: $H_0$ vs $V_0$. Center: $H_0$ vs $g$. Right: $V_0$ vs $g$.  We indicate the part of $\Lambda$ intersecting the realistic range of uncertainty for $g$ near sea level as a transparent vertical rectangle in the second and third plots.}
		\label{dropmarg2}
	\end{center}
\end{figure}
This posterior places the highest probability near the corner point 
$(43 \; 1  \; 8.8)^\top$. The posterior  gives cause for concern. Realistically, the variation in the gravitational constant near sea level is  on the order of $\pm 0.02 \,\text{m/s}^2$. The posterior implies that values within $g\in [9.78,9.82]$ have relatively low probability.

\paragraph*{{The effect of removing outliers}}

One possible reason for the unsatisfactory results is neglecting drag. The fall times of the volleyball and the tennis ball are generally larger than the other types of balls. If we fix $H_0=35$ and $V_0=0$ and use the crude estimate $g \approx 2 H_0/T^2$, the volleyball and the tennis ball give significantly lower estimates than the other balls.  

We repeat the computations after removing their data. We generate ``noisy'' output data by adding $70,000$ iid samples from the $\mathrm{Beta}(8,8)$ distribution scaled and shifted to $[q_i-.35,q_i+.35]$ for each $q_i$ value. The other computational parameters remain the same. We show the posterior in Fig.~\ref{dropmarg2-no}.
\begin{figure}[htb]
	\begin{center}
		\includegraphics[width=5.25in]{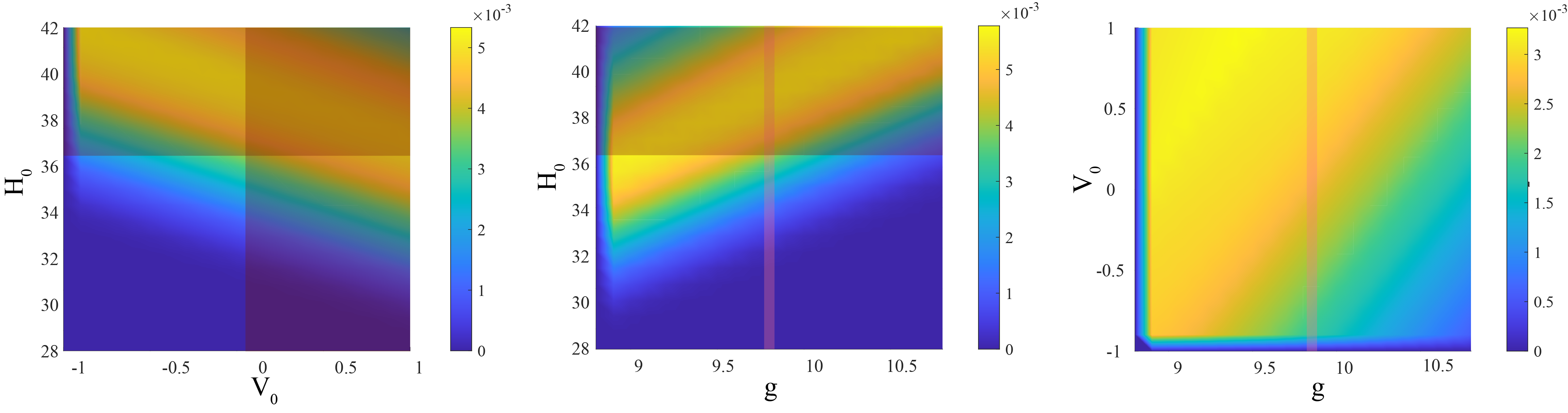}
		\caption{Heatmaps of marginal distributions for the posterior computed using noisy data. Left: $H_0$ vs $V_0$. Center: $H_0$ vs $g$. Right: $V_0$ vs $g$. We indicate the part of $\Lambda$ intersecting the realistic range of uncertainty for $g$ near sea level as a transparent vertical rectangle in the second and third plots. In the first and second plots, we indicate the intersection of $\Lambda$ with the highly probable range of $H_0$ given the realistic range of $g$ as a transparent horizontal rectangle. In the first plot, we indicate the intersection of $\Lambda$ with the highly probable range of $V_0$ as a transparent vertical rectangle.}
		\label{dropmarg2-no}
	\end{center}
\end{figure}

This posterior  is quite different. The structure in the densities imparted from the linear contours combined with the uniform prior is evident. A much larger set of possible parameter values produce computer model solutions consistent with the observed data with relatively high probability. 

In particular, there is a large range of $H_0$ and $V_0$ values that are consistent with $g$ close to $9.8$. We plot the event for $g \in [9.78,9.82]$ in Fig.~\ref{dropmarg2-no} as vertical rectangles. Using the intersection of that event with the region of relatively high probability determines ranges for $H_0$ of $[37.5,42.5]$ respectively $[36.75,43]$ of relatively high probability. This indicates that $V_0$ is in the range of $[0,1]$. The region of relatively high probability has the property that larger values of $H_0$ are associated with more negative values of $V_0$ and vice versa, which fits intuition. Unfortunately, there is a large degree of uncertainty due to the choice of $\Lambda$, prior, and variation in data for the different balls.

\paragraph*{Calibrating the computer model using bowling ball data and an informed prior}
Since the drag term in \eqref{fallingfull} is inversely proportional to mass, it appears that the flight of the bowling ball might best be described by the simple model \eqref{fallingnodrag} so we solve the SCP using only the bowling ball data.  We use a Beta distribution shifted and scaled to  $[q_i-.03,q_i+.03]$ to generate noise for each data point. We also choose an ``informed'' prior. Since we are dealing with a single type of ball, it is reasonable to believe that there is less variation in the experimental uncertainty in $H_0$ and $V_0$ during the trials. We  estimate those uncertainties to within a physically realistic level and choose a normal prior with independent marginals $N(35,.1)$ for $H_0$ and $N(0,.1)$ for $V_0$. Finally, we restrict the uncertainty in $g$ to a realistic level given we have strong prior knowledge about its value and choose $N(9.81,.01)$ as a prior for $g$. We use the same sample space as above since we do not know the uncertainty introduced by dropping the drag term. We use $2\times 10^{7}$ samples to generate the sample from the prior and fix all of the other computational parameters as above.

We show the posterior in Fig.~\ref{dropmarg2-BB}.  The mode of the posterior is near $(35.7 \; .5 \; 9.9)^\top$ and the event of relatively high probability places $g$ within $[9.78,9.95]$. 

\begin{figure}[htb]
	\begin{center}
		\includegraphics[width=5.25in]{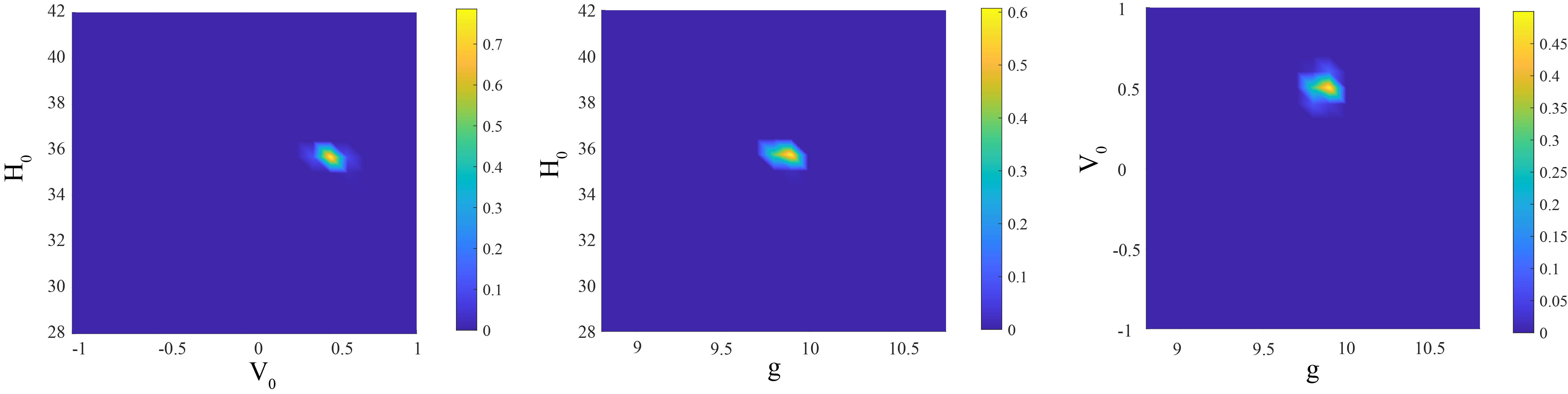}
		\caption{Heatmaps of marginal distributions for the posterior computed using noisy data for the bowling ball and an informed prior. Left: $H_0$ vs $V_0$. Center: $H_0$ vs $g$. Right: $V_0$ vs $g$.}
		\label{dropmarg2-BB}
	\end{center}
\end{figure}


\end{example}

\begin{example}\label{S:CM_Example}
	
	 We  use a slight variant of the accept-reject estimator of Sec.~\ref{sec:reject} to illustrate the performance of the SCP estimator on a model with a 15 dimensional parameter space and a 2 dimensional measurement vector.  We consider the MSEIR  epidemic model for disease dynamics \cite{Hethcote_2000}. In addition to the susceptible population (S), infected population (I), and the recovered population (R), the MSEIR model includes groups of infants protected by maternal antibodies (M) and groups of exposed and latent infected but not infectious (E). The process model $\mathcal{M}$ is,
	\begin{equation}\label{eq:MSEIRS}
		\begin{cases}
			\frac{dM}{dt} = B(S+E+I+R)-(\delta+\mu_M) M ,\\
			\frac{dS}{dt} = \delta M - \beta SI - (\mu_G+\iota) S  + fR ,\\
			\frac{dE}{dt} = \beta SI - (\epsilon + \mu_G)E, &  0 < t < T,\\
			\frac{dI}{dt} = \epsilon E - (\gamma+\mu_I+\mu_G)I,\\
			\frac{dR}{dt} = \gamma I - (\mu_G+f)R+\iota S,
		\end{cases}
	\end{equation}
	with initial conditions,
	\begin{equation}\label{eq:MSEIRSIC}
		M(0)=M_0, \quad S(0)=S_0, \quad E(0)=E_0, \quad I(0)=I_0, \quad R(0)=R_0.
	\end{equation}
	In this case, $Y=\big( M \; S \; E \; I \; R\big)^\top$ and 
	$ \lambda = \big( B \; \delta \; \mu_M \; \beta \; \mu_G \; 1/\epsilon \; \mu_I \; 1/\gamma \; f \; \iota \; M_0 \; S_0 \; E_0 \; I_0 \; R_0 \big)^\top \in \Lambda \subset \mathbb{R}^{15}$. Non-dimensionalizing the parameters and assuming a single unit of time is 1 week, a unit of population is $1E6$, and the birth/death rates are normalized to a population of size 300 million, we give the ranges for the characteristics in Table~\ref{tab:MSEIRS_params} that define $\Lambda$.
\begin{table}[h]
		\centering \footnotesize
		\begin{tabular}{| l | l | l l | | l | l | l |}
			\hline
			\multicolumn{7}{|c|}{Physical Interpretations and Ranges of Process Model Parameters in~\eqref{eq:MSEIRS}}\\
			\hline 
			 & Interpretation & Param.~Range & &  & Interpretation & Param.~Range \\
			\hline
			\hline
			B & birth rate & $[ 2.72E-4,3.04E-4 ]$  &  & $1/\epsilon$  &  infection time & $[0.571, 1]$   \\
			$\delta$ &  temp.~immunity & $[.0833,.25]$ & & $\mu_I$ &  infected death rate & $[4.81E-6 , 2.11E-5]$    \\
			$\mu_M$ & infant death rate & $[4E-3, 6E-3$  & &  $1/\gamma$ &  recovery time & $[0.7 , 2.33]$  \\
			$\beta$ & infectivity rate & $ [1.92E-3, 3.85E-3]$   & & $f$ &  immunity loss rate & $[0.125, 0.25]$ \\
			$\mu_G$ & general death rate & $[2.4E-4,2.72E-4]$ & & $\iota$ &  immunization rate & $[0.015,0.0375]$  \\
			\hline
		
			\multicolumn{7}{|c|}{Ranges of Initial Population Parameters in~\eqref{eq:MSEIRSIC}}\\
			\hline
			\multicolumn{7}{|c|}
			{\begin{tabular}{l l l l l l l}
			& $M_0$: $[2.5,3.5]$ & $S_0$: $[260,275]$  & $E_0$: $[0.01,0.5]$ & $I_0$: $[0.1,4]$ & $R_0$: $[10,20]$ &
			\end{tabular}} \\
			\hline
		\end{tabular}
		\caption{}\label{tab:MSEIRS_params}
	\end{table}
	We consider the observation vector $Q(\lambda)= \big( M(6)\; I(6)\big)^\top \in \mathcal{D} \subset \mathbb{R}^2$ with components equal to the numbers of immune infants and infected individuals at the sixth week respectively. 
	
	We compute synthetic data by assuming independent $\mathrm{Beta}(6,\beta)$ distributions for each parameter shifted to the  intervals listed in Table~\ref{tab:MSEIRS_params} with $\alpha = 6$ and $\beta$ is chosen to yield specified mean values in the intervals (see the supplemental files for details) and take $P^{\mathrm{t}}_\Lambda$ to be the product measure. We sample $5E2$ points from $P^{\mathrm{t}}_\Lambda$, solve the model to $T=6$ weeks numerically using  a variable step Runge-Kutta method, and compute $Q_1=M(6)$, $Q_2=I(6)$, and $\big( Q_1\; Q_2\big)^\top=\big( M(6)\; I(6)\big)^\top$. 
	
	Assuming a uniform prior, we calibrate using the synthetic data on $Q_1$, $Q_2$, and $\big(Q_1 \; Q_2\big)^\top$, respectively using $1E4$ uniform samples on $\Lambda$ to estimate the posterior. In view of Theorem~\ref{thm:gensampsolprop}, we estimate densities using KDE. The ratios for the rejection are used to compute a weighted KDE estimate of the posterior.

	We show the synthetic data and the estimated marginals of the posterior in Fig.~\ref{MSEIRfwd6}.  
	\begin{figure}[tbh]
	\centering
		\includegraphics[width=0.8\textwidth]{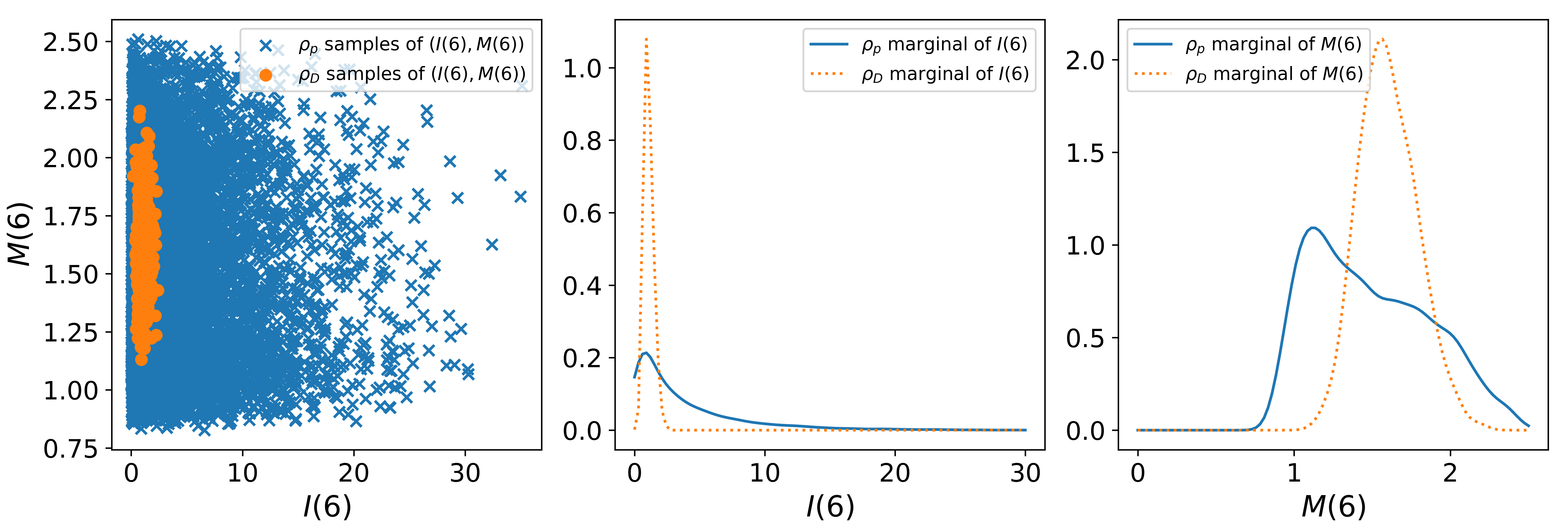}
		\caption{Left: scatter plot of values $Q(\lambda)= \big( M(6)\; I(6)\big)^\top$  computed from numerical solutions of \ref{eq:MSEIRSIC} for $5E2$ samples  of $\lambda$ distributed according to $P^{\mathrm{t}}_\Lambda$ (orange dots) and for $1E4$ samples of $\lambda$ distributed according to a uniform prior (blue x's). Center and right: KDE estimates of posterior marginals.
		}\label{MSEIRfwd6} 
	\end{figure}

	In Fig.~\ref{fig:gamma_delta_MSEIR}, we show $2$-dimensional marginals for $(\gamma^{-1},\delta)$ of the posterior using $I(6)$ data (left), $M(6)$ data (center), and the joint $(I(6)\;M(6))^\top$ data (right). We chose this marginal because it illustrates the impact of the 14-dimensional contour structures corresponding to the scalar data  and the 13-dimensional contour structure of the 2 dimensional data in the posterior respectively.  The code provided in the supplemental files can be used to display other marginals. However, many of the other marginals are approximately uniform because of the choice of uniform prior.
	
\begin{figure}[tbh]
	\begin{center}
		\includegraphics[width=0.3\textwidth]{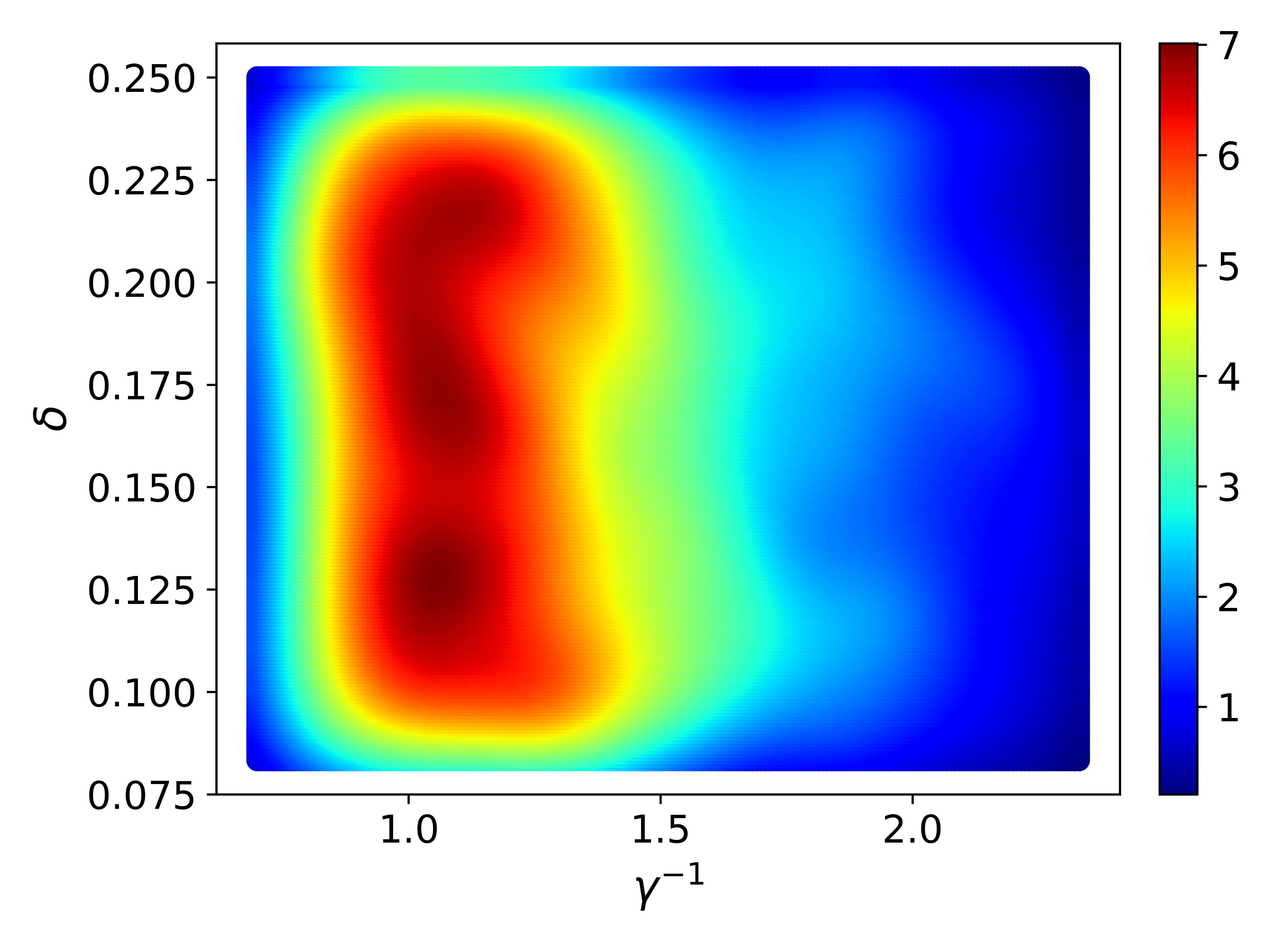}
		\includegraphics[width=0.3\textwidth]{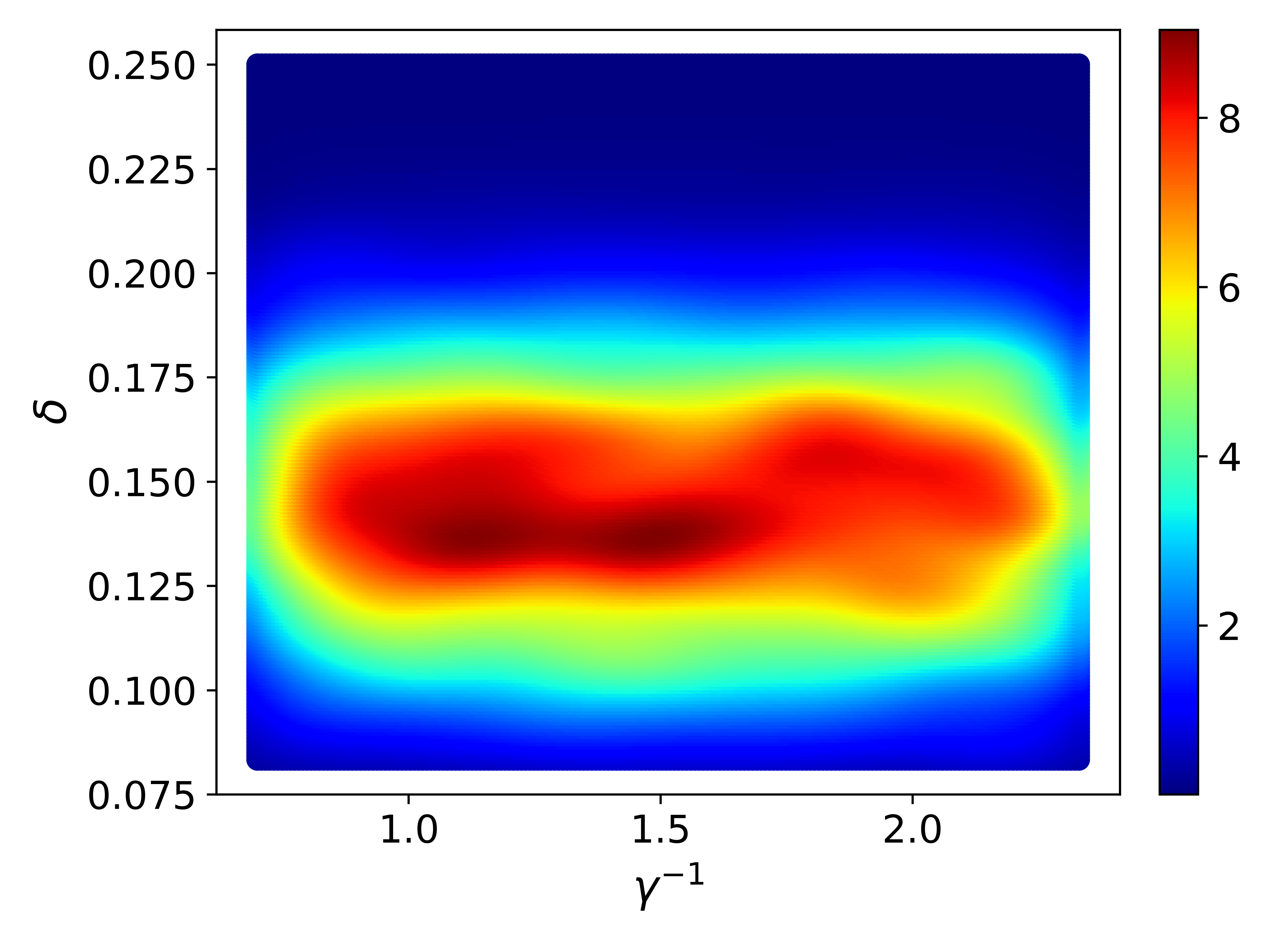}
		\includegraphics[width=0.3\textwidth]{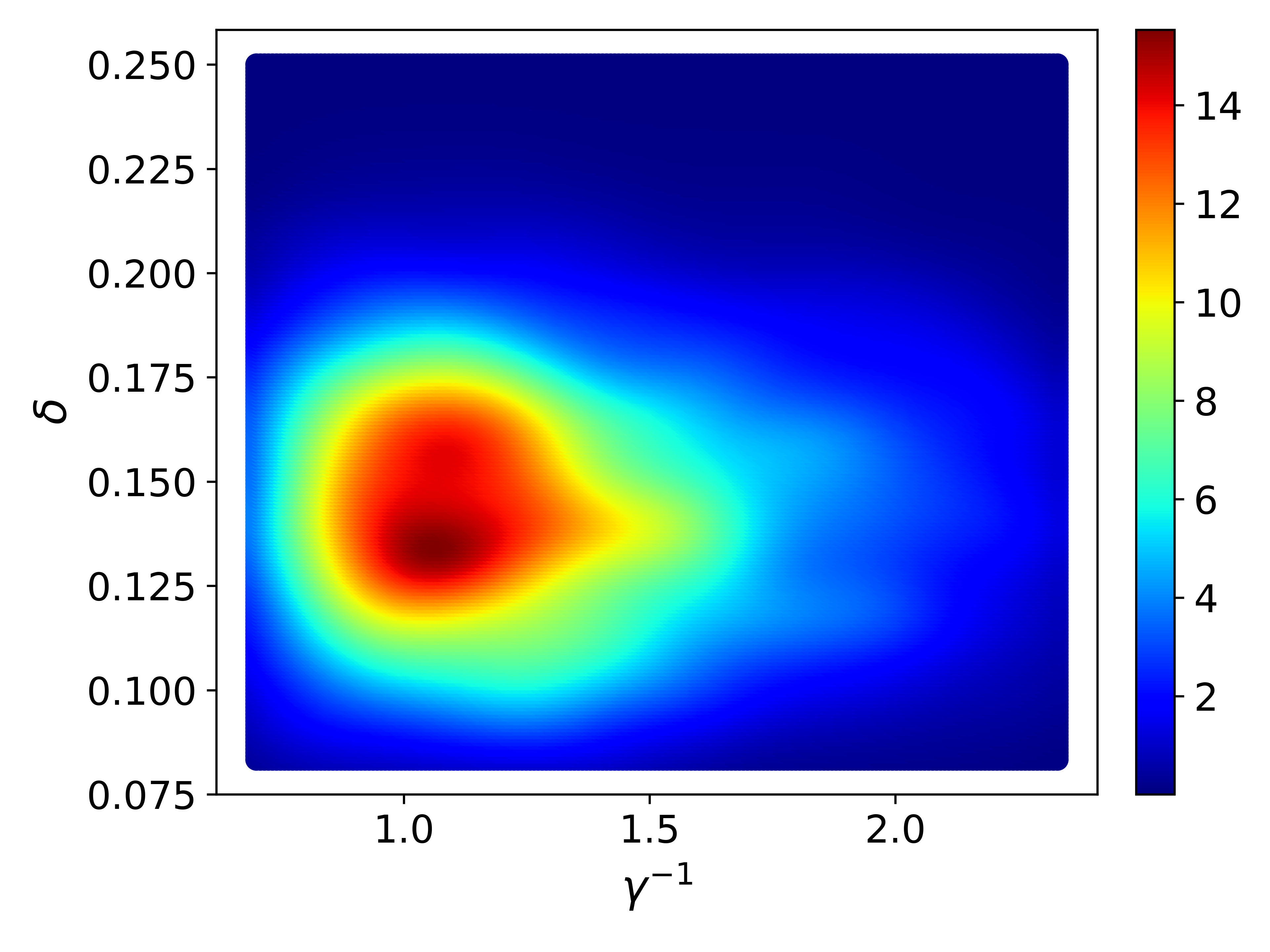}
			\caption{Marginal distributions of the posterior for $(\gamma^{-1},\delta)$ using data only on $Q_1$ associated with $I(6)$ (left), only on $Q_2$ associated with $M(6)$ (center), and the joint QoI data associated with $(I(6)\; M(6))^\top$ (right).}
			\label{fig:gamma_delta_MSEIR}
		\end{center}
\end{figure}

%
%
%
%
\end{example}

\section{Conclusion}\label{sec:conclusion}

 We formulate and analyze a nonparametric methodology for Bayesian computer model calibration to estimate a posterior on model parameters given a prior. We prove that the posterior exists using disintegration of measures and provide conditions under which the posterior can be represented by a conditional density. We prove that under general conditions, the posterior is continuous a.e. We also show that the posterior corresponding to the uniform prior has maximum entropy. A nonparametric estimator based on a form of importance sampling is constructed and analyzed along with an alternative accept–reject estimator. Finally, we illustrate the results with several examples.

Future work includes applying the aSCP framework to parametric Bayesian calibration to obtain a nonparametric hierarchical Bayesian methodology and investigating the links between nonparametric Bayesian calibration of computer models and traditional Bayesian statistics. The problem of asynchronous data collection for vector valued quantities of interest is also critically important. Ongoing efforts focus on extending this calibration approach to infinite-dimensional spaces and investigating the use of Markov Chain Monte Carlo methods and machine learning methods for estimating the posterior in high dimensional parameter spaces.

%

\bibliographystyle{siamplain.bst}
\bibliography{referencesmain,referencessuppbib}

\appendix

\section{Conceptual examples}\label{conceptexam}

We present two simple examples to illustrate the theoretical development.

\begin{example}\label{ex:condcircle} To illustrate disintegration, we adapt a simple example from \cite{PT_book} that models observing the distance to the origin of a point chosen at random in a disk. 	Let $\Lambda = \big\{(x_1,x_2)\in \mathbb{R}^2: ({x_1^2 + x_2^2})^{1/2}\leq 1\big\}$ and define $Q: (\Omega,\mathcal{B}_\Lambda) \to ([0,1],\mathcal{B}_{[0,1]})$ by $Q(x_1,x_2)=({x_1^2 + x_2^2})^{1/2}$.  The generalized contours are circles of radius $q$. We assume that $\Psi_\Lambda$ has density $\rho_\Lambda$ with respect to the Lebesgue measure.  There is a family $\{\Psi_N(\cdot|q)\}_{q\in [0,1]}$  concentrated on circles $\big\{(x_1,x_2): Q(x_1,x_2) =  q, x_1, x_2 \in \mathbb{R}\big\}$ such that 
	$$
	\Psi_\Lambda(A) = \int_{[0,1]} \Psi_N(A|q) \, d\Psi_{\mathcal{D}}(q), \quad A \in \mathcal{B}_\Lambda.
	$$
	We identify the components of the integral on the right by a simple matching of terms.  We first write an expression for the probability of an event $A$ using polar coordinates,
	\begin{equation*}
		\Psi_\Lambda(A) = \int_{[0,1]\times[0,2\pi]}\bigg( \chi_{A}(q,\theta) \rho_\Lambda(q \cos(\theta), q \sin(\theta)) \bigg) \, q \, d\mu_{\mathcal{L}}(\theta) \, d\mu_{\mathcal{L}}(q),
	\end{equation*}
	where $\mu_{\mathcal{L}}$ denotes the Lebesgue measure and $\chi_{A}$ denotes the \textbf{\textit{characteristic}} or \textbf{\textit{indicator function}} of $A$. We use Fubini's theorem to write,
	\begin{align}
		\Psi_\Lambda(A) &= \int_{[0,1]}\left(\frac{\int_{[0,2\pi]} \chi_{A}(q,\theta) \rho_\Lambda(q \cos(\theta), q \sin(\theta)) \, d\mu_{\mathcal{L}}(\theta)}{\int_{[0,2\pi]} \rho_\Lambda(q \cos(\theta), q \sin(\theta)) \, d\mu_{\mathcal{L}}(\theta)}\right) \nonumber\\
		&\qquad \qquad \qquad \times  \left( \int_{[0,2\pi]} \rho_\Lambda(q \cos(\theta), q \sin(\theta)) \, d\mu_{\mathcal{L}}(\theta)\right) \, q \, d\mu_{\mathcal{L}}(q),\nonumber
	\end{align}
	
	By identification, we recognize that,
	\[
	d \Psi_{\mathcal{D}}(q) = \left(\int_{[0,2\pi]} \rho_\Lambda (q \cos(\theta), q \sin(\theta)) \, d\mu_{\mathcal{L}}(\theta \right) \, q \,d\mu_{\mathcal{L}}(q).
	\]
	and
	\[
	\Psi_N (A|q) = \frac{\int_{[0,2\pi]} \chi_{A}(q,\theta) \rho_\Lambda(q \cos(\theta), q \sin(\theta)) \, d\mu_{\mathcal{L}}(\theta)}{\int_{[0,2\pi]} \rho_\Lambda(q \cos(\theta), q \sin(\theta)) \, d\mu_{\mathcal{L}}(\theta)},
	\]
	which has conditional density,
	\begin{equation}\label{ex:condcircleeq}
		\rho_{N}(\cdot|q) = \frac{\rho_\Lambda(q \cos(\theta), q \sin(\theta))(\theta)}{\int_{[0,2\pi]} \rho_\Lambda(q \cos(\theta), q \sin(\theta)) \, d\mu_{\mathcal{L}}(\theta)}.
	\end{equation}
	
\end{example}

\begin{example}\label{discretedisint}
	We present a discrete example adapted from \cite{JChi,PT_book} that provides a simplistic but intuitive illustration of the Bayesian solution of the SCP. We define the discrete probability space  $	\Lambda=\big\{(\lambda_1,\lambda_2) \in \{1,2,3\}\times\{1,2,3\}\big\}$, with the power set $\sigma-$algebra, and trial generating distribution,
	\[
	P^{\mathrm{t}}_\Lambda= \begin{bmatrix} 1/20 & 1/20 & 1/20 \\
		1/9  & 1/9 & 1/20\\
		5/12 &1/9 & 1/20 \end{bmatrix}.
	\]
	Setting $f(j) = \begin{cases} 0 , & j \text{ odd}, \\ 1, & j \text{ even},\end{cases}$, we define $Q(\lambda_1,\lambda_2) = f(\lambda_1+\lambda_2)$. $Q:\Lambda\to \mathcal{D}=\{0,1\}$, so the generalized contours are,
	\begin{equation*}
		Q^{-1}(0) = \{(1,2), (2,1), (2,3), (3,2)\}, \quad 
		Q^{-1}(1) = \{(1,1), (1,3), (2,2), (3,1), (3,3)\}.
	\end{equation*}
	The induced probability measure is $\begin{pmatrix} {P}_{\mathcal{D}}(0) \\ {P}_{\mathcal{D}}(1)\end{pmatrix} =
	\begin{pmatrix} 29/90 \\ 61/90\end{pmatrix} \approx \begin{pmatrix} .32\overline{2} \\.67\overline{7}\end{pmatrix}$.
	We now ``forget'' $P^{\mathrm{t}}_\Lambda$ and solve the SCP conditioned on a choice of prior.
	
	To obtain observational data for the SCP, we compute 200 draws from a Bernoulli random variable with probability $\frac{61}{90}$ of $1$ and use these to compute the estimate $
	\begin{pmatrix} \hat{P}_{\mathcal{D}}(0) \\ \hat{P}_{\mathcal{D}}(1)\end{pmatrix} =
	\begin{pmatrix} .34 \\ .66\end{pmatrix}$. We compute the posterior corresponding to $\hat{P}_{\mathcal{D}}$. In the absence of prior information, we assume a uniform prior,
	\[
	P_{\mathrm{p}} = \begin{pmatrix} 1/9 & 1/9 & 1/9 \\
		1/9 & 1/9 & 1/9\\
		1/9 &1/9 & 1/9 \end{pmatrix}.
	\]
	Since, $
	\frac{\frac{1}{9}}{\frac{1}{9}+\frac{1}{9}+\frac{1}{9}+\frac{1}{9}} = \frac{1}{4} $ and $
	\frac{\frac{1}{9}}{\frac{1}{9}+\frac{1}{9}+\frac{1}{9}+\frac{1}{9}+\frac{1}{9}} = \frac{1}{5}$,
	we obtain the conditional distributions,
	\begin{gather*}
		P_N(A|0) = \frac{1}{4}\chi_{A}(1,2) + \frac{1}{4}\chi_{A}(2,1) + \frac{1}{4}\chi_{A}(2,3)  + \frac{1}{4}\chi_{A}(3,2) ,\\
		P_N(A|1) = \frac{1}{5}\chi_{A}(1,1) + \frac{1}{5}\chi_{A}(1,3)  + \frac{1}{5}\chi_{A}(2,2) \ + \frac{1}{5}\chi_{A}(3,1) + \frac{1}{5}\chi_{A}(3,3) ,
	\end{gather*}
	where $A$ is any set in $\Lambda$ and $\chi_A$ is the indicator function for set $A$, as well as the estimated posterior,
	\begin{multline}\label{discdist}
		\hat{P}_\Lambda(A) = \bigg(\frac{1}{4}\chi_{A}(1,2) + \frac{1}{4}\chi_{A}(2,1) + \frac{1}{4}\chi_{A}(2,3)  + \frac{1}{4}\chi_{A}(3,2)\bigg)\,  .34\\
		+ \bigg( \frac{1}{5}\chi_{A}(1,1)  + \frac{1}{5}\chi_{A}(1,3)  + \frac{1}{5}\chi_{A}(2,2) \ + \frac{1}{5}\chi_{A}(3,1) + \frac{1}{5}\chi_{A}(3,3) \bigg) \, .66.
	\end{multline}
	\eqref{discdist} is a discrete version of disintegration. If we have prior information that the points near $(3,1)$ are more likely, we might choose the prior,
	\begin{equation}\label{disprior}
		\begin{pmatrix} 1/20 & 1/20 & 1/20 \\
			3/16 & 3/16 & 1/20\\
			3/16 &3/16 & 1/20 \end{pmatrix},
	\end{equation}
	yielding the estimate,
	\begin{multline*}
		\hat{P}_\Lambda(A) = \bigg(\frac{2}{19}\chi_{A}(1,2) + \frac{15}{38}\chi_{A}(2,1) + \frac{2}{19}\chi_{A}(2,3)  + \frac{15}{38}\chi_{A}(3,2)\bigg)\,  .34\\
		+ \bigg( \frac{2}{21}\chi_{A}(1,1)  + \frac{2}{21}\chi_{A}(1,3)  + \frac{5}{14}\chi_{A}(2,2) \ + \frac{5}{14}\chi_{A}(3,1) + \frac{2}{21}\chi_{A}(3,3) \bigg) \, .66,
	\end{multline*}
	which indeed assigns higher probability to the points near $(3,1)$.
\end{example}

\section{Proofs}\label{sec:proof}
\subsection{Proof of Theorem~\ref{thm:GCisMan}}\label{app:GCISMAN}
The result follows from the Implicit Function Theorem.
\begin{theorem}\label{thm:imp1}
	Let $f$ be a continuously differentiable map of $U\times V\subset \mathbb{R}^{n_1+n_2}\to\mathbb{R}^{n_1}$, where $U\subset \mathbb{R}^{n_1}$ and $V\subset \mathbb{R}^{n_2}$ are  open sets. Assume $f(\bar{x},\bar{y})=0$ for a point $(\bar{x},\bar{y})\in U\times V$. Set $A_x=J_{f,x}(\bar{x},\bar{y})$  and $A_y=J_{f,y}(\bar{x},\bar{y})$ and assume $A_x$ is invertible. 
	
	There exist open sets $\mathring{U}\subset U $ and $\mathring{V}\subset V$, with $(\bar{x},\bar{y})\in \mathring{U}\times \mathring{V}$, such that for every $y\in \mathring{V}$ there is a unique $x\in \mathring{U}$ with $ f(x,y)=0$. This defines a continuously differentiable map $x=g(y)$ of $\mathring{V}$ into $\mathring{U}$ such that $f(g(y),y)=0$  for $y\in \mathring{V}$ satisfying $\bar{x}=g(\bar{y})$ and $J_g(\bar{y})=-(A_x)^{-1}A_y$. 
\end{theorem}

A \textbf{\textit{regular parameterized manifold in $\mathbb{R}^n$}} is a  map $g:U\to \mathbb{R}^n$, where $U\subset \mathbb{R}^k$ is a non-empty open set, such that the $n\times k$ Jacobian matrix $J_{g}(x)$ has rank $k$ at all $x\in U$ \cite{sard1965hausdorff,morris1976differential, lee2012introduction}.  If a map $g:A\to B$, where $A\subset \mathbb{R}^k$ and $B\subset \mathbb{R}^n$,  is continuous, bijective and has a continuous inverse it is called a \textbf{\textit{homeomorphism}}. A regular parameterized manifold $g:U\to \mathbb{R}^n$, where $U\subset \mathbb{R}^k$ is a non-empty open set, that is a	homeomorphism is called an \textbf{\textit{embedded parameterized manifold}}. 

A $k$-dimensional manifold in $\mathbb{R}^n$ is locally a $k$-dimensional embedded parameterized manifold. Specifically, a \textbf{\textit{$k$-dimensional manifold}}   is a non-empty set $\mathcal{C}\subset\mathbb{R}^n$ such that each point $p\in \mathcal{C}$ there is an open neighborhood $N(p)$ of $p$ and an $k$-dimensional embedded parameterized manifold $g:U \to\mathbb{R}^n$ such that $g(U)=\mathcal{C}\cap N(p)$. Equivalently, a non-empty set $\mathcal{C}\subset\mathbb{R}^n$ is an $k$-dimensional manifold if and only there is an open neighborhood $N(p)$ of $p$, such that $\mathcal{C}\cap N(p)$ is the graph of a function $g$, where $n-k$ of the variables $x_1,\cdots,x_n$ are functions of the other $k$ variables.

The application of these concepts is straightforward given Assumptions \ref{Assump:1} and \ref{Assump:2}, and Theorem~\ref{thm:imp1} implies Theorem~\ref{thm:GCisMan}.

We recall that every continuously differentiable manifold is diffeomorphic to a continuously infinitely differentiable manifold, so we simply say that generalized contours are locally smooth $n-d$-dimensional manifolds. It further follows that $Q^{-1}(q)$ is locally approximated by the Jacobian of $Q^{-1}$ at each point, where the Jacobian is surjective a.e. A manifold with these properties is called a \textbf{\textit{submersion}}.

\subsection{Proof of Theorem~\ref{thm:sol}.}\label{sec:solproof}

Under Assumptions \ref{Assump:1} and \ref{Assump:2}, we may write the solution of the aSCP in terms of conditional densities with respect to the underlying Lebesgue measures. The proof of Theorem~\ref{thm:sol} depends on describing how $Q$ transforms the Lebesgue measure. We denote the Lebesgue measure on $(\Lambda,\mathcal{B}_\Lambda)$ by $\mu_{\Lambda}$ and the Lebesgue measure on $(\mathcal{D},\mathcal{B}_{\mathcal{D}})$ by $\mu_{\mathcal{D}}$.  We define the $Q$-induced measure $\tilde{\mu}_{\mathcal{D}}$ on $(\mathcal{D},\mathcal{B}_{\mathcal{D}})$ by $\tilde{\mu}_{\mathcal{D}}(A)=Q\mu_{\Lambda} =\mu_{\Lambda}(Q^{-1}(A))$ for $A\in \mathcal{B}_\mathcal{D}$.

The following result gives conditions that imply that $\tilde{\mu}_{\mathcal{D}}$ is absolutely continuous with respect to $\mu_{\mathcal{D}}$ with  density that involves a surface integral over each contour $Q^{-1}(y)$. 
\begin{theorem} \label{thm:mud}
	Assume Assumptions \ref{Assump:1} and \ref{Assump:2} hold. 
	Let $\nu_\Lambda$ be a measure on $(\Lambda,\mathcal{B}_\Lambda)$ that is absolutely continuous with respect to $\mu_\Lambda$ so $d\nu_\Lambda = f_\Lambda \, d\mu_\Lambda$ for a Lebesgue integrable function $f_\Lambda$.  Let $\tilde \nu_{\mathcal{D}} = Q \nu_\Lambda$. Then, $d\tilde{\nu}_{\mathcal{D}}=\tilde{f}_{\mathcal{D}}(q)\, d\mu_{\mathcal{D}}$, where 
	\begin{equation}\label{COVQgen}
		\tilde{f}_{\mathcal{D}}(q)=\int_{Q^{-1}(q)}f_\Lambda \, \frac{1}{\sqrt{\det (J_{Q}J_{Q}^\top) }}\, ds.
	\end{equation}
	
	In the case that $\nu_\Lambda=\mu_\Lambda$, then
	$d\tilde{\mu}_{\mathcal{D}}=\tilde{\rho}_{\mathcal{D}}(q)\, d\mu_{\mathcal{D}}$, where 
	\begin{equation}\label{COVQ}
		\tilde{\rho}_{\mathcal{D}}(q)=\int_{Q^{-1}(q)}\frac{1}{\sqrt{\det (J_{Q}J_{Q}^\top) }}\, ds.
	\end{equation}
	Moreover, $\tilde{\rho}_{\mathcal{D}}(q)>0$ for almost all $q\in \mathcal{D}$ with $Q^{-1}(q) \subset \mathrm{int}\,{\Lambda}$, where  $\mathrm{int}\,{\Lambda}$ is the interior of $\Lambda$. If $\mu_{\mathcal{D}}\big(Q(\Lambda)\setminus Q(\mathrm{int}\,\Lambda)\big)=0$, then $\mu_{\mathcal{D}}\big(\{ q\in \mathcal{D},\, \tilde{\rho}_{\mathcal{D}}(q)=0\}\big) = 0$. 
\end{theorem}

\begin{proof}
	Without loss of generality, we assume $f_\Lambda$ is nonnegative. 
	We start by showing \eqref{COVQ}. 	We prove that for any generalized rectangle $[a,b] \subset \mathbb{R}^m$ with  $b>a$, 
	\begin{equation*}
		\tilde{\mu}_{\mathcal{D}}([a,b])=\int_{[a,b]} \int_{Q^{-1}(q)}\frac{1}{\sqrt{\det (J_{Q}J_{Q}^\top) }}\, ds\,  d\mu_{\mathcal{D}}(q).
	\end{equation*}
	The  result follows from the Caratheodory Extension Theorem.
	
	Without loss of generality, we assume that $J_Q$ is full rank for every $\lambda \in Q^{-1}([a,b])$. This implies there are $m$ indices $i_1,i_2,\dots,i_m\in\{1,2,\dots,m\}$ such that $J_Q^{(m)}=\big(J_Q^{i_1}\,J_Q^{i_2}\,\cdots$ $J_Q^{i_m}\big)$ is invertible at $\lambda$, where $J_{Q}^i=i^{th}$ column of $J_Q$. Since $J_Q$ is continuous, there is an open ball $B_r(\lambda)$ of radius $r$ centered at $\lambda$, such that $J_Q^{(m)}$ is invertible in $B_r(\lambda)\cap Q^{-1}([a,b])$.
	
	$Q^{-1}([a,b])$ is a compact subset of $\Lambda$. We choose a finite collection $U_j = B_{r_j}(\lambda_j)$, $j = 1, \cdots, J$, such that $\displaystyle Q^{-1}([a,b])\subset \bigcup_{j=1}^J U_j$. We make the replacements $U_1\leftarrow U_1$ and, for $j>1$, $U_j \leftarrow U_j \setminus \overline{(\bigcup_{k=1}^{j-1} U_k)}$. We obtain a disjoint collection of open sets $\{U_j\}_{j=1}^J$  that covers $Q^{-1}([a,b])$ except possibly for a negligible set.
	
	In $U_j\cap Q^{-1}([a,b])$, $1\leq j \leq J$, we may write $(\lambda)\sim(\lambda^{(m_j)};\hat{\lambda}^{(m_j)})$ of which $\lambda^{(m_j)}$ are $m$ coordinates of $\lambda$ for which the Jacobian $J_Q^{(m_j)}$ is invertible, and $\hat{\lambda}^{(m_j)}$ the remaining $n-m$ entries with Jacobian $\hat{J}_Q^{(m_j)}$. We define $\pi_j:\mathbb{R}^n\to \mathbb{R}^{n-m}$ to be $\pi_j(\lambda)=\hat{\lambda}^{(m_j)}$,  $V_j=\{(Q(\lambda),\pi_j(\lambda)); \lambda \in U_j\cap Q^{-1}([a,b]) \}$, and $V_{jq}=\{\pi_j(\lambda); \lambda \in Q^{-1}(q)\cap U_j  \}$ for any $q\in [a,b]$.

	Temporarily,  we fix $j$ and drop the subscript on $M$. 	$V_j$ is diffeomorphic to $U_j \cap Q^{-1}([a,b])$ and 
	\begin{equation*}
		d\lambda^{(m)}\, d\hat{\lambda}^{(m)}= \left| \det \big(J_Q^{(m)}\big)^{-1} \right| \, d\mu_{\mathcal{D}}(q) \, d \hat{\lambda}^{(m)}.
	\end{equation*}
	For fixed $q$,  Theorem~\ref{thm:imp1} implies that for each $ U_j$, there is a continuously differentiable map $f$ such that $\lambda^{(m)}=f(\hat{\lambda}^{(m)})$. The map $f$ has Jacobian $J_f=\big(J_Q^{(m)}\big)^{-1}\hat{J}_Q^{(m)}$ so, under any suitable parameterization, the differential  form on the manifold $Q^{-1}(q)\cap  U_j$ can be written 
	\begin{equation*}
		ds =\sqrt{\det\big(\big(J_Q^{(m)}\big)^{-1}\hat{J}_Q^{(m)},I\big)\big(\big(J_Q^{(m)}\big)^{-1}\hat{J}_Q^{(m)},I\big)^\top}\,d\hat{\lambda}^{(m)}
		=k(\lambda)\, d\hat{\lambda}^{(m)},
	\end{equation*}
	where $
	k(\lambda) =\sqrt{\det(J_QJ_Q^\top)}\left|\, \det \big(J_Q^{(m)}\big)^{-1} \right|$. 
	Therefore
	\begin{align*}
		\tilde{\mu}_{\mathcal{D}}&([a,b]) =\mu_\Lambda(Q^{-1}([a,b])) =\sum_{j=1}^J \mu_\Lambda(U_j\cap Q^{-1}([a,b]))\\
		&=\sum_{j=1}^J \int_{U_j\cap Q^{-1}([a,b])} \, d\lambda_i^{(m_j)}\, d\hat{\lambda}_i^{(m_j)} =\sum_{j=1}^J \int_{V_j}\left| \det \big(J_Q^{(m_j)}\big)^{-1} \right|  \,  d \hat{\lambda}_i^{(m_j)}\, d\mu_{\mathcal{D}}(q)\\
		&=\int_{[a,b]}\, \sum_{j=1}^J\int_{V_{jq}} \left| \det \big(J_Q^{(m_j)}\big)^{-1} \right| \, d \hat{\lambda}_i^{(m)}\, d\mu_{\mathcal{D}}(q)\\
		&=\int_{[a,b]}\, \sum_{j=1}^J \int_{V_{jq}}  \frac{1}{\sqrt{\det(J_QJ_Q^\top)}}\, ds\, d\mu_{\mathcal{D}}(q)=\int_{[a,b]} \int_{Q^{-1}(q)}\frac{1}{\sqrt{\det (J_{Q}J_{Q}^\top) }}\, ds \, d\mu_{\mathcal{D}}(q).
	\end{align*}
	
	If $q\in \mathcal{D}$, and $Q^{-1}(q)\subset \mathrm{int}\,{\Lambda}$, there exists $\lambda\in \mathrm{int}\,{\Lambda}$ such that $Q(\lambda)=q$. So there exists a neighborhood $N(\lambda)$ of $\lambda$ with $N(\lambda)\subset \mathrm{int}\,{\Lambda}$. Since $Q^{-1}(q)\cap N(\lambda)$ is diffeomorphic to an open set in $\mathbb{R}^{n-m}$,
	\begin{equation*}
		\tilde{\rho}_{\mathcal{D}}(q)>\int_{Q^{-1}(y)\cap N(\lambda)}\frac{1}{\sqrt{\det (J_{Q}J_{Q}^\top) }}\, ds>0.
	\end{equation*}		
	This argument also shows that if $q\in Q(\Lambda)$ satisfies $ \tilde{\rho}_{\mathcal{D}}(q)=0$ then $q \in Q(\Lambda)\setminus Q(\mathrm{int}\, \Lambda)$.
	
	Next, we address the general case. First we assume  $f_\Lambda$ is a simple function, i.e., $f_{\Lambda}(\lambda) = \sum_{i=1}^{k} a_{i} \chi_{A_{i}}(\lambda) $ for positive numbers $\{a_i\}_{i=1}^k$ and  partition $\{A_i\}_{i=1}^k\subset\mathcal{B}_\Lambda$ of $\Lambda$. For any $B \in \mathcal{B}_{\mathcal{D}}$, we have
	\begin{eqnarray}
		\nu_{\Lambda}(Q^{-1}(B)) 
		&=& \sum_{i=1}^{k} a_{i} \mu_{\Lambda}(Q^{-1}(B)\cap A_{i})\nonumber \\
		&=&  \sum_{i=1}^{k}a_{i}\int_{B\cap Q(A_{i})} \int_{Q^{-1}(q)\cap A_{i} }\frac{1}{\sqrt{\det (J_{Q}J_{Q}^\top) }}\, ds \, d\mu_{\mathcal{D}}(q)    \nonumber \\
		&=& \sum_{i=1}^{k}a_{i}\int_{B} \int_{Q^{-1}(q)\cap A_{i} }\frac{1}{\sqrt{\det (J_{Q}J_{Q}^\top) }}\, ds \, d\mu_{\mathcal{D}}(q) \nonumber   \\
		&=&\int_{B} \int_{Q^{-1}(q)} \sum_{i=1}^{k} a_{i} \chi_{A_{i}}(\lambda)\frac{1}{\sqrt{\det (J_{Q}J_{Q}^\top) }}\, ds  \, d\mu_{\mathcal{D}}(q) \nonumber \\
		&=&  \int_{B} \int_{Q^{-1}(q)} f_{\Lambda}(\lambda)   \frac{1}{\sqrt{\det(J_{Q}{J_{Q}}^{\top})}} ds\, d\mu_{\mathcal{D}}(q). \nonumber
	\end{eqnarray} 
	
	We approximate a general nonnegative measurable function $f_\Lambda$ by a pointwise convergent monotone sequence of simple functions and use the Monotone Convergence Theorem to pass to a limit and obtain the result.
	
\end{proof}

We now are in position to prove the main result.

\begin{proof}[Proof of Theorem~\ref{thm:sol}]
	
	$P_{\mathrm{p}}$ defines an a.e. unique Ansatz prior $\{P_{\mathrm{p},N}(\cdot|q)\}_{q \in \mathcal{D}}$ via the disintegration,
	\begin{equation}\label{priordis}
		P_{\mathrm{p}}(A)=\int_{Q(A)}P_{\mathrm{p},N}(A|q)\, d \tilde{P}_{\mathrm{p},\mathcal{D}}(q)= \int_{Q(A)}\int_{Q^{-1}(q)\cap A}\, dP_{\mathrm{p},N}(\lambda|q)\, d \tilde{P}_{\mathrm{p},\mathcal{D}}(q)
	\end{equation}
	for all $A\in\mathcal{B}_\Lambda$, where $P_{\mathrm{p},N}(\cdot|q)$ is concentrated on $Q^{-1}(q)$.  We note that Assumption~\ref{Assump:shouldbe4} implies that $\tilde{\rho}_{\mathrm{p},\mathcal{D}}(q)>0$ for all $q \in \mathcal{D}$. Using \eqref{SIPAbsDis}, \eqref{mulamdist}, and Theorem~\ref{thm:mud},  we have
	\begin{align*}
		{P}_\Lambda(A)
		&=\int_{Q(A)}\int_{Q^{-1}(q)\cap A}\, dP_{\mathrm{p},N}(\lambda|q)\, dP_{\mathcal{D}}(q)=\int_{Q(A)}\int_{Q^{-1}(q)\cap A}\, dP_{\mathrm{p},N}(\lambda|q)\, \frac{\rho_{\mathcal{D}}(q)}{\tilde{\rho}_{\mathrm{p},\mathcal{D}}(q)} \, d\tilde{P}_{\mathrm{p},\mathcal{D}}(q)\\
		&=\int_A \frac{\rho_{\mathcal{D}}(Q(\lambda))}{\tilde{\rho}_{\mathrm{p},\mathcal{D}}(Q(\lambda))}\, dP_{\mathrm{p}} = \int_A \frac{\rho_{\mathcal{D}}(Q(\lambda))}{\tilde{\rho}_{\mathrm{p},\mathcal{D}}(Q(\lambda))}\,\rho_{\mathrm{p}} \, d \mu_\Lambda.
	\end{align*}
	The density $\frac{\rho_{\mathcal{D}}(Q(\lambda))\rho_{\mathrm{p}}(\lambda)}{\tilde{\rho}_{\mathrm{p},\mathcal{D}}(Q(\lambda))}$ is unique $\mu_\Lambda$ a.e.
	
\end{proof}

\subsection{Proof  of Theorem~\ref{entropy}}\label{sec:proofsunians}

\begin{proof}
	Let ${P}_\Lambda$ be the posterior corresponding any prior $P_{\mathrm{p}}$ and let $\overline{P}_\Lambda$ denote the posterior corresponding to the uniform prior. By Theorem~\ref{thm:sol}, these have respective densities,
	\begin{equation*}
		{\rho}_\Lambda(\lambda)=\frac{\rho_{\mathcal{D}}(Q(\lambda))}{\tilde{\rho}_{\mathrm{p},\mathcal{D}}(Q(\lambda))}\cdot \rho_{\mathrm{p}}(\lambda) ,
		\qquad \overline{\rho}_\Lambda(\lambda)=\frac{\rho_{\mathcal{D}}(Q(\lambda))}{\tilde{\rho}_{\mathcal{D}}(Q(\lambda))} ,
	\end{equation*}
	since $\overline{\rho}_\Lambda(\lambda)=0$ implies ${\rho}_\Lambda(\lambda)=0$. We compute
	\begin{align*}
		D\big({\rho}_\Lambda \|\overline{\rho}_\Lambda \big)& = 
		\int_{\Lambda} {\rho}_\Lambda(\lambda)
		\log \left(\frac{{\rho}_\Lambda(\lambda)}{\overline{\rho}_\Lambda(\lambda)}  \right) \, d\mu_{\Lambda}(\lambda)\\
		&= 
		\int_{\Lambda} {\rho}_\Lambda(\lambda)
		\log \left({\rho}_\Lambda(\lambda)  \right) \, d\mu_{\Lambda}(\lambda)
		- \int_{\Lambda} {\rho}_\Lambda(\lambda)
		\log \left(\frac{1}{\overline{\rho}_\Lambda(\lambda)}  \right) \, d\mu_{\Lambda}(\lambda)\\
		&= -H({\rho}_\Lambda) + H({\rho}_\Lambda,\overline{\rho}_\Lambda).
	\end{align*}
	
	By \eqref{COVQgen}, 
	
	\begin{align*}
		H({\rho}_\Lambda,\overline{\rho}_\Lambda) &= 
		\int_{\mathcal{D}} \int_{Q^{-1}(q)} {\rho}_\Lambda(\lambda)
		\log \left(\frac{1}{\overline{\rho}_\Lambda(\lambda)}  \right) 
		\frac{1}{\sqrt{\det(J_QJ_Q^\top)}}\, ds\, d \mu_{\mathcal{D}}(q) \\
		&= 
		\int_{\mathcal{D}} \int_{Q^{-1}(q)} \frac{\rho_{\mathcal{D}}(Q(\lambda))}{\tilde{\rho}_{\mathrm{p},\mathcal{D}}(Q(\lambda))} \rho_{\mathrm{p}}(\lambda) \cdot
		\log \left(\frac{\tilde{\rho}_{\mathcal{D}}(Q(\lambda))}{\rho_{\mathcal{D}}(Q(\lambda))} \right) 
		\frac{1}{\sqrt{\det(J_QJ_Q^\top)}}\, ds\, d \mu_{\mathcal{D}}(q) .
	\end{align*}
	Rearranging and using \eqref{COVQgen},
	\begin{align*}
		H({\rho}_\Lambda,\overline{\rho}_\Lambda) 
		&= 
		\int_{\mathcal{D}} \frac{\rho_{\mathcal{D}}(q)}{\tilde{\rho}_{\mathrm{p},\mathcal{D}}(q)}  
		\log \left(\frac{\tilde{\rho}_{\mathcal{D}}(q)}{\rho_{\mathcal{D}}(q)} \right) \left(
		\int_{Q^{-1}(q)}  \rho_{\mathrm{p}}(\lambda)
		\frac{1}{\sqrt{\det(J_QJ_Q^\top)}}\, ds\right)\, d \mu_{\mathcal{D}}(q) \\
		& = \int_{\mathcal{D}}  \rho_{\mathcal{D}}(q)  
		\log \left(\frac{\tilde{\rho}_{\mathcal{D}}(q)}{\rho_{\mathcal{D}}(q)} \right) \, d \mu_{\mathcal{D}}(q)  =
		\int_{\Lambda} \overline{\rho}_\Lambda(\lambda) \log \left( \frac{1}{\overline{\rho}_\Lambda(\lambda)}\right) \, d\mu_\Lambda(\lambda) = H(\overline{\rho}_\Lambda).
	\end{align*}
	We conclude that $
	H(\overline{\rho}_\Lambda) -H({\rho}_\Lambda)=	D\big({\rho}_\Lambda\| \overline{\rho}_\Lambda \big)  \geq 0$.
\end{proof}

\subsection{Proof  of Theorem~\ref{thm:conti}}\label{sec:proofconti}
For simplicity of notation, we first prove the result for the uniform prior. We treat the general case afterwards.

\begin{theorem}\label{lem:cont}
	Assume Assumptions \ref{Assump:1}, \ref{Assump:2}, and \ref{Assump:3} hold.  For $q_0\in\mathcal{D}$, assume that
	\begin{enumerate}
		\item \label{con1} There exists $\{ i_1,\dots, i_m\} \subset \{ 1,\dots, n\}$ such that $J_{Q}^{(m)}=\left(J_{Q}^{(i_1)},J_{Q}^{(i_2)},\dots,J_{Q}^{(i_m)}\right)$ is full rank on $Q^{-1}(q_0)$,
		\item \label{con2}$\big( 
		J_{Q} \;
		J_{B} 
		\big)^\top$ is full rank on $\partial \Lambda \cap Q^{-1}(q_0)$.
	\end{enumerate}
	Then, $\tilde\rho_\mathcal{D}(q)$ is continuous in a neighborhood of $q_0$.
\end{theorem}
\begin{proof}
	We first show the result for the uniform prior. 
	Without loss of generality, we assume $i_1=1, i_2=2, \dots, i_m = m$. We let $\lambda^{(m)}=(\lambda_{1},\dots, \lambda_{m})$, $\hat\lambda^{(m)}=(\lambda_{m+1},\dots, \lambda_{n})$, so $\lambda$ can be formally written as $\lambda\sim(\lambda^{(m)},\hat\lambda^{(m)})$. We define $\pi_{\hat{m}}:\mathbb{R}^n\to \mathbb{R}^{n-m}$ as $\pi_{\hat{m}}(\lambda)=\hat\lambda^{(m)}$.
	
	
	For a $\lambda_0\in \Lambda \cap Q^{-1}(q_0)$, since $J_{Q}^{(m)}(\lambda_0)$ is full rank, Theorem~\ref{thm:imp1} implies that $Q(\lambda)-q=0$ determines a unique continuous function $h:N(\hat\lambda_{0}^{(m)},r) \times N(q_0,r)\to N(\lambda_{0}^{(m)},r)$. Without loss of generality we assume $h$ is defined on $\overline{N(\hat\lambda_{0}^{(m)},r)} \times \overline{N(q_0,r)}$. Since the collection $\{N(\hat\lambda^{(m)},r)\times N(\lambda^{(m)},r): \lambda \in \Lambda\cap Q^{-1}(q_0)\}$ covers the compact set $\Lambda\cap Q^{-1}(q_0)$, there is a finite subcover $\{N(\hat\lambda_{k}^{(m)},r_k)\times N(\lambda_{k}^{(m)},r_k)\}, k=1,\dots, K$. Let $r_0=\min_{k} r_k$. There is a unique continuous function $h: \left(\bigcup_{k=1}^K \overline{N(\hat\lambda_{k}^{(m)},r_k)}\right)\times \overline{N(q_0,r_0)}\to \bigcup_{k=1}^K N(\lambda_{k}^{(m)},r_k)$.  We have
	\begin{align*}
		\tilde \rho_\mathcal{D}(q_0)=\int_{\pi_{\hat{m}}(\Lambda\cap Q^{-1}(q_0))}\frac{1}{\left|J_{Q}^{(m)}(h(\hat\lambda^{(m)},q_0),\hat\lambda^{(m)})\right|}\, d\hat\lambda^{(m)},
	\end{align*}
	where we use $|J_Q| $ to denote $|\det J_Q|$ to simplify notation.
	Since $J_{Q}^{(m)}(\lambda)$ is a continuous function, 
	$\left|J_{Q}^{(m)}(h(\hat\lambda^{(m)},q),\hat\lambda^{(m)})\right|^{-1}$
	is continuous on a compact domain $\left(\bigcup_{k=1}^K \overline{N(\hat\lambda_{k}^{(m)},r_k)}\right)\times \overline{N(q_0,r_0)}$, which means it is uniformly continuous and bounded. Let $M>0$ be bound so $\left|J_{Q}^{(m)}(h(\hat\lambda^{(m)},q),\hat\lambda^{(m)})\right|^{-1}< M$.
	
	Since $\partial \Lambda$ is piecewise smooth, for almost all $\lambda_0\in \partial \Lambda$, there exists $r$ such that $B:N(\lambda_0, r) \to \mathbb{R}$ is a continuously differentiable map and determines $\partial \Lambda\cap N(\lambda_0, r)$ by $B(\lambda)=0$.
	
	Take one such point $\lambda_0\in \partial \Lambda \cap Q^{-1}(q_0)$, then there exists $N(\lambda_0, r)$ such that $\partial \Lambda \cap Q^{-1}(q_0)$ is determined by $Q^*(\lambda)=\begin{pmatrix} 
		Q(\lambda)-q_0 \\
		B(\lambda)
	\end{pmatrix}=0$ 
	with
	\begin{align*}
		J_{Q^*}=\begin{pmatrix} 
			J_{Q} \\
			J_{B} 
		\end{pmatrix}=
		\begin{pmatrix} 
			J_{Q}^{(m)}& \hat J_{Q}^{(m)} \\
			J_{B}^{(m)}& \hat J_{B}^{(m)} 
		\end{pmatrix}
	\end{align*}
	having full rank because of Condition 2. 
	Since $J_{Q^*}(\lambda_0)$ and $J_{Q}^{(m)}(\lambda_0)$
	are full rank, there exists $j \in m+1,\dots, n$ such that 
	\begin{align*}
		\begin{pmatrix} 
			J_{Q}^{(m)} & (\hat J_{Q}^{(m)})_j \\
			J_{B}^{(m)} & (\hat J_{B}^{(m)})_j
		\end{pmatrix}(\lambda_0)
	\end{align*} is invertible. 
	
	We define $\lambda^{(m+1)}=(\lambda^{(m)},  \lambda_{j} ), \hat\lambda^{(m+1)}=\hat\lambda^{(m)}\setminus \lambda_{j}$, where we write $\lambda\sim (\lambda^{(m+1)},\hat\lambda^{(m+1)})$. By Theorem~\ref{thm:imp1}, there exists a unique continuous function,
	\[g:N(\hat\lambda_{0}^{(m+1)},r') \times N(q_0,r')\to N(\lambda_{0}^{(m)},r')\times N(\lambda_{0j},r'),
	\] 
	which satisfies  $
	Q^*(g(\hat\lambda^{(m+1)},y),\hat\lambda^{(m+1)})=0$. 
	Without loss of generality, we assume $g$ is defined on $\overline{N(\hat\lambda_{0}^{(m+1)},r')} \times \overline{N(q_0,r')}$. By compactness of the domain, $g$ is uniformly continuous.
	The collection of sets $\{N(\hat\lambda^{(m+1)},r') \times N(\lambda^{(m)},r')\times N(\lambda_{j},r'): \lambda \in \partial \Lambda \cap Q^{-1}(q_0)\}$ covers $\partial \Lambda \cap Q^{-1}(y)$ for $y\in N(q_0,r')$. We assume $r'$ is sufficiently small so that 
	\[N(\hat\lambda^{(m+1)},r') \times N(\lambda^{(m)},r')\times N(\lambda_{j},r')\subset \bigcup_{k=1}^K\left(N(\hat\lambda_{k}^{(m)},r_k)\times N(\lambda_{k}^{(m)},r_k)\right).\]
	
	Since $\partial \Lambda \cap Q^{-1}(q_0)$ is compact, there is a finite subcover  $\{N(\hat\lambda_{\ell}^{(m+1)},r_\ell') \times N(\lambda_{\ell}^{(m)},r_\ell')\times N(\lambda_{\ell j},r_\ell')\}$, $\ell=1,\dots, L$. Let $r_0'=\min\{r'_\ell,\ell=1,\cdots,L; r_0\}$ then for any $q \in N(q_0,r_0')$, the sets also cover $\partial \Lambda \cap Q^{-1}(q)$. So Condition 2 holds for all $q \in N(q_0,r_0')$. Thus,
	\[
	\bigcup_{\ell=1}^L \left(N(\hat\lambda_{\ell}^{(m+1)},r_\ell') \times N(\lambda_{\ell}^{(m)},r_\ell')\times N(\lambda_{\ell j},r_\ell')\times N(q_0,r_0')\}\right)
	\] 
	is covered by $\{ N(\hat\lambda_{k}^{(m)},r_k)\times N(\lambda_{k}^{(m)},r_k)\times N(q_0, r_0): k = 1, \cdots, K\}$.
	
	As a result, $\{ N(\hat\lambda_{k}^{(m)},r_k)\times N(\lambda_{k}^{(m)},r_k)\}$ covers $\Lambda\cap Q^{-1}(q)$ for all $q\in N(q_0,r_0')$. Thus,
	\begin{align*}
		\tilde \rho_\mathcal{D}(q)=\int_{\pi_{\hat{m}}(\Lambda\cap Q^{-1}(q))}\frac{1}{\left|J_{Q}^{(m)}(h(\hat\lambda^{(m)},q),\hat\lambda^{(m)})\right|}\, d\hat\lambda^{(m)}.
	\end{align*}
	For $q \in N(q_0,r'_0)$, let $C=\pi_{\hat{m}}(\Lambda\cap Q^{-1}(q_0)) \triangle \pi_{\hat{m}}(\Lambda\cap Q^{-1}(q))$. Thus,
	\begin{multline}\label{refthisone}
		|\tilde \rho_\mathcal{D}(q_0)-\tilde \rho_\mathcal{D}(q)| \\
		\leq \int\limits_{\pi_{\hat{m}}(\Lambda\cap Q^{-1}(q_0))}\left|\frac{1}{\left|J_{Q}^{(m)}(h(\hat\lambda^{(m)},q_0),\hat\lambda^{(m)})\right|}-\frac{1}{\left|J_{Q}^{(m)}(h(\hat\lambda^{(m)},q),\hat\lambda^{(m)})\right|}\right|\, d\hat\lambda^{(m)} +\int_C M \,d\hat\lambda^{(m)}. 
	\end{multline}
	For the first term, uniform continuity implies that given $\epsilon>0$, there exists $r_\epsilon<r_0'$ such that for $q\in N(q_0,r_\epsilon)$, and $\hat \lambda^{(m)}\in \left(\bigcup_{k=1}^K N(\hat\lambda_{k}^{(m)},r_k)\right)$,
	\begin{align*}
		\left|\frac{1}{\left|J_{Q}^{(m)}(h(\hat\lambda^{(m)},q_0),\hat\lambda^{(m)})\right|}-\frac{1}{\left|J_{Q}^{(m)}(h(\hat\lambda^{(m)},q),\hat\lambda^{(m)})\right|}\right|< \epsilon.
	\end{align*}
	Thus, 
	\begin{align*}
		\int_{\pi_{\hat{m}}(\Lambda\cap Q^{-1}(q_0))}\left|\frac{1}{\left|J_{Q}^{(m)}(h(\hat\lambda^{(m)},q_0),\hat\lambda^{(m)})\right|}-\frac{1}{\left|J_{Q}^{(m)}(h(\hat\lambda^{(m)},q),\hat\lambda^{(m)})\right|}\right|\, d\hat\lambda^{(m)}\\
		\leq \hat\mu^{(m)}(\pi_{\hat{m}}(\Lambda\cap Q^{-1}(q_0)))\; \epsilon
	\end{align*}
	For the second term, we need $\hat\mu^{(m)}(C)$ to be small when $q$ is close to $q_0$.  Since $\{\pi_{\hat{m}}\big(N(\hat\lambda_{\ell}^{(m+1)},$ $r_\ell') \times N(\lambda_{\ell}^{(m)},r_\ell')\times N(\lambda_{\ell j},r_\ell')\big)=N(\hat\lambda_{\ell}^{(m+1)},r_\ell') \times N(\lambda_{\ell j},r_\ell')\}$
	covers $\pi_{\hat{m}}(\partial \Lambda\cap Q^{-1}(q))$ when $q\in N(q_0,r_0')$, then
	\begin{align*}
		&\hat\mu^{(m)}(C)=\hat\mu^{(m)}(\pi_{\hat{m}}(\Lambda\cap Q^{-1}(q_0)) \triangle \pi_{\hat{m}}(\Lambda\cap Q^{-1}(q)))\\
		&\quad \leq\sum_{\ell=1}^L \hat\mu^{(m)}\left(\left(\pi_{\hat{m}}\left(\Lambda\cap Q^{-1}(q_0)\right) \triangle \pi_{\hat{m}}\left(\Lambda\cap Q^{-1}(q)\right)\right)\cap \left(N(\hat\lambda_{\ell}^{(m+1)},r_\ell') \times N(\lambda_{\ell j},r_\ell')\right)\right),
	\end{align*}
	and
	\begin{align*}
		\hat\mu^{(m)}\left(\left(\pi_{\hat{m}}\left(\Lambda\cap Q^{-1}(q_0)\right) \triangle \pi_{\hat{m}}\left(\Lambda\cap Q^{-1}(q)\right)\right)\cap \left(N(\hat\lambda_{\ell}^{(m+1)},r_\ell') \times N(\lambda_{\ell j},r_\ell')\right)\right)\\\leq\int_{N(\hat\lambda_{\ell}^{(m+1)},r_\ell')}\left|g(\hat\lambda^{(m+1)},q_0)-g(\hat\lambda^{(m+1)},q)\right|\, d\hat\lambda^{(m+1)}.
	\end{align*}
	By the uniform continuity of $g$, there exists an $r_\epsilon'\leq r_0'$ such that for all $q\in N(q_0,r_\epsilon')$, $\hat\mu^{(m)}(B)<\epsilon$.
	Finally, we have for all $q\in N(q_0,r)$, $r=\min(r_\epsilon,r_\epsilon')$,
	\begin{align*}
		|&\tilde \rho_\mathcal{D}(q_0) -\tilde \rho_\mathcal{D}(q)|\\
		&\leq \int\limits_{\pi_{\hat{m}}(\Lambda\cap Q^{-1}(q_0))}\left|\frac{1}{\left|J_{Q}^{(m)}(h(\hat\lambda^{(m)},q_0),\hat\lambda^{(m)})\right|}-\frac{1}{\left|J_{Q}^{(m)}(h(\hat\lambda^{(m)},q),\hat\lambda^{(m)})\right|}\right|\, d\hat\lambda^{(m)} +\int_C M\, d\hat\lambda^{(m)}\\
		&\leq \hat\mu^{(m)}(\pi_{\hat{m}}(\Lambda\cap Q^{-1}(q_0)))\, \epsilon+ M\epsilon.
	\end{align*}
	So $\tilde \rho_\mathcal{D}(q)$ is continuous at $q_0$. Furthermore, the two conditions hold for all $q$ in $N(q_0,r_0')$, so $\tilde \rho_\mathcal{D}(q)$ is continuous in $N(q_0,r)$ as well.
	
	Next, we treat a general prior. Arguing as for \eqref{refthisone}, we have
	\begin{align*}
		|&\tilde \rho_{p,\mathcal{D}}(q_0)-\tilde \rho_{p,\mathcal{D}}(q)|\\
		&\leq \int\limits_{\pi_{\hat{m}}(\Lambda\cap Q^{-1}(q_0))}\left|\frac{\rho_p(h(\hat\lambda^{(m)},q_0),\hat\lambda^{(m)})}{\left| J_{Q}^{(m)}(h(\hat\lambda^{(m)},q_0),\hat\lambda^{(m)})\right|}-\frac{\rho_p(h(\hat\lambda^{(m)},q),\hat\lambda^{(m)})}{\left| J_{Q}^{(m)}(h(\hat\lambda^{(m)},q),\hat\lambda^{(m)})\right|}\right|\, d\hat\lambda^{(m)} +\int_C M_2 \,d\hat\lambda^{(m)},
	\end{align*}
	where $M_2$ is a constant such that,
	\[
	\left|\frac{\rho_p(h(\hat\lambda^{(m)},q),\hat\lambda^{(m)})}{\left| J_{Q}^{(m)}(h(\hat\lambda^{(m)},q),\hat\lambda^{(m)})\right|}\right| < M_2 
	\textrm{ on }
	\bigcup_{k=1}^K\left(N(\hat\lambda_{k}^{(m)},r_k)\right)\times N(q_0,r_0).
	\]
	Using the same argument as for the case of a uniform prior, there exists a $\delta > 0$ such that $|q-q_0| < \delta$ implies $|\tilde \rho_{p,\mathcal{D}}(q_0)-\tilde \rho_{p,\mathcal{D}}(q)| < \epsilon$.

\end{proof}

We remove the requirement that $Q^{-1}(q_0)$ is uniformly parameterized to prove Theorem~\ref{thm:conti}.

\begin{proof}[Proof of Theorem~\ref{thm:conti}]
	For all $\lambda \in Q^{-1}(q_0)$, there exists a $N(\lambda, r)$ such that there exists $\{ i_1$, $\dots$, $i_m\}$ $\subset$ $\{ 1,\dots, n\}$ with $J_{Q}^{(m)}=\left(J_{Q}^{(i_1)}\, J_{Q}^{(i_2)}\,\cdots J_{Q}^{(i_m)}\right)$ being full rank on $Q^{-1}(q_0)$ in $N(\lambda, r)$. The boundary of $N(\lambda,r)\cap \Lambda$, which is the union of part of $\partial\Lambda$ and part of $\partial N(\lambda,r)$ is piecewise smooth. Therefore, there is a continuously differentiable map $B^\prime$ such that $\partial \left(N(\lambda,r)\cap \Lambda\right)$ is locally determined by $B'(\lambda)=0$. When $r$ is sufficiently small, $\big( J_{Q} $ $J_{B'} 
	\big)^\top$  is full rank on $\partial \left(N(\lambda,r)\cap \Lambda\right) \cap Q^{-1}(q_0)$ by continuity of $J_Q$.  The collection $\{N(\lambda,r)\}$ admits a finite sub cover $\{N(\lambda_k,r_k)\}_{k=1}^K$ of $\Lambda \cap Q^{-1}(q_0)$. Taking $r$ sufficiently small, the collection covers $\Lambda \cap Q^{-1}(q)$, for $q$ in $N(q_0,r)$. Define $N_k=N(\lambda_k,r_k)\setminus \left(\bigcup_{j=1}^{k-1} N(\lambda_j,r_j) \right)$, the boundary of $N_k\cap \Lambda$ is piecewise smooth. We assume there exists $B_k(\lambda)$ such that 
	$\partial \left(N_k\cap \Lambda\right)$ is  determined locally by $B_k(\lambda)=0$, $k = 1, \cdots, K$. Then, $\big(  J_{Q} \;J_{B_k} 	\big)^\top$ is full rank on $\partial \left(N_k\cap \Lambda\right)$ and
	\begin{align*}
		\tilde\rho_{p,\mathcal{D}}(q)&=\int_{Q^{-1}(q)\cap \Lambda}\frac{\rho_p}{\sqrt{\det (J_{Q}J_{Q}^\top) }}\, ds=\sum_{k=1}^K \int_{Q^{-1}(q)\cap \left(N_k\cap \Lambda\right)} \frac{\rho_p}{\sqrt{\det (J_{Q}J_{Q}^\top) }}\, ds.
	\end{align*}
	So for all $q \in N(q_0,r)$,
	\begin{multline*}
		|\tilde \rho_{p,\mathcal{D}}(q)-\tilde \rho_{p,\mathcal{D}}(q_0)|\\
		\leq \sum_{k=1}^K\left|\int_{Q^{-1}(q)\cap \left(N_k\cap \Lambda\right)} \frac{\rho_p}{\sqrt{\det (J_{Q}J_{Q}^\top) }}\, ds-\int_{Q^{-1}(q_0)\cap \left(N_k\cap \Lambda\right)} \frac{\rho_p}{\sqrt{\det (J_{Q}J_{Q}^\top) }}\, ds\right|
	\end{multline*}
	Theorem~\ref{lem:cont} implies that $\tilde \rho_\mathcal{D}(q)$ is continuous at $q_0$. Since $\int_{Q^{-1}(q)\cap \left(N_k\cap \Lambda\right)} \frac{\rho_p}{\sqrt{\det (J_{Q}J_{Q}^\top) }}\, ds$ is continuous in a neighborhood of $q_0$, $\tilde \rho_{p,\mathcal{D}}(q)$ is continuous in a neighborhood of $q_0$.
\end{proof}

\subsection{Background on random samples and approximation of sets}\label{sec:randprop}

The estimation meth\-odology is heavily influenced by the theory of stochastic geometry and the approximation of sets \cite{Stoyan1999,khmaladze2001almost,penrose2007}. We apply this material to both $\Lambda$ and $\mathcal{D}$, so for efficiency of presentation, we consider a measurable set $\Omega\subset\mathbb{R}^k$, $k\geq 1$, with measurable boundary of Lebesgue measure zero $\mu_{\mathcal{L}}(\partial \Omega) = 0$ and Borel $\sigma-$algebra $\mathcal{B}_\Omega$. We let $\mu_\Omega$ denote the Lebesgue measure restricted to $\Omega$. A \textbf{\textit{partition}} of a measurable set $\Omega$ is a finite or countable collection of  non-intersecting measurable subsets $\{B_i\}_{i}$ such that  $\Omega= \bigcup_{i} B_i$, except possibly for set of measure $0$.  A collection points $\{\omega_j\}_{j=1}^M\subset \Omega$ defines a \textbf{\textit{Voronoi tessellation}} $\mathcal{V}_M = \{V_j\}_{j=1}^M$ that partitions $\Omega$, where 
${V}_j =  \{\omega\in {\Omega} \, : \, d_v(\omega_j,\omega) \leq d_v(\omega_i,\omega), \; i=1,\dots,M\}$, 
for a specified   metric $d_v(\cdot,\cdot)$ on ${\Omega}$.  For a set $A \in \mathcal{B}_\Omega$ with $\mu_\Omega(\partial A)=0$, the \textbf{\textit{Voronoi coverage}} of $A$ is defined $A_M=\bigcup_{\omega_j\in A, 1\leq j\leq M} V_j$.

A key property for approximation of sets by sequences of Voronoi tessellations is that the maximum inter-cell distance of cells in the sequence of Voronoi tessellations goes to zero a.s. A rule for defining a sequence of samples $ \big\{\{\omega^{(\ell)}_j\}_{j=1}^{M_\ell}\big\}_{\ell = 1}^\infty$ with associated Voronoi tessellations $\{\mathcal{V}_{M_\ell}\}_{\ell = 1}^\infty$, where $0<M_1 < M_2 < \cdots$, is \textbf{\textit{$\mathcal{B}_{\Omega}$-consistent}} if
\[
r_{M_\ell} = \max_{1\leq j\leq M_\ell} r\big(\omega^{(\ell)}_j\big)=\max_{1\leq j\leq M_\ell} \max_{\omega\in V^{(\ell)}_j} d_v\big(\omega,\omega^{(\ell)}_j\big)\to 0\; \mathrm{ a.s.\   as }\;\ell \to\infty.
\]

The following is a consequence of the results in \cite{penrose2007}.
\begin{theorem}\label{lem:SLLNS}
	Let $ \big\{\{\omega^{(\ell)}_j\}_{j=1}^{M_\ell}\big\}_{\ell = 1}^\infty$ denote a sequence of random samples generated by a Poisson point process corresponding to a probability measure that has an  a.e. positive density function with respect to the Lebesgue measure. If $A\in\mathcal{B}_{{\Omega}}$ satisfies $\mu_{{\Omega}}(\partial A) = 0$, then
	\begin{equation}\label{setappconv}
		\lim_{\ell \to \infty} \mu_{{\Omega}}(A_{M_\ell} \triangle A) = 0\; \mathrm{ a.e.}
	\end{equation}
\end{theorem}

For an integrable non-negative function $\rho_\Omega$ defined on $(\Omega, \mathcal{B}_\Omega, \mu_\Omega)$, the formal approximation tool is the simple function  defined  with respect to a given measurable disjoint partition $\mathcal{I}=\{I_i\}_{i=1}^M\subset \mathcal{B}_\Omega$ of $\Omega$,
\begin{equation}\label{bassimpfun}
	\rho_{\Omega,M}(\omega)=\sum_{i=1}^M p_i \chi_{I_i}(\omega),\quad
	p_i=\frac{\int_{I_i}\rho_\Omega d\mu_\Omega}{\int_{I_i}d\mu_\Omega}.
\end{equation} 
In practice, we employ estimates of the coefficients $p_i$ computed using random sampling. 

We explore important properties of sequences of approximations associated with sequences of partitions that become finer. We assume that any cell in a partition has negligible boundary, i.e., $\mu_\Omega (\partial I_i)=0$ for all $i$. This holds for Voronoi tessellations corresponding to a Poisson process as well as partitions comprising balls and generalized rectangles. We define the \textbf{\textit{diameter}} of a set $A$ to be $\mathrm{diam}\, (A) = \sup_{\omega_1, \omega_2 \in A} \|\omega_1 - \omega_2\|$, with $\|\quad \|$ denoting the Euclidean norm. We consider a sequence of partitions  $\{\mathcal{I}_\ell\}_{\ell=1}^\infty=\big\{\{I^{(\ell)}_i\}_{i=1}^{M_\ell}\big\}_{\ell=1}^\infty$, where $0<M_1 < M_2 < \cdots$ is a sequence of positive integers. We define $\mathrm{diam}\,(\mathcal{I}_\ell)$ $=$ $\max_{1 \leq i \leq M_\ell}$ $\mathrm{diam}\,(I^{(\ell)}_i)$.

We have the following result.
\begin{theorem}\label{thm:funsimpapp}
	Assume the sequence  $\{\mathcal{I}_\ell\}_{\ell=1}^\infty$ of measurable disjoint partitions of $\Omega$ satisfy $\lim_{\ell \to \infty} \mathrm{diam}\,(\mathcal{I}_\ell) = 0$. Then, the family $\{\rho_{\Omega,M_\ell}\}_{\ell=1}^\infty$ is uniformly integrable and $\lim_{\ell \to \infty} \rho_{\Omega,{M_\ell}}$ $=$  $\rho_\Omega$ a.e.
\end{theorem}

\begin{proof}
	We set $\partial \mathcal{I} = \bigcup_{\ell = 1}^\infty \bigcup_{i=1}^{M_\ell} \partial 	I^{(\ell)}_i $. If $\omega \in \Omega\backslash \partial \mathcal{I}$, the Lebesgue Differentiation Theorem implies that $\rho_{\Omega,{M_\ell}}(\omega)\to \rho_\Omega(\omega)$  as $\ell \to \infty$. The approximation result follows immediately since  $\mu_\Omega(\partial \mathcal{I}) = 0$.
	
	We next show that $\{\rho_{\Omega,M_\ell}\}_{\ell=1}^\infty$ is uniformly integrable. For any $\epsilon>0$, there exists an $\eta$ such that $\int_{\{\rho_\Omega>\eta \}}\rho_\Omega\, d\mu_\Omega <\epsilon$. Let $\delta = \int_{\{\rho_\Omega>\eta\}}\, d\mu_\Omega=\mu_\Omega(\{\rho_\Omega>\eta\})$. 
	Since $\rho_{\Omega,M_\ell}$ is a simple function, there is a $\eta^\prime$ such that $\mu_\Omega(\{\rho_{\Omega,M_\ell}\leq \eta^\prime\})\geq \delta$ and $\mu_\Omega(\{\rho_{\Omega,M_\ell}> \eta^\prime\})< \delta$. We define $A_{M_\ell}=\{\rho_{\Omega,M_\ell}> \eta^\prime\}\cup B_{M_\ell}$, where $B_{M_\ell}$ is any measurable set satisfying $B_{M_\ell} \subset \{\rho_{\Omega,M_\ell}= \eta^\prime\}$ and $\mu_\Omega(B_{M_\ell})\leq \delta-\mu_\Omega(\{\rho_{\Omega,M_\ell}> \eta^\prime\})$.
	
	For all $C\in \mathcal{B}_\Omega$ such that $\mu_\Omega(C)\leq\delta$, 
	\begin{equation} \label{eq:ineq1}
		\int_C \rho_{\Omega,M_\ell}\, d\mu_{\Omega} \leq \int_{A_{M_\ell}} \rho_{\Omega,M_\ell}\, d\mu_{\Omega}\leq \int_{\{\rho_{\Omega,M_\ell}> \eta^\prime\}} \rho_{\Omega,M_\ell}\, d\mu_{\Omega}+\mu_\Omega(B_{M_\ell})\, \eta^\prime.
	\end{equation}
	
	By definition, we have
	\[
	\int_{\{\rho_{\Omega,M_\ell}> \eta^\prime\}} \rho_{\Omega,M_\ell}\, d\mu_\Omega=\int_{\{\rho_{\Omega,M_\ell}> \eta^\prime\}} \rho_{\Omega}\, d\mu_{\Omega},
	\;
	\int_{\{\rho_{\Omega,M_\ell}= \eta^\prime\}} \rho_\Omega\, d\mu_\Omega=\mu_\Omega(\{\rho_{\Omega,M_\ell}= \eta^\prime\})\, \eta^\prime.
	\] 
	There is an $B_{M_\ell}^\prime \subset \{\rho_{\Omega,M_\ell}= \eta^\prime\}$, such that $\mu_\Omega( B_{M_\ell}^\prime) = \mu_\Omega( B_{M_\ell})$ and 
	$
	\int_{ B_{M_\ell}^\prime} \rho_\Omega \, d\mu_\Omega\geq\mu_\Omega(B_{M_\ell}) \, \eta^\prime$. So~\eqref{eq:ineq1} implies
	\[
	\int_{\{\rho_{\Omega,M_\ell}> \eta^\prime\}} \rho_{\Omega,M_\ell}\, d\mu_\Omega+\mu_\Omega(B_{M_\ell})\eta^\prime
	\leq 
	\int_{\{\rho_{\Omega,M_\ell}> \eta^\prime\}} \rho_\Omega\, d\mu_\Omega +  \int_{B_{M_\ell}^\prime} \rho_\Omega \, d\mu_\Omega . 
	\]
	
	Since $\mu_\Omega(\{\rho_{\Omega,M_\ell}> \eta^\prime\}\cup B_{M_\ell}^\prime)=\delta=\mu_\Omega(\{\rho_\Omega>\eta\})$, we have
	\[
	\int_{\{\rho_{\Omega,M_\ell}> \eta^\prime\}} \rho_\Omega\, d\mu_\Omega +  \int_{B_{M_\ell}^\prime} \rho_\Omega\, d\mu_\Omega \leq \int_{\{\rho_\Omega>\eta\}}\rho_\Omega\, d\mu_\Omega<\epsilon
	\]
\end{proof}

\subsection{Proof of Theorem~\ref{thm:disap}}\label{sec:disapproof}

\begin{proof}
	To simplify notation, we prove the result in the case of the uniform prior and drop  superscripts $(\ell)$ and $(\jmath)$ and subscripts $\ell$ and $\jmath$ indicating dependency on the index of $\{M_\ell\}$ and $\{N_\jmath\}$ where possible without introducing confusion. We build the approximation in stages. We begin by discussing the approximation of densities. 
	
	Since $\mu_\mathcal{D}$ and $\tilde\mu_{\mathcal{D}}$ are equivalent with $d\tilde\mu_\mathcal{D}=\tilde\rho_\mathcal{D}d\mu_{\mathcal{D}}$, $P_\mathcal{D}$ has a density $\rho_{\mathcal{D}}^\prime=\frac{\rho_\mathcal{D}}{\tilde\rho_\mathcal{D}}$ with respect to $\tilde\mu_\mathcal{D}$. We define the simple function,
	\begin{equation*}
		\rho^\prime_{\mathcal{D},M_\ell}=\sum_{i=1}^{M_\ell} p_i^\prime \chi_{(I_i)}(y),\quad
		p_i^\prime=\frac{P_\mathcal{D}(I_i)}{\tilde\mu_\mathcal{D}(I_i)}=\frac{\int_{I_i}\frac{\rho_\mathcal{D}}{\tilde\rho_\mathcal{D}}d\tilde \mu_\mathcal{D}}{\int_{I_i}d\tilde\mu_\mathcal{D}},
	\end{equation*} 
	as an approximation of $\tilde \rho_{\mathcal{D}}$.
	We also define,
	\begin{equation*}
		{\rho}_{\Lambda,M_\ell}=\sum_{i=1}^{M_\ell} p_i^\prime \chi_{Q^{-1}(I_i)}(\lambda),\quad
		p_i^\prime=\frac{\int_{I_i}\frac{\rho_\mathcal{D}}{\tilde\rho_\mathcal{D}}\,d\tilde \mu_\mathcal{D}}{\int_{I_i}\,d\tilde\mu_\mathcal{D}}=\frac{P_\mathcal{D}(I_i)}{\tilde\mu_\mathcal{D}(I_i)}= \frac{P_\mathcal{D}(I_i)}{\mu_\Lambda(Q^{-1}(I_i))},
	\end{equation*} 
	as an approximation of ${\rho}_\Lambda$. By definition, ${\rho}_{\Lambda,M_\ell}(\lambda)=\rho^\prime_{\mathcal{D},M_\ell}(Q(\lambda))$.   By Theorem~\ref{thm:funsimpapp}, $\rho^\prime_{\mathcal{D},M_\ell}$ converges to $\rho_{\mathcal{D}}^\prime=\frac{\rho_\mathcal{D}}{\tilde\rho_\mathcal{D}}$, $\tilde\mu_{\mathcal{D}}$-a.e., so ${\rho}_{\Lambda,M_\ell}\to {\rho}_\Lambda$, $\mu_{\Lambda}$-a.e., where ${\rho}_\Lambda(\lambda)=\frac{\rho_\mathcal{D}(Q(\lambda))}{\tilde\rho_\mathcal{D}(Q(\lambda))}$ is the density for the uniform prior posterior. 
	
	Next, we construct the simple function approximations for the inner integral in the disintegration. The approximation of ${\rho}_{{\Lambda},M_\ell}$ is defined,
	\begin{equation*}
		{\rho}_{\Lambda,M_\ell,N_\jmath}=\sum_{j=1}^{N_\jmath} p_j \chi_{B_j}(\lambda),\quad
		p_j=\frac{P_\mathcal{D}(I_i)}{\sum_{Q(\lambda_k)\in I_i}\mu_\Lambda(B_k)} \ \text{if} \ \lambda_j\in Q^{-1}(I_i).
	\end{equation*} 
	Theorem~\ref{lem:SLLNS} implies
	$
	\sum_{k:Q(\lambda_k)\in I_i}\mu_\Lambda(B_k) \to \mu_\Lambda(Q^{-1}(I_i))$. 
	Theorem~\ref{thm:funsimpapp} implies  ${\rho}_{\Lambda,M_\ell,N_\jmath}(\lambda)$ $\to$ $ {\rho}_{{\Lambda},M_\ell}(\lambda)$ as $\jmath \to \infty$ for almost all $\lambda \in \textrm{int}\, (Q^{-1}(I_i))$ for some $I_i\in \mathcal{I}_{M_\ell}$ and therefore ${\rho}_{\Lambda,M_\ell,N_\jmath} \to {\rho}_{{\Lambda},M_\ell}$ a.e. as $\jmath\to \infty$.
	
	Actually, the convergence of ${\rho}_{\Lambda,M_\ell,N_\jmath}$ to  ${\rho}_{{\Lambda},M_\ell}$ is  uniform. Given $\epsilon > 0$, for every $I_i \in \mathcal{I}_{M_\ell}$, we choose $M_\jmath$ sufficiently large to guarantee
	\[
	\mu_{\Lambda}\bigg(\bigg(\bigcup_{\lambda_j\in Q^{-1}(I_i)} B_j\bigg) \triangle Q^{-1}(I_i)\bigg)<\epsilon.
	\]
	Thus, we can ensure that  $|p_j-p_i^\prime|<\epsilon$ for any $j$ with $\lambda_j\in Q^{-1}(I_i)$. The assumptions imply that  
	\[
	|{\rho}_{\Lambda,M_\ell,N_\jmath}(\lambda)-{\rho}_{\Lambda,M_\ell}(\lambda)|<\epsilon, \; \textrm{for} \; \lambda\in \Lambda^\prime\subset \Lambda \; \textrm{with} \; \mu_{{\Lambda}}(\Lambda\setminus \Lambda^\prime)<\epsilon.
	\] 
	Consequently, $|\max {\rho}_{\Lambda,M_\ell,N_\jmath}-\max {\rho}_{\Lambda, M_\ell}|<\epsilon$. Once we fix a resolution for the approximation of $\rho_{\mathcal{D}}$ by choosing $M_\ell$, we have to choose a sufficient number $N_\jmath$ of sample points in $\Lambda$ to obtain the maximum possible accuracy. The intuition is that we have to choose  sufficient samples in $\Lambda$ to accurately represent the geometry of the generalized contours, which of course are positioned in the higher dimensional space $\Lambda$.
	
	Next, we turn to the approximation result for ${P}_\Lambda$. Theorem~\ref{thm:funsimpapp} implies that the collection of densities $\{{\rho}_{\Lambda,M_\ell}\}_{\ell=1}^\infty$ is uniformly integrable. For each $M_\ell$, we show $\{{\rho}_{{\Lambda},M_\ell,N_\jmath}\}_{\jmath=1}^\infty$ is uniformly integrable. 	Given $\epsilon>0$, there exists $\delta>0$ such that $\int_C {\rho}_{\Lambda, M_\ell}\, d\mu_\Lambda <\epsilon$ for all $C\in \mathcal{B}_\Lambda$ satisfying $\mu_{\Lambda}(C)\leq\delta$. For $N_\jmath$ sufficiently large,  $|{\rho}_{\Lambda,M_\ell,N_\jmath}(\lambda)-{\rho}_{\Lambda, M_\ell}(\lambda)|<\epsilon/\delta$ for $\lambda\in\Lambda^\prime\subset \Lambda$ where $\mu_\Lambda(\Lambda\setminus\Lambda^\prime)<  \epsilon/(\max {\rho}_{\Lambda, M_\ell}+\epsilon/\delta)$, and $|\max {\rho}_{\Lambda,M_\ell,N_\jmath}-\max {\rho}_{\Lambda, M_\ell}|<\epsilon/\delta$. So, 
	\begin{align*}
		\int_C {\rho}_{\Lambda,M_\ell,N_\jmath}\, d\mu_{\Lambda} &=  \int_{C\cap \Lambda^\prime} {\rho}_{\Lambda,M_\ell,N_\jmath}\, d\mu_{\Lambda}+\int_{C\setminus \Lambda^\prime} {\rho}_{\Lambda,M_\ell,N_\jmath}\, d\mu_{\Lambda}\\
		&\leq\int_{C\cap \Lambda^\prime} ({\rho}_{\Lambda,M_\ell}+\epsilon/\delta)\, d\mu_{\Lambda} +\frac{\epsilon}{\max \overline{\rho}_{\Lambda, M_\ell}+\epsilon/\delta}(\max \overline{\rho}_{\Lambda, M_\ell}+\epsilon/\delta)\\ 
		&\leq \int_{C}\rho_{\Lambda, M_\ell}\, d\mu_{\Lambda}+\epsilon+\epsilon< 3\epsilon.
	\end{align*}
	for all $C\in \mathcal{B}_\Lambda$ satisfying $\mu_{\Lambda}(C)\leq\delta$. This shows the desired result. 
	
	If $A\in \mathcal{B}_\Lambda$ satisfies $\mu_\Lambda(\partial A)=0$, $A$ is approximated by $A_{N_\jmath}=\bigcup_{\lambda_j\in A}B_j$, and ${P}_\Lambda(A)$ is approximated by 
	\[
	{P}_{\Lambda, M_\ell,N_\jmath}(A_{N_\jmath})=\sum_{\lambda_j\in A} {P}_{\Lambda,M_\ell,N_\jmath}(B_j)=\int_{A_{N_\jmath}}{\rho}_{{\Lambda},M_\ell,N_\jmath}\, d\mu_{\Lambda}.
	\] 
	We prove ${P}_{\Lambda, M_\ell,N_\jmath}(A_{N_\jmath})$ converges to ${P}_\Lambda(A)$ in two steps. First, for fixed $M_\ell$, the Vitali Convergence Theorem implies
	\begin{equation} \label{nconst1}
		\int_\Lambda \big|{\rho}_{{\Lambda},M_\ell,N_\jmath}\, \chi_A-{\rho}_{\Lambda, M_\ell}\, \chi_A\big|\, d\mu_\Lambda\to 0 \; \textrm{as}\; \jmath\to \infty,
	\end{equation}
	while the uniform integrability of ${\rho}_{{\Lambda},M_\ell,N_\jmath}$ implies
	\begin{equation}\label{nconst2}
		\int_\Lambda \big|{\rho}_{{\Lambda},M_\ell,N_\jmath}\, \chi_{A_{N_\jmath}}-{\rho}_{{\Lambda},M_\ell,N_\jmath}\, \chi_A\big|\, d\mu_\Lambda=\int_{A_{N_\jmath}\triangle A} {\rho}_{{\Lambda},M_\ell,N_\jmath}\, d\mu_\Lambda\to 0  \; \textrm{as}\; \jmath\to \infty.
	\end{equation}
	Combining \eqref{nconst1} and \eqref{nconst2} shows
	\begin{equation} \label{ncon}
		\int_\Lambda \big|{\rho}_{{\Lambda},M_\ell,N_\jmath}\, \chi_{A_{N_\jmath}}-{\rho}_{\Lambda, M_\ell}\, \chi_A\big|\, d\mu_{\Lambda}\to 0 \ \textrm{as}\ \jmath\to \infty.
	\end{equation}
	Another application of the Vitali Convergence Theorem implies,
	\begin{equation}\label{mcon}
		\int_\Lambda \big|{\rho}_{{\Lambda},M_\ell}\, \chi_A-{\rho}_{\Lambda}\, \chi_A\big|\, d\mu_\Lambda\to 0 \;  \textrm{as}\; \ell\to \infty.
	\end{equation}
	Together \eqref{ncon} and \eqref{mcon} imply,
	\[
	\int_\Lambda \big|{\rho}_{{\Lambda},M_\ell,N_\jmath}\, \chi_{A_{N_\jmath}}-{\rho}_{\Lambda}\, \chi_A\big|\, d\mu_{\Lambda}\to 0 \;  \textrm{as} \; \jmath \to \infty \; \textrm{and then} \; \ell \to \infty,
	\]
	hence
	$
	\lim_{\ell\to \infty}\lim_{\jmath\to \infty}\int_{A_{N_\jmath}}{\rho}_{{\Lambda},M_\ell,N_\jmath}\, d\mu_{\Lambda}= \int_{A}{\rho}_{\Lambda}\, d\mu_{\Lambda} ={P}_\Lambda(A)$.
\end{proof}

\subsection{Proof of Theorem~\ref{thm:sampsolprop}}\label{sec:sampsolproof}
We now turn to the main estimation result.

\begin{proof}[Proof of Theorem~\ref{thm:sampsolprop}]
	$\quad$
	
	\noindent \textit{Theorem~\ref{thm:sampsolprop}: Part \ref{thm:sampsolprop1}}
	%
	
	We begin by computing,
	\begin{align}
		E\left(\frac{\sum_{k=1}^{N_\jmath} \chi_{\lambda_k}(Q^{-1}(I_{i}^{(K)})\cap A ) }{ \sum_{j=1}^{N_\jmath} \chi_{\lambda_j}(Q^{-1}(I_{i}^{(K)})) }\right) &
		= \sum_{k=1}^{N_\jmath} E\,\left(
		\frac{ \chi_{\lambda_k}(Q^{-1}(I_{i}^{(K)})\cap A )}{\sum_{j=1}^{N_\jmath} \chi_{\lambda_j}(Q^{-1}(I_{i}^{(K)}))}
		\right) \nonumber \\
		&= \sum_{k=1}^{N_\jmath} \breve{p}_i^{(K)} \sum_{j=1}^{N_\jmath} \frac{1}{j} {N_\jmath-1 \choose j-1} (p_i^{(K)})^{j-1}(1-p_i^{(K)})^{N_\jmath-j} \nonumber \\
		&= \frac{\breve{p}_i^{(K)}}{p_i^{(K)}} \sum_{j=1}^{N_\jmath} 
		\frac{N_\jmath!}{j!(N_\jmath-j)!}(p_{i}^{(K)})^{j}(1-p_{i}^{(K)})^{N_\jmath-j}  \nonumber \\
		&= \frac{\breve{p}_i^{(K)}}{p_i^{(K)}}(1 - (1-p_i^{(K)})^{N_\jmath} ). \label{l1}
	\end{align}
	Then
	\begin{align}\label{part1Result}
		E(\hat{P}_{\Lambda,K,N_\jmath}(A)) &= \sum_{i=1}^{M_K} E\, \left(\frac{\sum_{j=1}^{N_\jmath} \chi_{\lambda_j}(Q^{-1}(I_{i}^{(K)})\cap A ) }{ \sum_{j=1}^{N_\jmath} \chi_{\lambda_j}(Q^{-1}(I_{i}^{(K)})) }\cdot \frac{1}{K} \sum_{k=1}^{K} \chi_{q_k}(I_{i}^{(K)})\right) \nonumber \\
		&= \sum_{i=1}^{M_K} E\, \left(\frac{\sum_{j=1}^{N_\jmath} \chi_{\lambda_j}(Q^{-1}(I_{i}^{(K)})\cap A ) }{ \sum_{j=1}^{N_\jmath} \chi_{\lambda_j}(Q^{-1}(I_{i}^{(K)})) } \right)\, E\, \left( \frac{1}{K} \sum_{k=1}^{K} \chi_{q_k}(I_{i}^{(K)}) \right) \nonumber \\
		& =\sum_{i=1}^{M_K} \frac{\breve{p}_i^{(K)}}{ p_i^{(K)}}(1 - (1-p_i^{(K)})^{N_\jmath} ) P_{\mathcal{D}}(I_{i}^{(K)}).
	\end{align}

	We define a sequence of measurable functions $\eta_{\mathcal{D},K}(q) = \sum_{i=1}^{M_K} \frac{\breve{p}_i^{(K)}}{ p_i^{(K)}} \chi_{q}(I_{i}^{(K)}) $ on $\mathcal{D}$. Disintegration gives
	\begin{equation*}
		\frac{\breve{p}_i^{(K)}}{p_i^{(K)}} = 
		\frac{\int_{Q\big(Q^{-1}(I_{i}^{(K)})\cap A\big)} {P}_{\mathrm{p},N}(Q^{-1}(I_{i}^{(K)})\cap A|q)\, d\tilde{P}_{\mathrm{p},\mathcal{D}} }{\int_{I_{i}} \, d\tilde{P}_{\mathrm{p},\mathcal{D}} } 
		= \frac{\int_{I_{i}^{(K)}} {P}_{\mathrm{p},N}(A|q)\,  d\tilde{P}_{\mathrm{p},\mathcal{D}}}{\int_{I_{i}^{(K)}} \, d\tilde{P}_{\mathrm{p},\mathcal{D}} }.
	\end{equation*}
	The Lebesgue Differentiation Theorem implies that $\eta_{\mathcal{D},K}(q) \to {P}_{\mathrm{p},N}(A|q)$ $P_\mathcal{D}$-a.e. The collection $\{\eta_{\mathcal{D},M_K} \}$ is uniform integrable with respect to $P_{\mathcal{D}}$ so the Vitali Convergence Theorem implies,
	$$
	\int_{\mathcal{D}} |\eta_{\mathcal{D},K}(q) - {P}_{\mathrm{p},N}(A|q) |\, dP_{\mathcal{D}} \to 0 \; \mathrm{and} \; 
	\lim_{K\to \infty} \sum_{i=1}^{M_K} \frac{\breve{p}_i^{(K)}}{p_i^{(K)}} \, P_{\mathcal{D}}(I_{i}^{(K)}) = {P}_\Lambda(A).
	$$
	Combining this with \eqref{Esampsol} yields the result.
	
	\noindent \textit{Theorem~\ref{thm:sampsolprop}: Part \ref{thm:sampsolprop2}}
	
	We start with an estimate of the second moment,
	\begin{align}
		E & \left(\sum_{i=1}^{M_K}\frac{\sum_{j=1}^{N_\jmath} \chi_{\lambda_j}(Q^{-1}(I_{i}^{(K)})\cap A ) }{ \sum_{j=1}^{N_\jmath} \chi_{\lambda_j}(Q^{-1}(I_{i}^{(K)})) } \chi_{q_k}(I_{i}^{(K)})\right)^2 \nonumber \\
		&= \iint \left(\left(\sum_{i=1}^{M_K}\frac{\sum_{j=1}^{N_\jmath} \chi_{\lambda_j}(Q^{-1}(I_{i}^{(K)})\cap A ) }{ \sum_{j=1}^{N_\jmath} \chi_{\lambda_j}(Q^{-1}(I_{i}^{(K)})) } \chi_{q_k}(I_{i}^{(K)})\right)^2 \,d P_{\mathcal{D}}\right)  \, d P_{\mathrm{p}} \nonumber \\
		&= \int \sum_{i=1}^{M_K} P_{\mathcal{D}}(I_{i}^{(K)}) \left(\frac{\sum_{j=1}^{N_\jmath} \chi_{\lambda_j}(Q^{-1}(I_{i}^{(K)})\cap A ) }{ \sum_{j=1}^{N_\jmath} \chi_{\lambda_j}(Q^{-1}(I_{i}^{(K)})) } \right)^2   \, d P_{\mathrm{p}} \nonumber \\
		&= \sum_{i=1}^{M_K} P_{\mathcal{D}}(I_{i}^{(K)}) \Bigg( \sum_{j=1}^{N_\jmath} \int   \frac{ \chi_{\lambda_j}(Q^{-1}(I_{i}^{(K)})\cap A ) }{ \big(\sum_{j=1}^{N_\jmath} \chi_{\lambda_j}(Q^{-1}(I_{i}^{(K)}))\big)^{2} }    \, d P_{\mathrm{p}} \nonumber \\
		& \qquad \qquad \qquad \qquad +\sum_{j\neq k}^{N_\jmath} \int \frac{\chi_{\lambda_j^{(\jmath)},\lambda_k^{(\jmath)}}(Q^{-1}(I_{i}^{(K)})\cap A ) }{\big(\sum_{j=1}^{N_\jmath} \chi_{\lambda_j}(Q^{-1}(I_{i}^{(K)}))\big)^{2}}  \, d P_{\mathrm{p}}   \Bigg) \nonumber \\
		&= \sum_{i=1}^{M_K} P_{\mathcal{D}}(I_{i}^{(K)}) \big(N_\jmath A_{1} +N_\jmath(N_\jmath-1)A_{2}\big). \label{a1a2}
	\end{align}
	
	Following an argument similar to the computation leading to \eqref{part1Result}, we first estimate,
	\begin{align}
		A_{1} &= \sum_{j=1}^{N_\jmath} \frac{1}{j^2} \breve{p}_i^{(K)} \binom{N_\jmath-1}{j-1} (p_i^{(K)})^{j-1}(1-p_i^{(K)})^{N_\jmath-j} \nonumber \\
		&\leq \sum_{j=1}^{N_\jmath}  \frac{2}{j(j+1)}\breve{p}_i^{(K)} \binom{N_\jmath-1}{j-1} (p_i^{(K)})^{j-1}(1-p_i^{(K)})^{N_\jmath-j} \nonumber \\
		&=\frac{2\breve{p}_i^{(K)}}{(N_\jmath+1)N_\jmath (p_i^{(K)})^{2}} \sum_{j=1}^{N_\jmath} \binom{N_\jmath+1}{j+1} (p_i^{(K)})^{j+1}(1-p_i^{(K)})^{N_\jmath-j} \nonumber \\
		&= \frac{2\breve{p}_i^{(K)}}{(N_\jmath+1)N_\jmath (p_i^{(K)})^{2}} \big(1 - (1-p_i^{(K)})^{N_\jmath}(1+N_\jmath p_i^{(K)}) 
		\big).\label{a1a2:a1}
	\end{align}
	
	For $A_{2}$,
	\begin{align}
		A_2 &=  \sum_{j=2}^{N_\jmath} \frac{1}{j^2} (\breve{p}_i^{(K)})^{2} \binom{N_\jmath-2}{j-2} (p_i^{(K)})^{j-2}(1-p_i^{(K)})^{N_\jmath-j} \nonumber \\
		&\leq (\breve{p}_i^{(K)})^2 \sum_{j=2}^{N_\jmath} \frac{1}{j(j-1)} \binom{N_\jmath-2}{j-2}(p_i^{(K)})^{j-2}(1-p_i^{(K)})^{N_\jmath-j} \nonumber \\
		&= \frac{(\breve{p}_i^{(K)})^2}{N_\jmath(N_\jmath-1)(p_i^{(K)})^{2}}\big(
		1- (1-p_i^{(K)})^{N_\jmath-1}(1+(N_\jmath-1) p_i^{(K)}) \big). \label{a1a2:a2}
	\end{align}
	
	Combining \eqref{Esampsol}, \eqref{a1a2}, \eqref{a1a2:a1}, and \eqref{a1a2:a2} yields,
	\begin{align}
		\mathrm{Var}\,\big(\hat{P}_{\Lambda,M_K,N_\jmath}(A) \big) &= \frac{1}{K^2} \sum_{k=1}^{K} \mathrm{Var}\, \Big(\sum_{i=1}^{M_K}\frac{\sum_{j=1}^{N_\jmath} \chi_{\lambda_j}(I_{i}^{(K)}\cap A ) }{ \sum_{j=1}^{N_\jmath} \chi_{\lambda_j}(Q^{-1}(I_{i}^{(K)})) } \chi_{q_k}(I_{i}^{(K)})\Big) \nonumber \\
		& \leq  \frac{1}{K} \Bigg( \sum_{i=1}^{M_K}  P_{\mathcal{D}}(I_{i}^{(K)}) \Bigg(\frac{2\breve{p}_i^{(K)}}{(N_\jmath+1)(p_i^{(K)})^{2}} \big(1 - (1-p_i^{(K)})^{N_\jmath}(1+N_\jmath p_i^{(K)})) \nonumber \\
		&\qquad \qquad + \frac{(\breve{p}_i^{(K)})^2}{(p_i^{(K)})^{2}}(
		1- (1-p_i^{(K)})^{N_\jmath-1}(1+(N_\jmath-1)p_i^{(K)}) ) \Bigg) \nonumber \\
		&\qquad \qquad \qquad- \Bigg(\sum_{i=1}^{M_K} \frac{\breve{p}_i^{(K)}}{ p_i^{(K)}}(1 - (1-p_i^{(K)})^{N_\jmath} ) P_{\mathcal{D}}(I_{i}^{(K)}) \Bigg)^2
		\Bigg). \label{varub}
	\end{align}
	This implies \eqref{Varsampsol} and the convergence result.
	
	\noindent \textit{Theorem~\ref{thm:sampsolprop}: Part \ref{thm:sampsolprop3}}
	
	By the Law of Large Numbers, 
	$$
	\lim_{\jmath \to \infty}\frac{\sum_{j=1}^{N_\jmath} \chi_{\lambda_j}(Q^{-1}(I_{i}^{(K)})\cap A ) }{ \sum_{j=1}^{N_\jmath} \chi_{\lambda_j}(Q^{-1}(I_{i}^{(K)})) } = \frac{P_{\mathrm{p}}(Q^{-1}(I_{i}^{(K)})\cap A)}{P_{\mathrm{p}}(Q^{-1}(I_{i}^{(K)}))} \quad 
	\mathrm{a.e.}
	$$
	We write,
	\begin{align}
		\lim_{\jmath\to \infty}& \hat{P}_{\Lambda,M_K,N_\jmath}(A) = \sum_{i=1}^{M_K} \frac{\breve{p}_i^{(K)}}{p_i^{(K)}} \cdot \frac{1}{K} \sum_{k=1}^{K}\chi_{q_k}(I_{i}^{(K)}) \nonumber \\
		=& \sum_{i=1}^{M_K} \frac{\breve{p}_i^{(K)}}{p_i^{(K)}}\left( 
		\frac{1}{K} \sum_{k=1}^{K} \chi_{q_k}(I_{i}^{(K)}) -P_{\mathcal{D}}(I_{i}^{(K)}) \right) 
		+ \sum_{i=1}^{M_K} \frac{\breve{p}_i^{(K)}}{p_i^{(K)}} P_{\mathcal{D}}(I_{i}^{(K)}). \label{cons2}
	\end{align}
	\eqref{Esampsol} implies that
	$
	\sum_{i=1}^{M_K} \frac{\breve{p}_i^{(K)}}{p_i^{(K)}} P_{\mathcal{D}}(I_{i}^{(K)}) \to {P}_\Lambda(A)$  a.e. We set,
	\begin{align}
		\sum_{i=1}^{M_K} & \frac{\breve{p}_i^{(K)}}{p_i^{(K)}}\left( 
		\frac{1}{K} \sum_{k=1}^{K} \chi_{q_k}(I_{i}^{(K)}) -P_{\mathcal{D}}(I_{i}^{(K)}) \right)\nonumber\\
		& =
		\frac{1}{K}\sum_{k=1}^{K} \left( \sum_{i=1}^{M_K}	\frac{\breve{p}_i^{(K)}}{p_i^{(K)}}(\chi_{q_k}(I_{i}^{(K)}) - P_{\mathcal{D}}(I_{i}^{(K)})) \right)  = \frac{1}{K}\sum_{k=1}^{K} X^{(K)}_k .\nonumber
	\end{align}
	Since $\{X^{(K)}_k\}$ is bounded, Hoeffding's inequality implies that for $\epsilon > 0$,
	$$
	\sum_{K=1}^{\infty} P_{\mathcal{D}}\Big(\Big\vert \frac{1}{K} \sum_{k=1}^{K} X^{(K)}_k  \Big\vert \geq \epsilon\Big) < \infty .
	$$
	The First Borel-Cantelli Lemma implies $\lim_{K\to \infty} \frac{1}{K}\sum_{k=1}^{K} X^{(K)}_k= 0$ a.s.
	
\end{proof}

\subsection{Proof of  results in \S~\ref{sec:asymperror}}\label{sec:asymperrorproof}

\begin{proof}[Proof of Theorem~\ref{CLTerror2}]
	
	We define the triangular array,
	\begin{equation*}
		X_k^{(K)} = 
		\sum_{i=1}^{M_K} \frac{\breve{p}_i^{(K)}}{p_i^{(K)}} \Big(  \chi_{q_k}(I_{i}^{(K)})-	P_{\mathcal{D}}(I_{i}^{(K)})\Big).
	\end{equation*}
	We set $S_{K} = \sum_{k=1}^{K} X_k^{(K)}$ and $s_{K}^2  = \mathrm{Var}\,(S_{K}) = K \mathrm{Var}\,(X_k^{(K)})$. We have
	\[
	s_{K}^2  = K\, \left(
	\sum_{i=1}^{M_K} \left(\frac{\breve{p}_i^{(K)}}{p_i^{(K)}}\right)^2   P_{\mathcal{D}}(I_{i}^{(K)})\; -	\;
	\sum_{i=1}^{M_K} \left(\frac{\breve{p}_i^{(K)}}{p_i^{(K)}}   P_{\mathcal{D}}(I_{i}^{(K)})\right)^2
	\right)
	\]
	so
	$
	\lim_{K \to \infty} \frac{s_{K}^2}{K} = \sigma_{\mathrm{p}}^2 $. Markov's inequality implies that for $\epsilon > 0$, $P_{\mathcal{D}}\big(|X_k^{(K)} - E(X_k^{(K)})| > \epsilon s_{K} \big) < \frac{1}{\epsilon K}$. So, Lindeberg's condition,
	$$
	\lim_{K\to \infty} \frac{1}{s_{K}^2} \, \sum_{k=1}^{K} E \Big(
	\big(X_k^{(K)} - E(X_k^{(K)})\big)^2 \chi_{\{ \vert X_k^{(K)} - E(X_k^{(K)})\vert > \epsilon s_{K} \}}
	\Big) = 0,
	$$
	is satisfied. The Central Limit Theorem implies
	$
	\frac{1}{s_K} \,\sum_{k=1}^K X_k^{(K)}  \stackrel{d}{\to } \mathcal{N}(0,1)$.
\end{proof}

\begin{proof}[Proof of Theorem~\ref{CLTerror3}]
	
	From the proof of Theorem~\ref{thm:sampsolprop}, we know that
	\[
	E(T_{3,i,\jmath}) = -\frac{\breve{p}_i^{(K)}}{ p_i^{(K)}}\big( (1-p_i^{(K)})^{N_\jmath} \big) ,
	\]
	and
	\begin{align*}
		\mathrm{Var}\,(T_{3,i,\jmath}) & =   \frac{2\breve{p}_i^{(K)}}{(N_\jmath+1)(p_i^{(K)})^{2}} \big(1 - (1-p_i^{(K)})^{N_\jmath}(1+N_\jmath p_i^{(K)})) \\
		& \qquad + \left(\frac{(\breve{p}_i^{(K)})}{(p_i^{(K)})}\right)^2 \, (
		1- (1-p_i^{(K)})^{N_\jmath-1}(1+(N_\jmath-1)p_i^{(K)}) )\\
		&\qquad \qquad - \left(\frac{(\breve{p}_i^{(K)})}{(p_i^{(K)})}\right)^2 \big(1 - (1-p_i^{(K)})^{N_\jmath} \big)^2.
	\end{align*}
	The result follows.
\end{proof}
Both $E(T_{3,i,\jmath})$ and $\mathrm{Var}\,(T_{3,i,\jmath})$ become small as $\jmath \to \infty$ because $(1-p_i^{(K)})^{N_\jmath}$ becomes small.  However, $p_i^{(K)}$ approaches $0$.  To bound $(1-p_i^{(K)})^{N_\jmath} \leq \epsilon$ for $\epsilon > 0$,  we need to choose $N_\jmath$ so
\[
N_\jmath \geq \frac{\log(\epsilon)}{\log(1-p_i^{(K)})} \approx -\frac{\log(\epsilon)}{p_i^{(K)}}.
\]

\subsection{Proof of Theorem~\ref{thm:gensampsolprop}}\label{sec:gensampsolproof}

\begin{proof}
	As shown in the proof of \textit{Theorem~\ref{thm:sampsolprop}: Part \ref{thm:sampsolprop3}}, it suffices to show that
	\begin{equation}\label{gensampconverg}
		\sum_{i=1}^{M_K} \frac{\breve{p}_i^{(K)}}{p_i^{(K)}}\left( 
		\hat{P}_{\mathcal{D}}(I_{i}^{(K)})	  -P_{\mathcal{D}}(I_{i}^{(K)}) \right) \to 0 \quad \textrm{a.s.}
	\end{equation}
	We estimate,
	\begin{align*}
		\left|	\sum_{i=1}^{M_K} \frac{\breve{p}_i^{(K)}}{p_i^{(K)}}\left( 
		\hat{P}_{\mathcal{D}}(I_{i}^{(K)})	  -P_{\mathcal{D}}(I_{i}^{(K)}) \right) \right|
		& = \left| \sum_{i=1}^{M_K} \frac{\breve{p}_i^{(K)}}{p_i^{(K)}}\left( 
		\int_{I_{i}^{(K)}} \big(\hat{\rho}_{\mathcal{D}}(q)-	{\rho}_{\mathcal{D}}(q)\big) \, d\mu_{\mathcal{D}}(q)\right)\right|\\
		& \leq  \sum_{i=1}^{M_K} 
		\int_{I_{i}^{(K)}} \big| \hat{\rho}_{\mathcal{D}}(q)-	{\rho}_{\mathcal{D}}(q)\big| \, d\mu_{\mathcal{D}}(q)\\
		& = \int_{\mathcal{D}} \big| \hat{\rho}_{\mathcal{D}}(q)-	{\rho}_{\mathcal{D}}(q)\big| \, d\mu_{\mathcal{D}}(q).
	\end{align*}
	The result follows by assumption.
	
\end{proof}

\section{Details for the dropping ball experiment}\label{droppingballsextra}

\subsection{Experimental details}
Relevant physical parameters of the balls obtained from \cite{hyperphysics} are listed in Table~\ref{tab:ballchar}. Additionally, the temperature was about $65^\circ F$. The data from the experiment are listed in Table~\ref{tab:balldata}.
\begin{table}[h]
	\centering
	\begin{tabular}{|l|c|c|c|c|}
		\hline
		type & mass & diameter &Coefficient of Drag ($\mathrm{C}_{\mathrm{d}}$)& Terminal Velocity \\
		\hline
		golf ball & 46g & 42.7mm& .45 & 25 m/s \\
		\hline
		tennis ball	& 57g & 65mm & .6 & 20 m/s\\
		\hline 
		baseball & 142g & 70.8mm & .3 & 40 m/s \\
		\hline
		volleyball & 260g & 210mm & .12-.54 & 15.9 m/s\\
		\hline
		basketball & 591g & 240mm & .54 & 20.2 m/s\\
		\hline
		bowling ball & 7.26kg & 218mm & .47 & 83.1 m/s \\
		\hline
	\end{tabular}
	\caption{Physical parameters of the balls in the experiment \cite{hyperphysics}.}
	\label{tab:ballchar}
\end{table}
\begin{table}[h]
	\centering
	\begin{tabular}{|r|l|c|c|c|c|c|}
		\hline
		&baseball & basketball & volleyball & bowling ball & golf ball & tennis ball \\
		\hline
		&2.8367 & 2.9033 & 2.6033 & 2.7383 & 2.7700 & 3.0367 \\
		&2.8383 & 3.0050 & 3.0700 & 2.7717 & 2.8367 & 3.0717 \\
		&       & 2.8383 & 3.1383 & 2.7367 & &  \\
		&       & 2.9033 & & & &  \\
		&       & 2.8700 & & & &  \\
		\hline
		\hline
		average	&2.8375 & 2.9040 & 2.9372 & 2.7489 & 2.8034 &  3.0542\\
		\hline
	\end{tabular}
	\caption{Flight times in seconds of the balls in the experiment. Last row shows averages.}
	\label{tab:balldata}
\end{table}

\section{The algorithm for the accept-reject estimator}\label{sec:rejectalg}

The accept-reject estimator algorithm is given in Algorithm~\ref{alg:accept-reject}.

\normalem
\begin{algorithm}[htb]
	\caption{Accept-Reject Algorithm}\label{alg:accept-reject}
	\begin{algorithmic}[1]
		\STATE \textbf{Initialization:} \;
		
		\STATE $\qquad$ Observed data $\{q_k\}_{k=1}^K$ \;
		
		\STATE $\qquad$ Partition $\{I_i\}_{i=1}^M$ of $\mathcal{D}$\;
		
		\STATE \textbf{Estimation:} \;
		
		\STATE $\qquad$ Generate independent samples $\{\lambda_i\}_{i=1}^N \in \Lambda$ from $P_{\mathrm{p}}$ and compute $\{Q(\lambda_i)\}_{i=1}^N$\;
		
		\STATE $\qquad$ Construct $\hat{\rho}_{\mathcal{D}}$ and $\widehat{\tilde{\rho}}_{\mathrm{p},\mathcal{D}}$ and set  $C = \max\Lambda \frac{\hat{\rho}_{\mathcal{D}}}{\widehat{\tilde{\rho}}_{\mathrm{p},\mathcal{D}}}$\;

		\FOR{$j = 1$ to $N$}
		
		\STATE	Generate a random sample $\xi$ from the uniform distribution on $[0,1]$\;
		
		\STATE	\IF{$\xi > \frac{\hat{\rho}_{\mathcal{D}}(Q(\lambda_j))}{
				\widehat{\tilde{\rho}}_{\mathrm{p},\mathcal{D}}(Q(\lambda_j))C}$}
		
		\STATE remove $\lambda_j$ from the collection	
		\ENDIF 
		\ENDFOR
		
		\STATE Return independent samples $\{\breve{\lambda}_\ell\}_{\ell=1}^{\breve{N}} \subset \{\lambda_j\}_{j=1}^N$\;
	\end{algorithmic}
	
\end{algorithm}

\end{document}